\documentclass[11pt]{article}

\usepackage{fullpage}
\usepackage{amsmath,amsfonts,amsthm,mathrsfs,mathpazo,xspace,graphicx,
float,bm}
\usepackage[colorlinks]{hyperref}
\usepackage{endnotes}
\usepackage[dvipsnames]{xcolor}
\usepackage{bm}
\usepackage{times}
\usepackage{enumitem}
\usepackage{amssymb,latexsym}
\usepackage{cancel}

\newtheorem{theorem}{Theorem}[section]
\newtheorem*{theorem*}{Theorem}

\newtheorem{lemma}[theorem]{Lemma}
\newtheorem{claim}[theorem]{Claim}

\newtheorem{corollary}[theorem]{Corollary}
\newtheorem{conj*}[theorem]{Conjecture}

\theoremstyle{remark}
\newtheorem{remark}[theorem]{Remark}

\theoremstyle{definition}
\newtheorem{definition}[theorem]{Definition}
\newtheorem{example}[theorem]{Example}

\newcommand{\beq}{\begin{eqnarray}}
\newcommand{\eeq}{\end{eqnarray}}

\newcommand{\ket}[1]{|#1\rangle}
\newcommand{\bra}[1]{\langle#1|}
\newcommand{\braket}[2]{\langle #1 | #2 \rangle}
\newcommand{\proj}[1]{\ket{#1}\!\bra{#1}}

\newcommand{\tr}{\mbox{\rm tr}}
\newcommand{\Id}{\ensuremath{\mathop{\rm Id}\nolimits}}
\newcommand{\Es}[1]{\ensuremath{\mathop{\textsc{E}}_{#1}}}

\newcommand{\pild}{\pi_{\rm ld}}

\newcommand{\reg}[1]{{\textsf{#1}}}

\newcommand{\C}{\ensuremath{\mathbb{C}}}

\newcommand{\F}{\ensuremath{\mathbb{F}}}

\newcommand{\R}{\ensuremath{\mathbb{R}}}
\newcommand{\Z}{\ensuremath{\mathbb{Z}}}
\newcommand{\Fp}{ \ensuremath{{\mathbb{F}}}_p}
\newcommand{\Fq}{\ensuremath{\mathbb{F}_q}}
\newcommand{\Zp}{\ensuremath{\mathbb{Z}_p}}

\newcommand{\mC}{\mathcal{C}}
\newcommand{\mX}{\mathcal{X}}
\newcommand{\mA}{\mathcal{A}}

\newcommand{\cA}{\mathcal{A}}

\newcommand{\cM}{\mathcal{M}}
\newcommand{\cX}{\mathcal{X}}

\newcommand{\bij}{\pi}
\newcommand{\qp}{\tau}
\newcommand{\mH}{\mathcal{H}}

\newcommand{\setft}[1]{\mathrm{#1}}

\newcommand{\Unitary}{\setft{U}}

\newcommand{\Lin}{\setft{L}}

\DeclareMathOperator{\poly}{poly}

\newcommand{\eps}{\varepsilon}

\newcommand{\epr}{\textsc{EPR}}
\newcommand{\aux}{\textsc{aux}}
\newcommand{\MS}{\textsc{MS}}
\newcommand{\COM}{\textsc{COM}}

\newcommand{\ot}{\otimes}
\newcommand{\ol}[1]{\overline{#1}}
\newcommand{\ul}[1]{\underline{#1}}
\newcommand{\comp}[1]{{\color{MidnightBlue} \bm{#1}}}

\newcommand{\hS}{\hat{S}}

\newcommand{\al}{\alpha}
\newcommand{\be}{\beta}

\newcommand{\NEXP}{\textsc{NEXP}}
\newcommand{\MIP}{\textsc{MIP}}
\newcommand{\QMA}{\textsc{QMA}}
\newcommand{\NP}{\textsc{NP}}

\newcommand{\cld}{\textsc{c-lowdeg}}
\newcommand{\qld}{\textsc{q-lowdeg}}
\newcommand{\eval}{\textsc{eval}}
\newcommand{\code}{\textsc{code-check}}
\newcommand{\energy}{\textsc{energy}}
\newcommand{\XZ}{\textsc{XZ}}

\newcommand{\sumgame}{\textsc{sum}}
\newcommand{\lin}{\textsc{lin}}

\DeclareMathOperator{\ima}{Im}

\newcommand{\thmref}[1]{Theorem~\ref{thm:#1}}
\newcommand{\lemref}[1]{Lemma~\ref{lem:#1}}
\newcommand{\figref}[1]{Figure~\ref{fig:#1}}

\newif\ifnotes\notesfalse

\ifnotes
\usepackage{color}
\definecolor{mygrey}{gray}{0.50}
\newcommand{\notename}[2]{{\textcolor{mygrey}{\footnotesize{\bf (#1:} {#2}{\bf ) }}}}
\newcommand{\noteswarning}{{\begin{center} {\Large WARNING: NOTES ON}\endnote{Warning: notes on}\end{center}}}
\newcommand{\notesendofpaper}{{\theendnotes}}
\newcommand{\pnote}[1]{{\endnote{#1}}}

\newcommand{\tnote}[1]{\textcolor{magenta}{\small {\textbf{(Thomas:} #1\textbf{)
      }}}}
\newcommand{\anote}[1]{\textcolor{red}{\small {\textbf{(Anand:} #1\textbf{) }}}}

\else

\newcommand{\notename}[2]{{}}
\newcommand{\noteswarning}{{}}
\newcommand{\notesendofpaper}{}
\newcommand{\pnote}[1]{}
\newcommand{\anote}[1]{}

\newcommand{\tnote}[1]{}

\fi

\begin{document}

\title{Low-degree testing for quantum states,\\[1mm] and a quantum entangled games PCP for QMA}
\author{Anand Natarajan\thanks{Center for Theoretical Physics,
    MIT, Cambridge, USA. email:\texttt{anandn@mit.edu}.} \qquad Thomas
  Vidick\thanks{Department of Computing and Mathematical Sciences,
    California Institute of Technology, Pasadena, USA. email:
    \texttt{vidick@cms.caltech.edu}.}}
\date{}
\maketitle

\begin{abstract}
We show that given an explicit description of a multiplayer game, with a classical verifier and a constant number of players, it is $\QMA$-hard, under randomized reductions, to distinguish between the cases when the players have a strategy using entanglement that succeeds with probability $1$ in the game, or when no such strategy succeeds with probability larger than $\frac{1}{2}$. This proves the ``games quantum  PCP conjecture'' of Fitzsimons and the second author (ITCS'15), albeit under randomized reductions. 

The core component in our reduction is a construction of a family of two-player
games for testing $n$-qubit maximally entangled states. For any integer $n\geq
2$, we give such a game in which questions from the verifier are $O(\log n)$
bits long, and answers are $\poly(\log\log n)$ bits long. We show that for any
constant $\eps\geq 0$, any strategy that succeeds with probability at least $1-\eps$ in
the test must use a state that is within distance $\delta(\eps) =
O(\eps^c)$ from a state that is locally equivalent to a
maximally entangled state on $n$ qubits, for some universal constant
$c>0$. The construction is based on the classical plane-vs-point test for multivariate low-degree polynomials of Raz and
Safra (STOC'97). We extend the classical test to the quantum regime by executing independent copies of the test in the generalized Pauli $X$ and $Z$ bases over $\F_q$, where $q$ is a sufficiently large prime power, and combine the two through a test for the Pauli twisted commutation relations. 

Our main complexity-theoretic result is obtained by combining this family of games with techniques from the classical PCP literature. More specifically, we use constructions of PCPs of proximity introduced by Ben-Sasson et al. (CCC'05), and crucially rely on a linear property of such PCPs. Another consequence of our results is a deterministic reduction from the
 games quantum PCP conjecture to a suitable formulation of the constraint satisfaction quantum
PCP conjecture. 

\end{abstract}

\noteswarning

\vfill
\thispagestyle{empty}
\pagebreak
\setcounter{page}{1}

\section{Introduction}

The PCP theorem~\cite{AroLunMotSudSze98JACM,AroSaf98JACM} makes a remarkable statement: any language that admits efficiently verifiable proofs of membership, i.e. any problem in $\NP$, also admits proofs that can be verified by reading only a \emph{constant} number of bits of the proof. Do similar encodings exist for problems that admit \emph{quantum} proofs? Consider the local Hamiltonian problem. Is there a way to encode a witness for the minimal energy of a Hamiltonian in a way that the energy can be estimated to within inverse polynomial accuracy while accessing only a constant number of bits, or qubits, from the witness? 
The pursuit of this question, which, broadly speaking, asks for quantum extensions of the PCP theorem, has been one of the most fruitful and challenging problems animating quantum complexity theory in the past decade: it ties in to the theory of quantum error-correcting codes, has applications to quantum cryptography, and promises insights into the study of entanglement in ground states of local Hamiltonians~\cite{AharonovAV13qpcp}. 

The question can be formalized in multiple ways. A first formulation, the
``constraint satisfaction'' variant of the quantum PCP (QPCP)
conjecture~\cite{AharonovILV09detectability}, asks for the complexity of constant-factor approximations to the minimal energy of a local Hamiltonian $H$, normalized so that $\|H\|=1$. Despite considerable attention progress on the
conjecture has been difficult~\cite{AharonovE15commuting,BrandaoH13product,eldar2015local}.  

More recently a second formulation has been put forward. The ``multiplayer
games'' variant of the QPCP conjecture, introduced in~\cite{FV14}, asks for the
complexity of estimating, to within constant accuracy, the maximum success
probability of provers (we use the terminology ``provers'' and ``players'' interchangeably)  sharing entanglement in a multiplayer game, a quantity referred to as the \emph{entangled value} of the game. The
conjecture is a natural analogue of the ``oracularized'' formulation of the PCP
theorem, which states that the maximum success probability of \emph{classical}
provers in a multiplayer game is $\NP$-hard to approximate to within constant
factors. (This can be thought of as a ``scaled down'' formulation of the equality $\MIP = \NEXP$~\cite{BabForLun91CC}.) 

 In~\cite{Vidick13xor}, building on~\cite{IV12} it was shown that the approximation problem for the entangled value of a multiplayer game remains $\NP$-hard, provided there are at
least three provers. This was extended to games with two provers only in~\cite{NatarajanV17twoprover} (this result will be used as a building block in the present paper). In~\cite{FV14,ji2015classical} it was shown that
inverse-polynomial approximations are $\QMA$-hard (provided there are at least
five provers), a result that is akin to a ``quantum Cook-Levin
theorem for entangled games.'' These results motivate the following conjecture, first made in~\cite{FV14}:

\begin{conj*}[Games QPCP conjecture (informal)]\label{conj:qpcp}
Suppose given as input an explicit description of a classical multiplayer
game. Then it is $\QMA$-hard to determine whether provers sharing quantum
entanglement (of arbitrary dimension) have optimal success probability at least $\frac{2}{3}$ or at most $\frac{1}{3}$ in the game.  
\end{conj*}

We show that the conjecture holds, under randomized reductions.

\begin{theorem}[Games QPCP under randomized reductions]\label{conj:weak_qpcp}
 Suppose given as input an explicit description of a classical multiplayer
game. Then it is $\QMA$-hard, under randomized reductions, to determine whether provers sharing quantum entanglement (of arbitrary dimension) have optimal success probability at least $1$ or at most $\frac{1}{2}$ in the game.  
\end{theorem}

Theorem~\ref{conj:weak_qpcp} is stated and proved as
Corollary~\ref{cor:randomized} in the body of the paper. 
The choice of constant $\frac{1}{2}$ in Theorem~\ref{conj:weak_qpcp} is arbitrary, as for the kind of games we consider soundness amplification can be performed efficiently in parallel~\cite{bavarian2017hardness}.

We explain the need for a
randomized reduction. Informally, the reason is that we do not
know of a strong enough $\QMA$-hardness result for the local Hamiltonian problem
to initiate our reduction. In fact, we give two alternate formulations of
Theorem~\ref{conj:weak_qpcp} that would also establish the same $\QMA$-hardness
result, under deterministic reductions, provided that either:
\begin{enumerate}
  \item[(i)] it is
$\QMA$-hard to approximate the minimum energy of a local Hamiltonian in $Y$-free form (Definition~\ref{def:gen-h}) to within constant accuracy (this is a variant of the quantum PCP conjecture for
local Hamiltonians), or 
\item[(ii)] it is $\QMA$ hard to approximate the ground energy
of (not necessarily local) frustration-free Hamiltonian whose every term is a
tensor product of generalized Pauli 
$\qp_X$ or $\qp_Z$ observables. 
\end{enumerate}
Note that point (i) amounts to a deterministic reduction from
Conjecture~\ref{conj:weak_qpcp} to the constraint satisfaction quantum PCP conjecture, and establishes the first proven relation between the two
conjectures (see~\cite{grilo2016pointer} for an incomparable result that relates stronger variants of both conjectures). Point (ii) is arguably a weaker assumption, as the gap is not required to be a constant and the terms of the Hamiltonian are not required to be local. However, due to the restriction that the Hamiltonian is frustration-free, it is currently not known whether the problem is $\QMA$-hard (or even $\QMA_1$-hard --- though the frustration-free assumption can be relaxed to having exponentially small ground state energy). 
 
Our results build on two main tools: a framework for protocols to test
  ground states, introduced in~\cite{FV14} and
further developed in~\cite{ji2015classical,NV17}, and a new proof of
soundness of the classical low-degree test of Raz and Safra against two
entangled provers~\cite{NatarajanV17twoprover}. The main result that underlies the complexity-theoretic applications is a two-prover test for $n$-qudit maximally entangled states, where each   qudit has dimension $q=p^t = \poly\log(n)$ for a prime $p$ and integer $t$, that
  has inverse robustness independent of $n$ (for all $\eps$ that are at least inverse polylogarithmic in $n$) and in which the verifier sends only
  $O(\log(n))$ bits to the provers, who reply with $O(\log\log n)$ bits each
  (Theorem~\ref{thm:qld}). This is an exponential improvement over all previous
  results, and provides the first robust entanglement test with  sub-linear
  communication. While the ability to ``test'' structured objects with sub-linear efficiency has become customary in classical computer science, we find it remarkable that the framework for such tests may be extended to test such a complex object as quantum entanglement.
	
	We first describe this test in more detail, before expanding on the complexity-theoretic consequences. 
	
\paragraph{Efficient, robust entanglement tests.} The driving question
behind our work is the following: ``Is it possible to verify a quantum
state using an amount of resources that scales sub-linearly in the
number of qubits of the state?'' We start with the ``simplest'' such
state---the maximally entangled state. Results in self-testing have
yielded increasingly efficient and robust tests for this state and
other, more general  families of highly entangled states. Here we
loosely refer to the ``efficiency'' of a test as a measure of the
total number of bits of communication involved in an execution of the
test. The ``robustness'' of the test indicates how tightly success in
the test characterizes the desired state: a test is
$\delta(\eps)$-robust if for all $\eps\geq 0$, any strategy for the
provers that succeeds with probability at least $1-\eps$ in the test
must use an entangled state that is within distance $\delta(\eps)$
from the tested state (see Definition~\ref{def:self-test}). Using
these measures, the best prior self-tests for a maximally entangled
state of $n$ qubits  are a test with communication $O(\log n)$ and
robustness $O(n^{5/2}\sqrt{\eps})$~\cite{chao2017test} and a test with
communication $O(n)$ and robustness
$O(\sqrt{\eps})$~\cite{NV17}. Other recent results in this direction
include~\cite{OV16,CN16,Coladangelo16,ColadangeloS17MS}.

Our test is the first to combine robustness $\delta(\eps) = \poly(\eps)$ that is independent of $n$, and logarithmic communication. Achieving both simultaneously is crucial to applications: constant (in $n$) robustness allows us to achieve gap-preserving reductions; logarithmic communication allows us to achieve efficient reductions. 

As in previous results, the test is designed to constrain successful provers to
use  observables satisfying suitable relations; a statement about the entangled
state follows by using that the state is stabilized by (a subset of) these
observables. In the case of the maximally entangled state, the observables are
all $n$-fold tensor products of Pauli observables. For reasons to be discussed
below we test for qudits of dimension $q=p^t$ a prime power of order
$q=\poly\log(n)$. This leads us to consider tensor products of single-qudit
Pauli observables defined over the prime power field $\Fq$, which we denote
using the symbol $\qp$:
\begin{equation}\label{eq:gen-pauli}
 \qp_X(a) = \sum_{j \in \Fq} \ket{j+a}\bra{j}\qquad\text{and}\qquad \qp_Z(b)= \sum_{j \in \Fq}
  \omega^{\tr(b j)} \ket{j}\bra{j}\;, 
	\end{equation}
where $a,b\in\Fq$, $\omega = e^\frac{2i\pi}{p}$, addition and multiplication are
over $\Fq$, and $\tr(\cdot)$ denotes the trace of $\Fq$ over $\Fp$. The main
difficulty we face is that there are $2\cdot q^n$ such observables,
$\qp_X(a)=\qp_X(a_1)\otimes \cdots\otimes  \qp_X(a_n)$ and
$\qp_Z(b)=\qp_Z(b_1)\otimes \cdots \otimes \qp_Z(b_n)$ for $a,b\in \Fq^n$, an
exponentially larger number than any test with polylogarithmic communication
gives us direct access to. It is then natural to consider a test that certifies
observables $\qp_X(a)$ and $\qp_Z(b)$ for $a,b\in T \subseteq \Fq^n $, where
$|T|=\poly(n)$, and attempt to construct observables for all $a,b\in
\Fq^n $ in an inductive fashion, as is done in e.g.~\cite{chao2017test}, where
$T$ is the set of all strings of Hamming weight at most $2$. Unfortunately, any
na\"ive procedure will induce an error accumulation at each step of the
induction, eventually resulting in a robustness parameter that depends
polynomially on $n$ (as is the case in~\cite{chao2017test}). 

It is thus crucial to choose the set $T$ carefully --- informally, it
seems natural to require that this set behave in a ``pseudorandom''
way. We take direct inspiration from the classical proof of the PCP
theorem, and use a set $T$ specified as the set of all codewords of a
suitably chosen Reed-Muller code; this is the reason for using a
sufficiently large qudit dimension $q$. Our proof eventually reduces
the analysis to the soundness of the entangled-prover classical
low-degree test~\cite{NatarajanV17twoprover}. We explain the test, and
its analysis, in more detail in Section~\ref{sec:techniques} below.

\paragraph{Testing ground states and a ``gap preserving'' reduction.}
We sketch how our test for entanglement is applied to obtain results on the complexity of multiplayer entangled games. In the classical case, the
proof that the value of a multiplayer game is at least as hard to approximate as
the maximum fraction of constraints simultaneously satisfiable in a local constraint
satisfaction problem proceeds via the technique of oracularization: the verifier
selects a constraint at random and asks one prover for an assignment to all
variables in the constraint and the other for an assignment to a
single one of the
variables. Given the provers' answers, the verifier checks the natural
satisfaction and consistency constraints. In the quantum case the analogous idea
would require each prover to hold a copy of the ground state of a
$\QMA$-complete local Hamiltonian, and return qubits as requested by the
verifier. This reduction
does not work: it is not possible in general to check for ``consistency''
between the same qubit taken from two copies of an entangled
state. In~\cite{FV14} the idea was introduced of encoding the ground state using an error-correcting
code and distributing a share to each prover. Subsequent
work~\cite{ji2015classical} showed that this idea can be used to show $\QMA$-hardness of
inverse-polynomial approximations to the entangled value of a multiplayer
game. Unfortunately the reduction in~\cite{ji2015classical} is not
``gap-preserving'': a large promised energy gap in the starting
instance of the local Hamiltonian problem does not lead to a large
completeness-soundness gap in the resulting game. As a result, even assuming the ``constraint satisfaction'' QPCP does not lead to hardness for
approximation factors larger than a fixed inverse polynomial. In~\cite{NV17} we
leveraged an entanglement test with  constant robustness to achieve a
gap-preserving reduction; unfortunately communication in the test is linear, resulting in a game with
exponential size, so that no new complexity-theoretic consequence is
obtained.   

Armed with an exponentially more efficient entanglement test we are able to
provide a much more effective reduction, yielding games of polynomial size from instances of the local
Hamiltonian problem. The reduction follows
similar lines as previous work, but with a new difficulty. Our entanglement
test only certifies a specific family of observables: tensor products of
generalized Pauli observables~\eqref{eq:gen-pauli} over $\F_q$, for $q$ a
sufficiently large prime power. This requires us to initiate any direct reduction with a
specific class of Hamiltonians, in so-called $Y$-free form (see
Definition~\ref{def:gen-h}); informally, these are local Hamiltonians such that
each local term is a tensor product of generalized $\qp_X$ and $\qp_Z$ observables.  In the absence of general
gap-preserving reductions between different variants of the local Hamiltonian problem (perturbation
techniques~\cite{CM13} do not generally preserve the promise gap) we obtain a
reduction to the hardness of constant-factor approximations to the ground
energy of local Hamiltonian of this form only. Nevertheless, even though the
entanglement test requires a qudit dimension that scales (poly-logarithmically)
with $n$, we show that any qubit Hamiltonian in $Y$-free form can be embedded in a
Hamiltonian in $Y$-free form over qudits of dimension $2^t$ for any
$t\geq 1$. As a result, we immediately obtain point (i) discussed earlier: that
Conjecture~\ref{conj:weak_qpcp} would follow from $\QMA$-hardness of
constant-factor approximations to local Hamiltonian whose every local term is a
tensor product of $\qp_X$ and $\qp_Z$ Pauli observables (signed weights of up to
poly-logarithmic size are allowed).   

\paragraph{Composition and PCP.}
To obtain strong results we develop more elaborate reductions, with the  aim of removing the assumption on \emph{locality} of the Hamiltonian whose ground state energy is being
tested. As our entanglement test has direct access only to local Pauli observables, it cannot be used to evaluate the expectation value of non-local observables (acting on more than a constant number of qudits). We get around this as follows. Say the verifier would like to estimate the expectation value of a nonlocal tensor product observable such as $\qp_{X}(b)$, for some $b\in\Fq^n$.
The verifier asks each prover to measure all its qudits in the $X$ basis, obtaining an outcome
$a \in \Fq^n$, and report the value of the inner product $c = b \cdot a$. This provides the verifier with an estimate of the energy of $\qp_{X}(b)$. However, it remains to ensure that the outcome reported by the prover was obtained honestly, i.e.\ by measuring all qudits on which the observable acts, without having the ability to ``read'' all the single-qubit outcomes obtained. This sounds very similar to the kind of NP statements that PCPs are designed to allow efficient verification of, and indeed we employ classical PCP techniques, more specifically the notion of \emph{PCP of proximity} (PCPP). 

In order to verify that a prover honestly computed the inner product $c = b \cdot a$, the verifier asks it to provide PCPP of this fact. A PCPP for a language is a proof that a given
input is in the language, which can be verified by making only a few queries to
both the proof and the input.
In our setting, the verifier asks each prover to compute
a PCPP $\Pi$ for the claim that
the measurement outcome string $a$ is in the language $L = \{x : b \cdot x = c\}$.  This proof can be
verified by making constantly-many queries to $\Pi$, together with constantly
many queries to $a$. Both of these correspond to \emph{local} measurements,
either of the shared quantum state, or the proof string $\Pi$ generated from
the measurement outcomes, and can thus be certified using our entanglement test.

There are two subtleties that arise. First, a PCPP (viewed as a
nonlocal game) that is classically sound need not be sound against entangled provers. To address this, we perform a further layer of
composition, encoding the PCPP proof $\Pi$ in a low-degree polynomial and
querying this polynomial. Secondly, in our
setting \emph{completeness} does not automatically hold either. This is because
each prover $j$ only has access to one share of the shared state, which is a qudit-by-qudit encoding of the actual QMA witness. The prover can thus only supply bits from a
proof $\Pi_j$ computed from its share. As a result the usual method of
transforming a PCP into a game, namely by querying multiple provers for
locations in the proof and checking consistency between them, fails since even
honest provers do not know each other's measurement outcomes and thus cannot
answer consistently. To surmount this
obstacle, we exploit the linearity of the error correcting code, together with a
linear PCPP construction from~\cite{BGHSV05}, for which the proof $\Pi$
 is a linear function of the input $a$; the linearity holds as long as the language $L$ is
itself specified by a set of linear equations, i.e. $L = \{x: Ax = b\}$. The
linearity of the PCPP allows the verifier to check consistency between one prover's
answers and the appropriate linear combination of answers returned by the other
provers.\footnote{We note that, just as
in~\cite{BGHSV05}, we require linearity of the PCP in order for it to interface with
a linear error correcting code.}.

With this PCPP-based protocol for measuring nonlocal Pauli observables in place, the proof of Theorem~\ref{conj:weak_qpcp} follows: starting with a
$\QMA$-hard instance of the local Hamiltonian problem with
inverse-polynomial promise gap, we amplify the gap by taking a large tensor product,
and then randomly sample a polynomial subset of the exponentially many terms in
the tensor product. By the matrix Chernoff bound~\cite{AW02}, with high
probability this sampling preserves the promise gap, and the resulting
nonlocal Hamiltonian can be tested using our protocol. (This random sampling is the source of the ``randomized reductions'' in Theorem~\ref{conj:weak_qpcp}.)

Finally, our PCPP-based protocol enables us to
check not just one nonlocal term but also many terms at once, provided that they
are all tensor products of Paulis in the same basis. This allows us to obtain a
protocol that accommodates an inverse-polynomial promise gap for the ground energy,
provided the Hamilton is frustration free (all of its terms are
simultaneously satisfied in the ground state), and each of it terms can be expressed as a tensor product of generalized $\qp_X$ or
  $\qp_Z$ observables, acting on an arbitrary number of qudits (see
  Definition~\ref{def:linear-xz}). This shows point (ii) discussed earlier.

\subsection{Techniques}
\label{sec:techniques}

Our main result, a robust entanglement test with logarithmic communication, can be stated informally as follows. For a formal statement, we refer to Theorem~\ref{thm:qld} in Section~\ref{sec:qld}.

\begin{theorem*}
Let $n$ be an integer and $q = p^t$ a prime power such that $q = \Theta(\frac{\log^2 n}{\log \log n})$.
	Then there exists a two-prover test $\qld$ in which the
        verifier sends questions of length $\poly(\log n,\log q)$ and
        receives answers of length 
        $O(\poly\log\log(n) \cdot \log q)$ such that the following hold: 
	\begin{enumerate}[nolistsep]
	\item \emph{(Completeness:)} There exists a strategy for the provers based on sharing an $n$-qudit maximally entangled state, with qudits of local dimension $q$, and making measurements in the eigenbasis of tensor products of generalized $\qp_X$ or $\qp_Z$ observables over $\Fq$;
	\item \emph{(Soundness:)} For any $\eps \geq 0$, any strategy that is accepted
    with probability at least $1-\eps$ in the test must use an entangled state
    that is (up to local isometries) within distance $\delta =
    \poly(\poly(p)\cdot\poly(\eps) )$ from an $n$-qudit maximally entangled state.\footnote{Here and
      throughout we use the notation $f(X)=\poly(h(X)))$ as an abbreviation for ``there exists a universal constant $c>0$ such that $f(X) = O(h(X)^c)$ as $X\to 0$ (if $X=\eps$) or as $X\to\infty$ (if $X=n$); in the theorem $p,t$ and $q$ are all allowed to be implicitly functions of $n$, but not $\eps$.}
\end{enumerate}
\end{theorem*}

A typical setting of parameters for the theorem is to choose $p$ a constant, e.g. $p=2$, $t= \Theta(\log\log n)$, and $\eps$ a small constant, which leads to constant soundness $\delta$. 

The test mentioned in the theorem has three components: (a) a low-degree test in the $X$ basis; (b) a low-degree test in the $Z$ basis; (c) an anti-commutation test relating the two bases. Both (a) and (b) are direct adaptations of the ``plane-vs-point'' low-degree test from~\cite{raz1997sub}. The basis label, $X$ or $Z$, asks the prover to measure its $n$ qudits in the simultaneous eigenbasis of the observables $\qp_X(a)$ or $\qp_Z(b)$ defined in~\eqref{eq:gen-pauli} respectively.  The prover is then asked to encode the resulting outcome $a\in \Fq^n$ as a low-degree polynomial $g_a:\Fq^m \to \Fq$, where $m = O(\log n / \log\log n)$, and return either the evaluation of the polynomial at a randomly chosen point $x\in \Fq^m$, or its restriction to a randomly chosen two-dimensional subspace $s$ of $\Fq^m$. Part (c) is designed to enforce the ``twisted commutation'' relations $\qp_X(a)\qp_Z(b) = \omega^{-\tr(ab)} \qp_Z(b)\qp_X(a)$ satisfied by these observables. 
Before explaining the test and its analysis in greater detail, we first review the main steps that go into showing soundness of the classical low-degree test. 

\paragraph{Classical low-degree tests.}
The effectiveness of the classical low-degree test is based on the use of the following Reed-Muller encoding of an $n$-variable assignment $a=(a_1,\ldots,a_n)\in\{0,1\}^n$. First, integer values $h$ and $m$ are chosen so that $h^m \geq n$, and an injection $\bij:\{1,\ldots,n\} \to \{0,\ldots,h-1\}^m$ is fixed. Second, a finite field $\Fq$ is chosen such that $q\geq h$. Third, a function $g_a : \Fq^m\to \Fq$ is defined such that $g_a(\bij(i)) = a_i$ for all $i\in\{1,\ldots,n\}$, and $g_a$ has degree at most $h$ in each of its $m$ variables; $g_a$ can be obtained by straightforward polynomial interpolation. Finally, the encoding of $a$ is defined as the concatenation of the evaluation table of $g_a$ at every point $x\in \Fq^m$ with a table describing the restriction of $g_a$ to every two-dimensional subspace $s\subseteq \Fq^m$. The encoding has roughly $q^{3m}$ entries, and each entry has size $O(d^2 \log q)$, where $d=mh$ is the total degree of $g_a$. Choosing $h \approx \log n$ and $m \approx \log n/\log\log n$ yields an encoding of quasi-polynomial size, $n^{O(\log n)}$, as long as $q$ is also polynomial in $n$.

When used for constructions of PCPs, the low-degree test provides an encoding that can be tested and evaluated while making only a small number of queries. This is achieved based on the following observations. First, the encoding can be checked by making only a constant number of queries: the test selects a pair $(x,s)$ such that $s$ is a uniformly random subspace and $x$ a uniformly random point in $s$, and checks consistency between the corresponding entries of the encoding. Second, the evaluation of $g_a$ at any point $z\in \Fq^m$ can be recovered by making $O(d)$ queries to the encoding in a way that each query is uniformly distributed: select a uniformly random line going through $z$, query $d+1$ points at random on the line, and interpolate to recover the value at $z$. 

The analysis of the low-degree test described in the previous paragraph is not simple. The goal is to show that any table which passes the test with probability $1-\eps$ must be close to the encoding of a polynomial of the form $g_a$, for some $a\in \Fq^n$. The proof is constructive: it recovers a  low-degree polynomial $g_a$ through $m$ successive steps of interpolation. The case $m=2$ is immediate, since by definition the encoding contains the restriction of $g_a$ to any two-dimensional subspace. For general $m$, one selects $(d+1)$ parallel $(m-1)$-dimensional subspaces, applies the induction hypothesis to each, and interpolates to recover a $m$-variate polynomial defined over the whole space. The key difficulty in the analysis is to control the error: na\"ively, it would, at best, double at each step, resulting in an unmanageable blow-up. The key innovation of the test, and its analysis, is a method to limit this blow-up by a procedure of ``self-improvement''. 

\paragraph{Entanglement tests.} 
Before moving on to our quantum low-degree test, it is useful to first
recall the intuition behind our prior work~\cite{NV17}, which
establishes a similar quantum  analogue for the Hadamard encoding, which is based on the linearity
test of Blum et al.~\cite{BLR93}.  

In the linearity test, the assignment $a\in\{0,1\}^n$ is encoded as the evaluation table of the function $f_a : \F_2^n \to \F_2$, $f_a(x)=x\cdot a$. Each entry of the encoding is a single bit, but there are $2^n$ entries, thus the table has exponential size. The linearity test makes three queries, $x,y$ and $x+y$ for $x,y$ uniformly distributed in $\F_2^n$, and verifies that $f_a(x)+f_a(y)=f_a(x+y)$. The soundness analysis of the test is based on Fourier analysis; no induction is needed. 

To turn the linearity test into a test for entanglement we first re-interpret it using the language of representation theory. The additive structure of $\F_2^n$ makes it into an abelian group, whose irreducible representations are the $2^n$ characters $\chi_a(x)=(-1)^{a\cdot x}$. An arbitrary table $f: \F_2^n \to \F_2$ can also be seen as  a mapping $g=(-1)^f$ from the additive group  of $\F_2^n $ to the $1$-dimensional unitary group, $U( \C)$. A table $f$ which is accepted in the linearity test with probability $1-\eps$ is an approximate representation of the group, in the sense that $\Es{x,y} |g(x)g(y)-g(x+y)|^2  = O(\eps)$, where the expectation is uniform. Thus the analysis of the linearity test exactly amounts to showing that approximate representations of abelian groups are close to exact representations (i.e. the characters, which precisely correspond to the linear functions). 

We can try to apply the same reasoning to entangled-prover
strategies. Using matrix-valued Fourier analysis it is possible to
show that a quantum strategy which succeeds with probability $1-\eps$
in an $X$-basis linearity test (resp. an $Z$-basis linearity test)
implies the existence of observables for the provers which satisfy
approximate linearity conditions $X(a)X(b)\approx X(a+b)$
(resp. $Z(a)Z(b)\approx Z(a+b)$), where the approximation holds on
average over uniform $a,b\in\F_2^n$ and is measured using the
state-dependent norm that is standard in testing. These relations by
themselves do not imply anything ``quantum''; in particular they are
satisfied by one-dimensional observables $X(a)=Z(a)=(-1)^{f(a)}$. To
obtain a truly quantum test we are missing a constraint relating the
two bases: the Pauli (anti)-commutation relation
$X(a)Z(b)=(-1)^{a\cdot b}Z(b)X(a)$. Enforcing this relation would
allow us to frame the family of unitaries $\{\pm
X(a)Z(b),\,a,b\in\F_2^n\}$ as a representation of the Pauli group
modulo complex phase (also known as the Weyl-Heisenberg group) and
combine results on the stability of approximate
representations~\cite{gowers2015inverse} with information on the
structure of irreducible representations of that group to
conclude. This is what justifies the inclusion of part (c), an
anti-commutation test, which can be based on e.g. the Mermin-Peres
Magic Square game~\cite{Arvind:02} to test for the desired
anti-commutation relations.  

\paragraph{A quantum low-degree test.}
The previous outline of an entanglement test based on the BLR linearity test is implemented in~\cite{NV17}. The use of the linearity test has two main advantages: (i) when executed in a single basis, its analysis with two entangled provers follows a direct argument using Fourier analysis; (ii) combining the linearity test in the $X$ and $Z$ bases naturally gives access to two families of observables $X(a)$ and $Z(b)$ for the provers, that can be used to specify an approximate representation of the $n$-qubit Weyl-Heisenberg group as described above, with the (anti)-commutation test certifying all required pairwise group relations. 

To reduce the communication required in the test, it is natural to turn to  low-degree tests: as described above, the latter only require poly-logarithmic, instead of linear, communication. Due to the fact that the test has only a quasi-polynomial number of questions, however, a strategy for the provers only involves a quasi-polynomial number of observables: how can one show that all exponentially many (anti)-commutation relations hold, in principle, between observables defined on the prover's space, if the test itself only requires the existence of a tiny subset of these observables in order to be played? 

This difficulty can be overcome as follows. From the classical analysis of the low-degree test, or rather its entanglement-resistant analogue~\cite{NatarajanV17twoprover}, it is possible to show that a strategy that succeeds in the $X$-basis (resp. $Z$-basis) low-degree test implies the existence of a family of observables $X(a)$ (resp. $Z(b)$), for $a\in \Fq^n$, that satisfy the commutation relations $X(a)X(b)=X(a+b)$ (resp. $Z(a)Z(b)=Z(a+b)$). Moreover, the use of an appropriate generalization of the Magic Square game over $\Z_2$, introduced in~\cite{ColadangeloS17MS}, to $\Z_s$, for any integer $s$, that allows us to test for the appropriate twisted commutation relation between any two observables that are actually queried in the test. The difficulty is to establish the right relations between observables $X(a)$ and $Z(b)$ that are not queried from the test, but whose existence follows from the independent application of the entangled-prover analysis of the low-degree test to the $X$- and $Z$- basis executions of the test. 

Our solution proceeds in three steps. The first step consists in combining $X$ and $Z$ observables together into a single family of commuting observables. We do this by adjoining two ancilla systems for each prover, each initialized in a maximally entangled state local to the prover, and setting $\hat{X}(x) = X(x)_{\reg{A}} \otimes \qp_X(x)_{\reg{A'}} \otimes \Id_{\reg{A''}}$, where $\qp_X(x)_{\reg{A'}}$ denotes the $n$-qudit Pauli that the prover's  $X(x)$ is supposed to implement, for $x$ among the possible queries in the test. Defining $\hat{Z}(z)$ similarly, provided $\hat{X}(x)$ and $\hat{Z}(z)$ satisfy the (conjugate of) the twisted commutation relation satisfied by $\qp_X(x)$ and $\qp_Z(z)$ we have obtained a family of (approximately) commuting observables. 

In the second step we use these commuting observables to define a
strategy for the classical low-degree test, not over $m$-variate
polynomials as the initial test requires, but over $2m$ variables,
half of which are ``$X$'' variables, and half of which are ``$Z$''
variables. To construct such a strategy we have to define ``points''
and ``subspace'' measurements from the $\hat{X}(x)$ and $\hat{Z}(z)$,
using the information that the initial observables $X(x)$ and $Z(z)$
came from a strategy for the provers that independently succeeded,
with good probability, in the classical low-degree test. Once this has
been completed we apply the analysis of the classical low-degree test
against two entangled provers to recover a single family of
measurements $\{\hat{S}^g\}$ with outcomes in the set of low-degree
polynomials $g$ over $\Fq^{2m}$. 

The last step consists in ``pulling apart'' the measurements obtained
in the previous step to recover observables $\tilde{X}(x)$ and
$\tilde{Z}(z)$, now defined for all $x,z\in\Fq^n$ (and not only the
special subset used as queries in the test). Given the definition of
$\hat{X}(x)$ from $X(x)$, it is natural to define $\tilde{X} =
(\Id_{\reg{A}} \otimes \Id_{\reg{A'}} \otimes \qp_X(x)_{\reg{A''}})
\cdot \hat{X}(x)$, which has the effect of ``undoing'' the initial
tensoring of $X(x)$ by a Pauli on $\reg{A'}$ (this uses that the
ancillas $\reg{A'A''}$ are initialized in a maximally entangled
state). It remains to argue that the exponentially many operators thus
constructed approximately satisfy the Pauli twisted commutation
relations. Once this has been established the result follows as in our
previous work~\cite{NV17}, as it can be shown directly that such
operators must be close to operators exactly satisfying all Pauli
relations, whose only joint eigenvalue-$1$ eigenstate is the maximally
entangled state.

\paragraph{Pauli observables over a prime power field.} 
To conclude this
overview we briefly discuss some difficulties encountered while working with
generalized Pauli observables over a prime power field. Had we
restricted attention to prime fields the proof (and certainly the notation!) would have been somewhat simpler. The
motivation for considering prime powers comes from the desire to allow embedding
qubit Hamiltonians, which we can achieve if $q=2^t$, but did not see how to
implement for odd values of $q$. Over prime power fields, we are
faced with two possible definitions of generalized Pauli observables: the
``clock'' and ``shift'' operators mod $q$, with eigenvalues that are $q$-th roots
of unity, and the definition~\eqref{eq:gen-pauli}, with eigenvalues that are $p$-th roots of unity. The
former are more common in the literature and offer the convenience of allowing to encode a projective measurement with outcomes in
$\Fq$ into a single generalized observable. However, they are not well-suited
for describing strategies in the low-degree test, since they are defined in terms of addition and multiplication over
$\Z_q$, whereas in the low degree test, all operations are performed over
$\Fq$. Hence, we opted for the second definition, using families of $t$ such
observables to encode a single measurement with outcomes in $\Fq \simeq \Fp^t$.

\subsection{Further work} 

There are several open problems raised by our work. Firstly, it would be interesting to expand the range of Hamiltonians for which we are able to give constant-gap interactive proofs, with the goal of eventually reaching a $\QMA$-complete family, and thus a proof of Conjecture~\ref{conj:qpcp} based on a deterministic reduction. Secondly, a different route towards the proof of the conjecture would consist in establishing
$\QMA$-hardness results for either of the two classes of Hamiltonians described
in Definition~\ref{def:gen-h} and Definition~\ref{def:linear-xz}, for
which we do already have a deterministic reduction to a game.  As further motivation, we note that, if such a $\QMA$-hardness result were achieved by
constructing a ``history Hamiltonian'' from a polynomial quantum circuit---as
in all such hardness results known---then by an observation of
Fitzsimons and Hajdu{\v{s}}ek~\cite{FitzsimonsH15}, our results could be
used to give an efficient delegation scheme for BQP in the ``post-hoc''
model. More broadly, the classical PCP theorem and MIP proof systems have
become important tools in the design of delegated computation schemes (e.g.~\cite{KRR14,RRR16}), and we hope that similar applications may arise from the games variant of
QPCP.  Beyond the quantum games PCP conjecture, essentially resolved in this work, the complexity of the class $\MIP^*$ of languages that have multi-prover interactive proof systems with entangled provers remains wildly open. Recent work of~\cite{ji16nexp} introduces a ``compression'' technique, that allows him to obtain $\MIP^*$ protocols for language in NEEXP (non-deterministic doubly-exponential time), albeit at the cost of an exponentially small completeness-soundness gap. Could our techniques be used to obtain the same result, for a constant gap? Such a result would provide an unconditional separation between $\MIP$ and $\MIP^*$.

In a different direction, it could be useful to extend our
entanglement test to sub-constant error, in the same spirit
as~\cite{arnon2017noise,arnon2017device}. Currently, all self-testing
results we are aware of only provide guarantees in a regime where the
success probability is close to $1$, which is arguably more
challenging to demonstrate in experiments. 

\paragraph{Organization.} We start with important notation and general
preliminaries in Section~\ref{sec:notation}. The quantum low-degree test is
stated in Section~\ref{sec:qld}, and its soundness analysis is given in
Section~\ref{sec:soundness}. In Section~\ref{sec:code} we extend the test to
allow testing arbitrary states encoded using a suitable error-correcting
code. Finally, Section~\ref{sec:testing} applies the test to prove Theorem~\ref{conj:weak_qpcp}, together with the two variants discussed as items (i) and (ii) in the introduction. 

\paragraph{Acknowledgments.} 
AN was supported by NSF CAREER Grant CCF-1452616. TV was supported by NSF CAREER Grant CCF-1553477,
AFOSR YIP award number FA9550-16-1-0495, a CIFAR Azrieli Global Scholar award,
and the IQIM, an NSF Physics Frontiers Center (NSF Grant PHY-1125565) with
support of the Gordon and Betty Moore Foundation (GBMF-12500028). Parts of this
work was completed while both authors were hosted at the Institut Henri
Poincar\'e in Paris, as part of the special semester on Analysis in Quantum
Information Theory (Fall 2017), supported by NSF Grant DMS-1700168. The hospitality of the IHP is gratefully acknowledged.

\section{Preliminaries}
\label{sec:notation}

\subsection{Notation}

We use $\mH$ to denote a finite-dimensional Hilbert space, $\Lin(\mH)$ for the linear operators on $\mH$, and $\Unitary(\mH)$ the set of unitary operators. Subscripts $\mH_\reg{A}$, $\mH_\reg{B}$ indicate distinct spaces. 

We use the notation $\poly(f(n))$ to denote $O(f^c(n))$ for some universal constant $c>0$ (which may vary each time the notation is used). Similarly, we write $\poly^{-1}(f(n))$ to denote $\Omega(f^{-c}(n))$. All parameters used in the paper will generally be a function of a single parameter $n$, and asymptotic notation $O(\cdot)$, $\Omega(\cdot)$, etc., should be understood as $n\to \infty$. 

\subsection{Finite fields and polynomials}

Throughout we use $p$ to denote a prime and $q = p^t$ a prime power. We let
$\Fq$ denote the finite field with $q$ elements, and $\Zp$ denote the cyclic
group mod $p$. The additive group of $\Fp$ coincides with $\Z_p$, but this is no longer the case for $\Fq$. The
finite field trace is denoted by $\tr(a)$; it is a map from $\Fq$ to the prime subfield $\Fp$, defined by $\tr(a) = \sum_{\ell = 0}^{t-1}
a^{p^{\ell}} $. The trace respects linear combinations with coefficients drawn
from the prime subfield:  $\tr(\alpha a + \beta b) = \alpha \tr(a) + \beta
\tr(b)$ for $\alpha, \beta \in \Fp$.  A useful alternative view of $\Fq$ is as a $t$-dimensional vector
space over $\Fp$. Each element $e \in \Fq$ can
be written as $e_1 b_1 + e_2 b_2 + \dots + e_{t} b_t$, where $(b_1, \dots,
b_t)$ is a basis for $\Fq$ over $\Fp$ and the coefficients $e_\ell$ lie in the field of scalars
$\Fp$. This representation of $\Fq$ is convenient for addition, since one can
add the individual components $e_\ell$ separately, but in general, it is
hard to do multiplication. However, if $q$ is even or $q = p^t$ with both $p$ and
$t$ odd there always exists a basis
satisfying the property of \emph{self-duality}, i.e.
\begin{equation}\label{eq:self-dual}
\tr(b_i b_j) \,=\, \delta_{ij}
\end{equation}
for all $i,j\in\{1,\ldots,t\}$ (see e.g.~\cite[Theorem~1.9]{menezes2013applications}). This property allows to express
 $\tr(ef)$, for $e, f \in \Fq$, as the inner product, over $\Fp$, of their respective vector of components along the
basis. As shown below, this property will make it convenient to express $q$-dimensional qudits as
tensor products of $p$-dimensional qudits. For the remainder of the
paper we only consider choices of $q$ such that $\Fq$ admits a
self-dual basis over $\Fp$.

For integer $d,m$ and a subspace $s\subset \Fq^m$ we let $\deg_d(s)$ denote the set of polynomials on $s$ of total degree at most $d$ (specified with respect to some fixed, implicit basis for $s$). We write $\omega = e^{\frac{2i\pi}{p}}$ for a fixed primitive $p$-th root of unity. Let 
\begin{equation}\label{eq:def-epr-p}
\ket{\epr_q} = \frac{1}{\sqrt{q}}\sum_{i \in \Fq} \ket{i}\otimes \ket{i}\,\in\,\C^q\otimes\C^q\;.
\end{equation}

\paragraph{Coordinates and polynomials.}
Let $n\geq 1$ be an integer, and  $h,m$ two integers such that $h^m \geq n$ and $h\leq q$. 
Throughout we fix an arbitrary injection $\bij:\{1,\ldots,n\}\to\{0,1,\ldots,h-1\}^m \subseteq \Fq^m$, where $n,h,m$ are integers such that $h^m \geq n$ that will be clear from context. For $x\in \Fq^m$ and $i\in\{1,\ldots,n\}$ define 
$$ x_{\bij(i)}\,=\,\prod_{j=1}^m \frac{\prod_{\substack{k=0\\ k\neq \bij(i)_j}}^{h-1} (k-x_j)}{\prod_{\substack{k=0\\ k\neq \bij(i)_j}}^{h-1} (k-\bij(i)_j)}\,\in\,\Fq\;,$$
and let $x_\bij=(x_{\bij(1)},\ldots,x_{\bij(n)})\in\Fq^n$. Note that
  for $x\in \{0,1,\ldots,h-1\}^m$, $x_{\bij(i)}=1$ if $x=\bij(i)$ and $x_{\bij(i)}=0$
  otherwise.
By ranging over all possible
values for $x$ we obtain a subset of $\Fq^n$ of size $q^m$; we think of $x\mapsto x_\bij$ as a
pseudo-random ``coordinate expansion'' map. 

Let $g:\Fq^m\to\Fq$ be an $m$-variate polynomial of degree at most $h$ in each
coordinate. Then by interpolation we can write
\begin{equation}\label{eq:interpolation}
g(x) = \sum_{i=1}^n x_{\bij(i)} g(\bij(i)) = g\cdot x_\bij\;,
\end{equation}
where we abuse notation and write $g$ for the vector
$(g(\bij(1)),\ldots,g(\bij(n)))\in\Fq^n$. Conversely, for any $a\in\Fq^n$ we let
$g_a$ be the $m$-variate polynomial of individual degree at most $h$ over $\Fq$
defined by  
\begin{equation}\label{eq:def-ga}
g_a:\;x\in\Fq^m\,\mapsto\, \sum_i a_i x_{\bij(i)} \,=\, a\cdot x_\bij\;.
\end{equation}
The map from $\Fq^n$ to $\Fq^{q^m}$ that maps $a$ to the evaluation table of $g_a$ is the $m$-variate Reed-Muller code of individual degree $h$.
Note that $(g_a(\bij(1)), \dots, g_a(\bij(n))) = a$.

We recall the Schwartz-Zippel lemma~\cite{Zippel79,schwartz1980fast}, which we will use repeatedly. 

\begin{lemma}[Schwartz-Zippel]\label{lem:sz}
Let $d,m\geq 1$ be integers and $r$ a non-zero polynomial in $m$ variables of total degree at most $d$ defined over the finite field $\Fq$. Then $r$ has at most $d|\Fq|^{m-1}$ zeros.
\end{lemma}

\subsection{Pauli measurements and observables for qudits}
\label{sec:qauli}

To any projective measurement  $\{M^a\}$ with outcomes $a \in \Zp$ we can associate a generalized observable with eigenvalues that are $p$-th roots of unity: the unitary matrix $M = \sum_a \omega^a M^a$, where
$\omega = e^{\frac{2i\pi}{ p}}$. The generalized Pauli operators over $\Fp$ are a set of generalized observables indexed by a basis
setting $X$ or $Z$ and an element $a$ or $b$ of $\Fp$, with eigenvalues that are $p$-th
roots of unity. They are given by
\begin{equation}\label{eq:pauli-fp}
 \sigma_X(a) = \sum_{j \in \Fp} \ket{j + a} \bra{j}\qquad \text{and}\qquad\sigma_Z(b) =\sum_{j \in \Fp} \omega^{bj} \ket{j} \bra{j}\;, 
\end{equation}
where addition and multiplication are over $\Fp$. These observables obey the ``twisted
commutation'' relations
\begin{equation}\label{eq:twisted-fp}
 \forall a,b\in\Fp,\qquad \sigma_X(a) \sigma_Z(b) \,=\, \omega^{-ab} \,\sigma_Z(b)\sigma_X(a)\;. 
\end{equation}
Similarly, over a field $\Fq$ we can define a set of generalized Pauli operators,
indexed by a basis setting $X$ or $Z$ and an element of $\Fq$. There are different possible definitions for these operators. We choose them
to have eigenvalues that are $p$-th roots of unity. For $a,b\in\Fq$ they are given by
\[ \qp_X(a) = \sum_{j \in \Fq} \ket{j+a}\bra{j}\qquad\text{and}\qquad \qp_Z(b)= \sum_{j \in \Fq}
  \omega^{\tr(b j)} \ket{j}\bra{j}\;, \]   
where addition and multiplication are over $\Fq$. Powers of these
observables obey the relation
\[ \forall W \in \{X,Z\},\, \forall a \in \Fq,\,\forall b \in \Fp,\qquad(\qp_W(a))^b = \qp_W(a b) \;.\]
In particular, since $pa=0$ for any $a\in \Fq$ we get that 
that $(\qp_W(a))^p = \Id$ for any $a\in\Fq$. The observables  obey analogous ``twisted commutation'' relations to~\eqref{eq:twisted-fp},
\begin{equation}\label{eq:twisted-fq}
\forall a,b\in\Fq,\qquad \qp_X(a) \qp_Z(b) = \omega^{-\tr(a b)} \qp_Z(b) \qp_X(a)\;. 
\end{equation}
It is clear from the definition that all of the $\qp_X$ operators commute with each other, and similarly all the
$\qp_Z$ operators with each other. Thus, it is meaningful to speak of a common eigenbasis for all $\qp_X$ operators, and a common eigenbasis for all $\qp_Z$ operators. The common
 eigenbasis for the $\qp_Z$ operators is the computational basis. To map this basis to the common eigenbasis of the $\qp_X$ operators, one can apply the Fourier transform
\begin{equation}\label{eq:fourier-f}
 F \,=\, \frac{1}{\sqrt{q}} \sum_{j ,k\in \Fq} \omega^{-\tr(jk)} \ket{j}\bra{k}\;. 
\end{equation}
Explicitly, the eigenbases consist of the vectors $\ket{e_W}$ labeled by an element
$e \in \Fq$ and $W \in \{X, Z\}$, given by
\[\ket{e_X} = \frac{1}{\sqrt{q}} \sum_j \omega^{- \tr(e j)} \ket{j}\;, \qquad
\ket{e_Z}   = \ket{e}\;. \] 
We denote the POVM whose elements are projectors onto basis
vectors of the eigenbasis associated with the observables $\qp_W$ by
$\{\qp_W^e\}_e$. 
Then the observables $\qp_W(a)$ can be written
as 
\[ \forall W \in\{X,Z\},\,\forall a \in \Fq,\qquad\qp_W(a) = \sum_{e \in \Fq}  \omega^{\tr(a e)} \qp_W^e\; . \]

For choices of $q$ such that $\Fq$ admits a self-dual basis $(b_1,\ldots,b_t)$, we can decompose a $q$-dimensional
qudit (a ``quqit'') as a tensor product of $t$ $p$-dimensional qudits
(``qupits''). Based on this decomposition, for $W\in\{X,Z\}$ and $\ell\in\{1,\ldots,t\}$ we define the $W$-basis Pauli operator
acting on the $\ell$-th qupit by 
\begin{equation}\label{eq:pauli-l}
\forall  a \in  \Fp,\qquad  \sigma_{W,\ell}(a) \,=\, \sum_{e_1, \dots, e_t\in\Fp} \omega^{ae_\ell} \,\qp^{(e_1 b_1 +
    \dots + e_t b_t)}_W\,=\, \qp_W(a b_\ell) \;. 
		\end{equation}
It can be verified by direct computation that for every $\ell\in\{1,\ldots,t\}$, $\sigma_{X, \ell}$ and $\sigma_{Z, \ell}$ obey the Pauli twisted
commutation relations~\eqref{eq:twisted-fp}, and that when $\ell \neq \ell' \in \{1,\ldots,t\}$,
$\sigma_{X, \ell}$ and $\sigma_{Z, \ell'}$ commute. Both of these facts also follow from noting that the transformation $F$ that maps
$\ket{e_Z}$ to $\ket{e_X}$ decomposes as a tensor product over the qupits: 
\begin{align*}
  F &= \frac{1}{\sqrt{q}} \sum_{jk} \omega^{-\tr(jk)} \ket{j}\bra{k} \\
      &= \frac{1}{\sqrt{q}} \sum_{j_1, \dots j_t, k_1, \dots k_t} \omega^{- \tr(\sum_{\ell, \ell'}
        j_\ell k_{\ell'} b_\ell b_{\ell'})} \ket{j_1} \bra{k_1} \ot \cdots \ot
        \ket{j_t} \bra{k_t} \\
      &= \frac{1}{\sqrt{q}} \sum_{j_1, \dots j_t, k_1, \dots k_t} \omega^{ -\sum_{\ell, \ell'}
        j_\ell k_{\ell'} \tr(b_\ell b_{\ell'})} \ket{j_1} \bra{k_1} \ot \cdots \ot
        \ket{j_t} \bra{k_t} \\ 
      &= \frac{1}{\sqrt{q}}\sum_{j_1, \dots j_t, k_1, \dots k_t} \omega^{- \sum_{\ell}
        j_\ell k_{\ell}} \ket{j_1} \bra{k_1} \ot \cdots \ot
        \ket{j_t} \bra{k_t} \\ 
      &= \bigotimes_{\ell = 1}^{t} \Big( \frac{1}{\sqrt{p}}\sum_{j_\ell, k_\ell}  \omega^{-j_\ell
        k_\ell} \ket{j_\ell} \bra{k_\ell} \Big)\;,
\end{align*}
where in going from the second to the third line we used the linearity of
the trace and the fact that $j_\ell, k_{\ell'}$ are elements of the prime
subfield $\Fp$. We will sometimes
consider the case where $p = 2$, in which case the $\sigma_{W,\ell}$ behave as the standard
Pauli spin matrices acting on $t$ qubits, with the index $\ell$ labeling the qubit
acted on. Also, it will be sometimes useful to allow the index $a$ to range over all of $\Fq$ instead of just $\Fp$; extending~\eqref{eq:pauli-l} we define $\sigma_{W,\ell}(a)$ to be $\qp_W(ab_\ell)$ for any $a \in \Fq$.

For systems with many qudits, we will consider tensor products of the operators $\qp_W$. Slightly abusing notation, for $W \in \{X, Z\}$ and $a \in \Fq^n$ we
denote by $\qp_W(a)$ the tensor product $\qp_W(a_1) \ot \dots \ot
\qp_W(a_n)$. These obey the twisted commutation relations
\[ \forall a,b\in \Fq^n,\qquad \qp_X(a) \qp_Z(b) \,=\, \omega^{-\tr(a \cdot b)} \qp_Z(b) \qp_X(a)\;, \]
where $a \cdot b = \sum_{i=1}^{n} a_i b_i \in \Fq$. For $W\in\{X,Z\}$ and $e \in \Fq^n$
define the eigenstates
\[ \ket{e_W} = \ket{(e_1)_W} \ot \dots \ot \ket{(e_n)_W}\;, \]
and associated rank-$1$ projectors $\qp_W^e$. 


\paragraph{State-dependent distance.}
For operators $A,B\in\Lin(\mH)$, where $\mH$ is a finite-dimensional Hilbert space, and a vector $\ket{\psi} \in \mH\otimes \mH'$, where $\mH'$ is another finite-dimensional Hilbert space, we write $A\approx_\delta B$ for $\| (A-B)\otimes \Id \ket{\psi} \|^2 = O(\delta)$. Note the state $\ket{\psi}$ and the space $\mH'$ are usually kept implicit. We sometimes write the same with some free variables, e.g. $A_x^a \approx_\delta B_x^a$. By this we mean 
$$\Es{x}\sum_a \| (A_x^a-B_x^a)\otimes \Id \ket{\psi}\|^2 \,=\, O(\delta)\;.$$
 Variables appearing as subscript will most often be considered ``inputs'', and
 should be averaged; superscripts are considered ``answers'' and should be
 summed over. Which is which will always be clear from context, including the
 distribution on inputs.

 For a family of POVM $\{A_x^a\}$ acting on $\mH_A$, we
 will say that $\{A_x^a\}$ is $\delta$-self-consistent if there exists a family
 of POVM $\{\comp{A}_x^a\}$ acting on $\mH_B$ such that $A_x^a \otimes \Id_B
 \approx_\delta \Id_A \otimes \comp{A}_x^a$.\footnote{We combine the use of bold
   font together with the ``{\color{MidnightBlue} MidnightBlue}'' color to
   (hopefully) make it easier to distinguish the operators, as well as make
   (online) reading of our paper a more colorful experience.}  Note that this definition relies on an implicit understanding of the space $\mH_B$ and the operators  $\{\comp{A}_x^a\}$, and we will only use the terminology when the space and operators are clear from context. The following lemma relates two measures of consistency, defined via observables or the underlying projective measurement. 

\begin{claim}\label{claim:obs-meas-cons}
 Let $s$ be any integer, and let $\{A^a\}, \{B^b\}$ be projective measurements with outcomes
 $a,b\in\Z_s$. Let $A = \sum_a \omega_s^a A^a$ and $B= \sum_b \omega_s^b B^b$ where $\omega_s = \exp(2\pi i / s)$. Then for any state $\ket{\psi}$, 
\begin{align*}
\frac{1}{2}\Big(1-\Re\big(\bra{\psi} A\otimes B^\dag \ket{\psi}\big)\Big)\,\leq\,\sum_{a\neq b} \big\| A^a\otimes B^b \ket{\psi}\big\|^2  \,\leq\, \frac{s^2}{2\pi^2}\Big(1-\Re\big(\bra{\psi} A\otimes B^\dag \ket{\psi}\big)\Big)\;.
\end{align*}
\end{claim}


\begin{proof}
Expand
\begin{align*}
 \big\|(A\otimes \Id - \Id \otimes B) \ket{\psi} \big\|^2 
&= \sum_{a\neq b} \big|\omega_s^a-\omega_s^b\big|^2 \big\|A^a \otimes B^b \ket{\psi}\big\|^2\;,
\end{align*}
and use $\frac{2\pi}{s} \leq |\omega_s^a - \omega_s^b|\leq 2$ for all $a\neq b$. 
\end{proof}

\begin{claim}
  Let $s$ be an integer, $\delta\geq 0$, $\ket{\psi}\in\mH$ a state and $B\in\Lin(\mH)$ a normal operator such that $\| (B^s - \Id)
  \ket{\psi} \|^2 \leq \delta$. Then there exists a unitary $U$ with the same
  eigenvectors as $B$ such that $U^s=\Id$ and 
	$$\big\| (B  - U) \ket{\psi}\big\|^2
  \,\leq \,4\,\delta\;.$$
  \label{claim:round_to_roots}
\end{claim}

\begin{proof}
		Let $\lambda = re^{2i\pi\theta}$ be any complex number, where $r\in \R_+$ and $\theta\in[-1/2,1/2)$. Then
		\begin{align}
		\big|re^{2i\pi\theta}-1\big| &\geq \max\big(|r-1|,\, \big|e^{2i\pi\theta}-1\big|\big)\notag\\
		&\geq \max\big(|r-1|,\, 4|\theta|\big)\;.\label{eq:angle-bound}
		\end{align}
Thus, by the triangle inequality,
\begin{align}
\Big| \lambda - e^{\frac{2i\pi}{s} \lfloor s \theta \rfloor}\Big|
		&\leq |r-1| +\Big|e^{2i\pi \theta }- e^{\frac{2i\pi}{s} \lfloor s \theta \rfloor}\Big|\notag\\
		&\leq |r^s-1| + 8\Big| \theta - \Big[ \frac{1}{s} \lfloor s \theta \rfloor\Big]_1 \Big|\notag\\
		&\leq \big| \lambda^s -1 \big| + \frac{2}{s} \big| \lambda^s -1 \big|\;,\label{eq:angle-bound-1}
\end{align}	
where in the second line we wrote $[x]_1$ for the representative of $x\bmod 1$ in $[-1/2,1/2)$, and the last line follows from~\eqref{eq:angle-bound}. 

  Write the eigendecomposition of $B = \sum_i \lambda_i \Pi_i$, where $\Pi_i$ is
  a Hermitian projector and $\lambda_i$ a complex number. We include zero
  eigenvalues, so that $\sum_i \Pi_i = \Id$. The assumption made in the claim can be written as
  \begin{equation}\label{eq:angle-bound-0}
	\| (B^s - \Id) \ket{\psi} \|^2 = \Big\| \sum_i (\lambda_i^s - 1) \Pi_i
    \ket{\psi} \Big\|^2 = \sum_i | \lambda_i^s - 1|^2 \|\Pi_i \ket{\psi}\|^2 \leq
    \delta\;. 
		\end{equation}
  Let $\omega = e^{\frac{2\pi i}{s}}$. For each $i$, let $\omega^{a_i}$ be the closest $s$-th root of unity to
  $\lambda_i$.  Define $U = \sum_i
  \omega^{a_i} \Pi_i$. Then
  \begin{align*}
    \| (B - U) \ket{\psi} \|^2 &=  \sum_i \big|\lambda_i - \omega^{a_i}\big|^2 \|\Pi_i
                                 \ket{\psi} \|^2 \\
    &\leq \sum_i \,4\,\big| \lambda_i^s - 1\big|^2 \| \Pi_i \ket{\psi} \|^2 \\
    &\leq 4\,\delta\;,
  \end{align*}
where the second line uses~\eqref{eq:angle-bound-1}, and the last is
by~\eqref{eq:angle-bound-0}. 
\end{proof}

\subsection{Self-testing}

We use the language of multi-player self-tests (we will often call the players ``provers'' as well). 

\begin{definition}
 Let $k\geq 1$ be an integer. A \emph{$k$-partite strategy} $S = (\ket{\psi},\mathcal{X},
  \mathcal{A}, \mathcal{M})$ consists of finite question and answer sets $\mathcal{X}=X_1\times\cdots\times X_k$
  and $\mathcal{A}=A_1\times\cdots\times A_k$ respectively, a $k$-partite
   quantum state $\ket{\psi} \in \mathcal{H}_1\otimes \cdots \otimes\mH_k$, and for each $i\in\{1,\ldots,k\}$ a collection of measurement
  operators $\{M^a_x\}_{a \in A_i}$ on $\mathcal{H}_i$ and indexed by $x \in X_i$.\footnote{Although this is left implicit in the notation, the measurement operators associated with different spaces need not be equal.} We say that the strategy is \emph{partial} if it only specifies measurement operators for a subset of the possible questions, or if it does not specify a state $\ket{\psi}$.  
\end{definition}

We reproduce a standard definition in self-testing. 

\begin{definition}\label{def:self-test}
A \emph{$k$-player self-test} with completeness $c$ and robustness $\delta(\eps)$ for a (partial) strategy $S=(\ket{\Psi},\mathcal{X},
  \mathcal{A}, \mathcal{M})$ is a distribution $\pi$ on $\mX$ and a family of coefficients $V(a_1,\ldots,a_k|x_1,\ldots,x_k)\in [0,1]$, for $(x_1,\ldots,x_k)\in\mX$ and $(a_1,\ldots,a_k)\in\mA$, such that the following hold:
	\begin{itemize}[nolistsep]
	\item There exists a strategy $\hat{S}$ that extends the (partial) strategy $S$ and succeeds in the test with probability at least $c$; formally, 
	$$ \sum_{(x_1,\ldots,x_k)} \pi(x_1,\ldots,x_k) \sum_{a_1,\ldots,a_k} V(a_1,\ldots,a_k|x_1,\ldots,x_k) \,\bra{\psi} \hat{M}_{x_1}^{a_1}\otimes \cdots \otimes \hat{M}_{x_k}^{a_k} \ket{\psi} \,\geq\, c\;.$$
	\item Any strategy with success at least $c-\eps$ in the test must be $\delta(\eps)$-close to the optimal strategy. Formally, for any strategy $\hat{S}=(\ket{\hat{\psi}},\mathcal{X},
  \mathcal{A}, \mathcal{\hat{M}})$ such that 
		$$ \sum_{(x_1,\ldots,x_k)} \pi(x_1,\ldots,x_k) \sum_{a_1,\ldots,a_k} V(a_1,\ldots,a_k|x_1,\ldots,x_k) \,\bra{\hat{\psi}} \hat{M}_{x_1}^{a_1}\otimes \cdots \otimes \hat{M}_{x_k}^{a_k} \ket{\hat{\psi}} \,\geq\, c-\eps\;,$$
		there exists a local isometry $\Phi=\Phi_1\otimes\cdots\Phi_k$ and a state $\ket{\aux}$ such that 
		$$ \big\| \Phi(\ket{\hat{\psi}}) - \ket{\aux}\ket{\psi}\big\| \,\leq\,\delta(\eps)\;,$$
		and
		$$ \sum_{x_1,\ldots,x_k} \pi(x_1,\ldots,x_k) \sum_{a_1,\ldots,a_k}\, \big\| \Phi\big( \hat{M}_{x_1}^{a_1} \otimes \cdots\otimes \hat{M}_{x_k}^{a_k} \ket{\hat{\psi}}\big) - \ket{\aux} M_{x_1}^{a_1}\otimes \cdots \otimes M_{x_k}^{a_k} \ket{\psi} \big\| \,\leq\, \delta(\eps)\;.$$
	\end{itemize}
	In case $S$ only specifies a partial strategy, then the above expression is restricted to questions for which $S$ is defined.
\end{definition}

\subsection{The commutation test}

In designing self-tests, it is useful to have the ability to test commutation
relations between pairs of observables applied by the provers. The following
well-known test can be employed to certify that two observables commute:

\begin{theorem}\label{thm:com_test}
  Let $s$ be an integer and $\eps > 0$. There exists a two-player self-test $\COM(M,N)$ with
  completeness $1$ and robustness $\delta(\eps) = O(s\sqrt{\eps})$, for the
  (partial) strategy $S$ that uses commuting generalized observables $M$ and $N$
  (with outcomes in $\Z_s$) for two special questions labelled $1$ and $2$,
  respectively. The test has $3$ questions per player and answers either in
  $\Z_s$ (for questions $1$ and $2$) or $\Z_s^2$ (for question $3$). Moreover, for any two commuting observables $A$ and $B$, there
  exists a strategy in which the first player uses the observables $M$ and $N$ for questions
  $1$ and $2$, using a shared state $\ket{\psi}$ that is a 
  maximally entangled state of appropriate dimension.
\end{theorem}

The guarantees of the theorem are achieved by the following test, which is a
simple instance of the idea of ``oracularization'' in multiprover interactive
proofs. In the test, the verifier performs either of the following with equal probability $\frac{1}{2}$:

\begin{enumerate}
\item Send the first player a question $q$ chosen
  uniformly from
  $\{1, 2\}$, and send the second player the question $3$. Receive an answer $a
\in \Z_s$ from the first player and $(b_1, b_2) \in \Z_s^2$ from the second
player. Accept if $a = b_q$, and reject otherwise.
\item Perform the same as in item 1., but
  with the players interchanged.
\end{enumerate}

The analysis of this test is standard; see, e.g.~\cite[Lemma 28]{coladangelo2017verifier}.

\subsection{The generalized Magic Square}

In~\cite{ColadangeloS17MS} a generalized version of the Magic Square game~\cite{Arvind:02} is introduced and shown to robustly self-test generalized observables satisfying twisted commutation relations over $\Z_s$, for any integer $s$. 

\begin{theorem}[Theorem 5.9 in~\cite{ColadangeloS17MS}]\label{thm:ms-rigid}
Let $s$ be an integer and $\eps>0$. There exists a two-player self-test $\MS(X,Z)$, with completeness $1$ and robustness $\delta(\eps) = O(s^3\sqrt{\eps})$, for the (partial) strategy $S$ that uses observables $\sigma_{X}$ and $\sigma_{Z}$ on two special questions labeled $X$ and $Z$ respectively. The test has $O(1)$ questions per player (including two questions labeled $X$ and $Z$) and answers in $\Z_s^2$. Furthermore, there is a strategy that succeeds with probability $1$ using only $\sigma_X$, $\sigma_Y$ and $\sigma_Z$ observables on two $s$-dimensional qudits per player initialized in $\ket{\psi} = \ket{\epr_s}\otimes \ket{\epr_s}$. 
\end{theorem}

\subsection{The classical low-degree test}

A stepping stone in our analysis is an extension of the ``classical low-degree test'' from~\cite{Vidick13xor} to the case of only two provers. 

\begin{theorem}[Theorem 2 in \cite{NatarajanV17twoprover}]\label{thm:ml}
Let $\eps > 0$, $m,d$ integers, and $q$ a prime power such that $q \geq
(dm/\eps)^{c}$ for a universal constant $c\geq 1$.
There is a two-prover test, called the \emph{classical low-degree test} $\cld(m,d,q)$, in which queries to the provers are chosen among affine subspaces $s\subseteq \Fq^m$, and answers are polynomials $r$ on $s$ of total degree at most $d$, such that the following holds. For any strategy for the provers using entangled state $\ket{\psi}$ and projective measurements $\{M_s^r\}$ that succeeds with probability at least $1-\eps$ in the test there exists a POVM $\{S^g\}$, where $g$ ranges over the polynomials  on $\Fq^m$ of total degree at most $d$, and a $\delta = \poly(\eps)$ such that the following hold:
  \begin{enumerate}
  \item Approximate consistency with $M$: 
    \[ \Es{s}\, \sum_{g} \sum_{r \neq g|_s} \bra{\psi} M_s^r
      \ot S^g \ket{\psi} \,\leq\, \delta, \]
  \item Self-consistency: 
	$$\sum_g \bra{\psi}  S^g  \ot (\Id - S^g) \ket{\psi} \,\leq\, \delta.$$
  \end{enumerate}
\end{theorem}

We let $\pild$ denote the distribution on questions used by the verifier in the low-degree test from Theorem~\ref{thm:ml}. This distribution is symmetric, and we slightly abuse notation by also writing $\pild$ for either marginal. We will  use that the test from Theorem~\ref{thm:ml} that it satisfies the following properties:  
\begin{itemize}[nolistsep]
\item[(i)] $\pild$ is a uniform mixture of the uniform distribution on pairs $(s,w)$ such that $s$ is an affine subspace of dimension $2$ in $\Fq^m$ and $w\in s$ is a uniformly random point in $s$, and its permutation $(w,s)$.  
\item[(ii)] Whenever provers in the test are queried for a pair of subspaces
  $(s,w)$, they are required to return a polynomial $r$  defined on $s$ and a value $a$ in $\Fq$ such that $r(w)=a$. 
	\end{itemize}

Theorem~\ref{thm:ml} assumes that the strategy employed by the provers in the test is invariant under permutation of the two provers. It will be convenient to allow non-symmetric strategies as well. 

\begin{corollary}\label{cor:ml}
Let $m,d,q$ and $\eps$ be as in Theorem~\ref{thm:ml}. Let $\ket{\psi} \in \mH_A
\otimes \mH_B$ be a bipartite state, and $\{M_s^r\}$ and $\{\comp{M}_s^r\}$ be
POVMs on $\mH_A$ and $\mH_B$ respectively, such that the associated strategy for
the provers succeeds in the test $\cld(m,d,q)$ from Theorem~\ref{thm:ml} with probability at
least $1-\eps$. Then there exist POVMs $\{S^g\}$ and $\{\comp{S}^g\}$, where $g$
ranges over the polynomials  on $\Fq^m$ of total degree at most $d$, defined on
$\mH_A$ and $\mH_B$ respectively, and a $\delta = \poly(\eps)$ such that the
following relations hold, on average over $s\sim\pild$:
$$\sum_{g}\,M_s^{ g_{|s}} \ot \comp{S}^g \,\approx_\delta\,\Id\;,\qquad \sum_{g}\,S^g \ot \comp{M}_s^{ g_{|s}}\,\approx_\delta\,\Id\;,$$
and
$$ S^g  \ot \Id \,\approx_\delta \, \Id\otimes \comp{S}^g\;.$$
 \end{corollary}

\begin{proof}
Extending $\reg{A}$ or $\reg{B}$ as needed, assume without loss of generality that $\mH_{\reg{A}}$ and $\mH_{\reg{B}}$ have the same dimension, and fix a canonical isomorphism between the two. Adjoin ancilla spaces $\mH_{\reg{A'}}$ and $\mH_{\reg{B'}}$, each isomorphic to $\C^2$. 
From an arbitrary strategy we can construct a symmetric one by letting 
$$\hat{M}_s^r \,=\, M_s^r \otimes \proj{0}_{\reg{A'}} + \comp{M}_s^r \otimes \proj{1}_{\reg{A'}}\;,$$
and 
$$\ket{\hat{\psi}} \,=\, \frac{1}{\sqrt{2}} \big( \ket{\psi}_{\reg{AB}} \otimes \ket{0}_{\reg{A'}} \otimes \ket{1}_{\reg{B'}} +\ket{\psi'}_{\reg{AB}} \otimes \ket{1}_{\reg{A'}} \otimes \ket{0}_{\reg{B'}} \big)\;,$$
where $\ket{\psi'}_{\reg{AB}}$ is obtained by swapping registers $\reg{A}$ and $\reg{B}$ in $\ket{\psi}_{\reg{AB}}$. Using that the test from Theorem~\ref{thm:ml} is symmetric, the success probability of this strategy is the same as that of the non-symmetric one. Applying Theorem~\ref{thm:ml} gives POVM $\{\hat{S}^g\}$ defined on $\reg{AA'}$ that are consistent with the $\{\hat{M}_s^r\}$, on the state $\ket{\hat{\psi}}$. It then suffices to define 
$$S^g \,=\, \big(\Id \otimes \bra{0}_{\reg{A'}}\big)\,S^g\,\big(\Id\otimes \ket{0}_{\reg{A'}}\big)\;,\qquad \comp{S}^g \,=\, \big(\Id \otimes \bra{1}_{\reg{A'}}\big)\,S^g\,\big(\Id\otimes \ket{1}_{\reg{A'}}\big)\;.$$
\end{proof}

The length of questions in the low-degree test $\cld(m,d,q)$ from Theorem~\ref{thm:ml} is $O(m\log q)$, which for a choice of $q=\poly\log(n)$ is logarithmic in $n$. However, answers have length $O(d^2 \log q)$, which is super-logarithmic. To achieve reduced answer length it is standard to compose the test with itself: any answer $r$ from a prover is interpreted as an $n'=O(d^2 \log q)$-long string of bits, that can be encoded as a multilinear polynomial over $\Fq^{m'}$, for $m'$ such that $2^{m'} \geq n'$. Questions in the composed test are a subspace $s\subseteq\Fq^m$, together with a subspace $s'\subseteq \Fq^{m'}$, and answers are the restriction to $s'$ of the low-degree encoding of the polynomial $r$ that the prover would answer to the question $s$. The analysis of the composition is standard, and we state the result as the following theorem. 

 \begin{theorem}\label{thm:2ml}
Let $\eps > 0$, $m,d$ be integers, and $q$ a prime power such that $q \geq
(dm/\eps)^{c}$ for a universal constant $c\geq 1$. Let $n' = O(d^2)$ be the answer length (in number of $\Fq$-symbols) in $\cld(m,d,q)$, and $m'=\log(n')=O(\log\log n)$.  
There is a two-prover test, called the \emph{composed classical low-degree test} $\cld^{(2)}(m,d,q)$, in which queries to the provers are chosen among, either pairs of affine subspaces $(s,s')\subseteq \Fq^m\times \Fq^{m'}$, or points in $\Fq^m$, and answers are, either multilinear polynomials $r'$ on $s'$, or values $a\in\Fq$, such that the following holds. For any strategy specified by a shared state $\ket{\psi} \in \mH_A
\otimes \mH_B$ and measurement operators $\{M_{s,s'}^{r'}\}$ and $\{\comp{M}_{s,s'}^{r'}\}$ on $\mH_A$ and $\mH_B$ respectively, such that the associated strategy succeeds in the test $\cld^{(2)}(m,d,q)$ with probability at least $1-\eps$, there exist POVM $\{S^g\}$ and $\{\comp{S}^g\}$, where $g$
ranges over the polynomials  on $\Fq^m$ of total degree at most $d$, defined on
$\mH_A$ and $\mH_B$ respectively, and a $\delta = \poly(\eps)$ such that the
following relations hold, on average over $s\sim\pild$:
$$\sum_{g}\Es{s'} \,M_{s,s'}^{ g_{|s,s'}} \ot \comp{S}^g \,\approx_\delta\,\Id\;,\qquad \sum_{g}\Es{s'}\,S^g \ot \comp{M}_{s,s'}^{ g_{|s,s'}}\,\approx_\delta\,\Id\;,$$
where the expectation is over an $s'$ as sampled in the test (conditioned on $s$), and $g_{|s,s'}$ denotes the polynomial on $s'$ obtained by restricting to $s'$ the low-degree extension of the description of the restriction $g_{|s}$ of $g$ to $s$. Furthermore,
$$ S^g  \ot \Id \,\approx_\delta \, \Id\otimes \comp{S}^g\;.$$
\end{theorem}

\section{The quantum low-degree test}
\label{sec:qld}

\subsection{Description of the test}
\label{sec:lowdeg-protocol}

\begin{figure}[htbp]
\rule[1ex]{16.5cm}{0.5pt}\\
Test~$\qld^{(l)}(m,d,q)$. $m,d$ are integer, and $q=p^t$ is a prime power such that $\Fq$ admits a self-dual basis $(b_1,\ldots,b_t)$ over $\Fp$. $l\in\{1,2\}$ is a parameter that indicates the level of the test.\\
The verifier performs the following with equal probability:
\begin{enumerate}
\item[(a)] Select $W\in\{X,Z\}$ uniformly at random and send $W$ to both provers. If $l=2$ execute the test $\cld^{(2)}(m,d,q)$ from Theorem~\ref{thm:2ml} with the provers. If $l=1$ execute the test $\cld(m,d,q)$ from Theorem~\ref{thm:ml}. Let $r$ be the polynomial returned by the first prover, and $r'$ by the second. If $W=X$, set $A=r$ and $\comp{A'}=-r'$. If $W=Z$, set $A=r$ and $\comp{A'}=r'$. Accept if and only if the pair of answers $(A,\comp{A'})$ would have been accepted in the classical test. 
\item[(b)]
Select $x,z\in\Fq^m$ and $u, u' \in \Fq$ uniformly at random,
and let  $a = \tr((u x_\bij) \cdot
  (u' z_\bij))  \in \Fp$.
	\begin{itemize}
	\item If $a = 0$, execute the self-test $\COM$ (see Theorem~\ref{thm:com_test}), replacing
  queries 1, 2, and 3 in the test  by $(X, x)$, $(Z,z)$, and $(x,z,uu')$ respectively, and
  in the case of queries 1 and 2,  replacing the prover's answer $b \in \Fq$ by
  $\tr(u b)$  or $\tr(u' b) \in \Fp$, respectively, before making the same decision as the verifier in the test.    
	\item If $a \neq
  0$, execute the self-test 
  $\MS$ (see Theorem~\ref{thm:ms-rigid}) with the following modification: the
  question labeled $X$ is replaced by the query $(X, x)$ 
  as in part (a), and the prover's answer $b\in \Fq$ is replaced by $\tr(u
  b)\in \Fp$; the question labeled $Z$ is replaced by the query 
  $(Z, z)$ as in part (a), and the prover's answer $b\in \Fq$ is replaced by
  $a^{-1}\tr(u' b)\in \Fp$. 
	\end{itemize}
\end{enumerate}
\rule[1ex]{16.5cm}{0.5pt}
\caption{The quantum low-degree test. $l\in\{1,2\}$ denotes the ``level'' of the test, before ($l=1$) or after ($l=2$) composition.}
\label{fig:protocol}
\end{figure}

We denote our quantum low-degree test by $\qld^{(l)}$, for $l\in\{1,2\}$. Here $l$ 
denotes the ``level'' of the test, before ($l=1$) or after ($l=2$) composition. In general we also write $\qld$ for the ``composed quantum low-degree test'' $\qld^{(2)}$, which is the variant of the test with reduced answer size, and is the variant that will be used in our applications. 
The test is described in \figref{protocol}. We show that the test is a self-test for the following class of Pauli strategies. To define the strategy, recall the definition of the POVM  $\{\qp_W^a\}$ in Section~\ref{sec:qauli}, defined for each $W\in \{X,Z\}$. For $s \subset \Fq^m$ either a point or a $2$-dimensional subspace, and $r$ a polynomial defined on $s$, define
\begin{equation}\label{eq:honest-m-def}
 \qp_{W,s}^r = \sum_{a\in \Fq^n:\, (g_a)_{|s}=r}\qp_W^a\;,
\end{equation}
where $g_{a}$ is defined in~\eqref{eq:def-ga}. Finally, for reasons that will become clear later, it is convenient to introduce
\begin{equation}\label{eq:honest-m-def-2}
\comp{\qp}_{X,s}^r = \qp_{X,s}^{-r}\qquad\text{and}\qquad\comp{\qp}_{Z,s}^{r} =
\qp_{Z,s}^r\;.
\end{equation} 

\begin{definition}\label{def:pauli-strategy}
Let $p$ be a prime, $t\geq 1$ an integer, and $q=p^t$. 
The low-degree Pauli strategy $S_{\textrm{P}}$ on $n$ qudits of local dimension $q$ is
the strategy $(\ket{\psi}, \cX, \cA, \cM)$ where $\ket{\psi} =
\ket{\epr_q}^{\otimes n}$, $\cX = \{X, Z\} \times  (\cX_1 \cup \cX_2)$, where $\cX_1 = \Fq^m$ and $\cX_2$ is the set of all two-dimensional subspaces of $\Fq^m$, $\cA = \cA_1 \cup \cA_2$, where $\cA_1 = \Fq$ and $\cA_2 = \deg_d(\Fq^2)$, and
$\cM = \cM_1 \cup \cM_2$, where $\cM_1 = \{\qp_{W,w}^a\}\times\{\comp{\qp}_{W,w}^a\}$
 and $\cM_2 =\{\qp_{W,s}^r\}\times\{\comp{\qp}_{W,s}^r\}$, with $\qp_{W,w}^a$,  $\qp_{W,s}^r$, and $\comp{\qp}_{W,w}^a$,  $\comp{\qp}_{W,s}^r$ defined as in~\eqref{eq:honest-m-def} and~\eqref{eq:honest-m-def-2} respectively.
\end{definition}

\begin{theorem}\label{thm:qld}
Let $n\geq 1$ be an integer. Let $h,m$ be integer such that $h^m \geq n$, and let $d=hm$. Let $q = p^t$ be a prime power such that $\Fq$ admits a
self-dual basis over $\Fp$. Then for any $\eps\geq 0$ the test $\qld^{(2)}(m, d, q)$ is a $2$-prover self-test for the low-degree Pauli
strategy $S_{\textrm{P}}$ on $n$ qudits of local dimension $q$ with completeness $1$ and robustness $\delta = \poly(\poly(p)\cdot \poly(\eps) +
\poly(d/q))$. Moreover, 
the test has questions of length $O(m \log q)$ and answers of length
$O(\log^2(d) \log (q))$.

\end{theorem}

 Completeness of the test is shown  in \lemref{completeness} in
 Section~\ref{sec:lowdeg-completeness}. Soundness is shown in
 \lemref{soundness} in Section~\ref{sec:soundness}.

\begin{remark}
  In a typical application of the test $\qld^{(2)}$, the parameters
  are chosen such that $m = \Theta(\frac{\log n}{\log \log n})$ and $h =
  \Theta(\log n)$, resulting in $d = \Theta(\frac{\log^2(n)}{\log \log
    n})$. Further, we chose $p$ to be constant and $q = \Theta(\frac{\log^2(n)}{\log \log
    n})$ such that $d/q$ is a small constant. This results in a
  question length that is $O(\log n)$ and an answer length that is
  $\poly(\log\log n)$.
\end{remark}

\subsection{Completeness}
\label{sec:lowdeg-completeness}

The proof of the following lemma specifies the ``honest'' strategy that is expected of the provers in the quantum low-degree test. 

\begin{lemma}[Completeness]\label{lem:completeness}
For $m,d,q$ as in Theorem~\ref{thm:qld} the strategy $S_{\textrm{P}}$ introduced in Definition~\ref{def:pauli-strategy} can be extended to a strategy that succeeds with probability $1$ in the test $\qld(m,d,q)$. 
\end{lemma}

\begin{proof}
Let 
$$\ket{\psi_\epr} \,=\, \bigotimes_{j=1}^{n+1} \ket{\epr_q}\;,$$
where $\ket{\epr_q}$ is defined in~\eqref{eq:def-epr-p}. 
We first describe a strategy for the players assuming questions in part (a) of the test come from $\cld$, instead of the composed test $\cld^{(2)}$. Once a strategy for the former has been defined it is straightforward to adapt it to a strategy for the latter; this only requires classical post-processing. 

To define the strategy
 we use the generalized Pauli operators and
projections defined in Section~\ref{sec:qauli}. When queried for a subspace $s\subseteq
\Fp^m$ in a basis $W\in\{X,Z\}$, the prover measures the first $n$ qudits using
the projective measurement $\{\qp_W^a\}$ and returns the polynomial $(g_a)_{|s}$; this corresponds to the POVM described in~\eqref{eq:honest-m-def}. 

To see that these measurements define a strategy which succeeds with probability
$1$ in part (a) of the test, note that the state $\ket{\epr_q}$ is stabilized by
$\qp_X(a) \ot \qp_X(a)$ and $\qp_Z(b) \ot
\qp_Z(-b)$ for any $a, b \in \Fq$. Hence, if both provers measure the
state $\ket{\epr_q}$ in the $X$
eigenbasis, and the first prover obtains an outcome $a\in\Fq$, the second prover will obtain the outcome $-a$; if they 
 measure in the $Z$ eigenbasis, they will both always obtain the same
outcome. As a consequence, the following consistency relations hold for any $s$:
\begin{equation}
\begin{aligned}
  \sum_{r\in \deg_d(s)} \,\qp_{X,s}^r \ot \qp_{X,s}^{-r} \,\ket{\psi_\epr} &= \ket{\psi_\epr}\;,
  \\
  \sum_{r\in\deg_d(s)} \,\qp_{Z,s}^r \ot \qp_{Z,s}^{r} \,\ket{\psi_\epr} &= \ket{\psi_\epr}\;.
\end{aligned}
\label{eq:epr_stab-0}
\end{equation}
Thus whenever $W=X$ is selected in part (a) of the low-degree test the first prover's answers are consistent
with the negation of the second prover's, as the verifier expects; in case $W=Z$ both provers' answers are consistent. 

Using the notation introduced in~\eqref{eq:honest-m-def-2}, the
consistency relations~\eqref{eq:epr_stab-0} become
\begin{equation}
\begin{aligned}
  \sum_q \qp_{W,s}^q \ot \comp{\qp}_{W,s}^q \ket{\psi_\epr} &=
  \ket{\psi_\epr}\;, \\
\end{aligned}
\label{eq:epr_stab}
\end{equation}
for any $W\in\{X,Z\}$.

To show completeness in part (b) of the
test we introduce a family of generalized observables associated with the measurement performed by a prover in part (a) of the test when it is queried for a value at a single point $w\in\Fq^m$. The prover's answer in this case is a
value in $\Fq$. To the provers' strategy for determining his answer we introduce a family of $q$
observables over $\Fp$, indexed by $u\in \Fq$, each of which is associated with the value $\tr(ub)\in \Fp$, where $b\in\Fq$ is the answer obtained by the prover. 
 We denote the corresponding for query $(W,w)$ by
$W_u(w_\bij)$:
\begin{align}
    W_u(w_\bij) &= \sum_{s \in \Fq^n}\, \omega^{ \tr(g_s(w) u)} \,\qp_W^a \label{eq:def-local-obs}\\
              &= \sum_{s \in \Fq^n} \,\omega^{\tr(s\cdot (u w_\bij))} \,\qp_W^a \notag\\
              &= \qp_W(u w_\bij)\;.\notag
\end{align}
From this expression it is clear that for any $x,z\in\Fq^m$, $u,u'\in\Fq$, and $a=\tr((u x_\bij) \cdot
  (u' z_\bij))$,
 $$X_u(x_\bij) Z_{u'}(a^{-1} z_\bij) =
\omega^{-\tr(a^{-1} (u x_\bij)(u' z_\bij))} Z_{u'}(a^{-1} z_\bij) X_u(x_\bij) =
\omega^{-1} Z_{u'}(a^{-1} z_\bij) X_u(x_\bij)\;.$$

 Hence, the measurement operators
corresponding to the questions labeled $X$ and $Z$ in part (b) of the test satisfy the
required twisted commutation relation. It is then straightforward that each prover can implement a strategy that succeeds in part (b) of the test. In case the test $\COM$ is executed this is immediate; in case it is $\MS$ the provers may use the $(n+1)$-th qudit and a second pair of observables $(X',Z')$ satisfying the same twisted commutation relation, so that $(X,Z)$ and $(X',Z')$ together form a strategy which succeeds with probability $1$ in the test $\MS$.

%
\end{proof}

\section{Soundness analysis}
\label{sec:soundness}

\begin{lemma}[Soundness]\label{lem:soundness}
Let $n\geq 1$ be an integer and $m,h,d,$ and $q=p^t$ as in Theorem~\ref{thm:qld}. Let $\eps\geq 0$. Suppose a strategy using state
$\ket{\psi}_{\reg{AB}} \in \mH_{\reg{A}} \otimes \mH_{\reg{B}}$ and projective
measurements $\{M_{W,s,s'}^r\}$ and $\{M_{W,w}^a\}$ succeeds in test $\qld(m,d)$ with probability at
least $1-\eps$. Then there is a $\delta = \poly(\poly(p)\cdot \poly(\eps) +
\poly(d/q))$, isometries $V_D: 
\mH_{\reg{D}} \to (\C^q)^{\otimes n}_{\reg{D'}}\otimes \mH_{\reg{D''}}$ for
$D\in \{A,B\}$, and a state $\ket{\aux}\in\mH_{\reg{A''}} \otimes \mH_{\reg{B''}}$  such that  
$$ \big\| V_A \otimes V_B\ket{\psi} -  \ket{\epr_q}^{\otimes n} \ket{\aux} \big\|^2 \,\leq\, \delta\;,$$
and for all $W\in \{X,Z\}$, 
\[
\Es{w\in\Fq^m} \,\sum_{a\in \Fq}\, \big\|(V_A \ot V_B)(M_{W,w}^a \ot \Id) \ket{\psi} -
  (\qp_{W,w}^a \ot \Id)\ket{\epr_q}^{\otimes n} \ket{\aux} \big\|^2 \, \leq \,
  \delta\;.
\]
Moreover, an analogous relation holds for the second prover's operators.
\end{lemma}

The outline for the proof of Lemma~\ref{lem:soundness} is as follows: 
\begin{enumerate}
\item In Section~\ref{sec:strat}, we describe the conditions satisfied by any strategy
  that succeeds with high probability in the test.
\item In Section~\ref{sec:hatx}, we adjoin an ancilla to each of the provers' private registers, and define a set of approximately commuting
  ``points'' observables $\hat{X}_u(x_\bij)$ and $\hat{Z}_u(z_\bij)$, for each $u\in\Fq$, that act on the
  original shared state tensored with a maximally entangled state on the ancilla.
\item In Section~\ref{sec:hatq}, we construct a family of joint measurements
  $\{\hat{Q}_{x,z}^{c}\}$ indexed by a pair of values $x, z \in \Fq^m$ and with
  outcomes $c \in \Fq$, obtained as a common refinement of the $2q$ approximately commuting observables
  $\hat{X}_{u}(x_\bij)$ and $\hat{Z}_v(z_\bij)$, for all $u,v\in\Fq$, defined in the previous step. Joint measurability is proved by showing that the observables satisfy approximate linearity relations, in the sense of implying a successful strategy in the two-prover linearity test from~\cite{NV17}.  
\item From these joint measurements we define a strategy for the test $\cld(2m+2,d+1,q)$ over $\Fq^m\times\Fq^m \times \Fq^2$, denoted by $\{\hat{Q}_{s}^r\}$. Applying Theorem~\ref{thm:2ml}, we deduce a single low-degree measurement $\{\hat{S}^g\}$.
\item We argue that $g$ must take the form $g(x,z,\alpha,\beta)=\alpha g_1(x) + \beta g_2(z)$ for low-degree polynomials $g_1,g_2:\Fq^m\to \Fq$. This allows us to recover commuting low-degree measurements $\{\hat{S}_X^{g_1}\}$ and $\{\hat{S}_Z^{g_2}\}$.
\item In Section~\ref{sec:tildex} we use the low-degree measurements obtained in the previous step to recover observables
  $\tilde{X}_\ell(x)$ and $\tilde{Z}_{\ell'}(z)$ defined for all points $x,z\in\Fp^n$ and indices $\ell \in \{1,\ldots,t\}$. By construction these operators exactly satisfy
  the same twisted commutation relations as the ``honest'' generalized observables $\qp_X(b_\ell x)$ and $\qp_Z(b_{\ell'} z)$  (recall that $\{b_1,\ldots,b_t\}$ is a self-dual basis of $\Fq$ over $\Fp$); moreover, we
  show that they are consistent with the provers' original observables $X_{b_\ell} ,Z_{b_{\ell'}}$ at points
  of the form $x_\bij,z_\bij$. 
\item Finally, in Section~\ref{sec:finish_soundness}, we use that $\tilde{X}_\ell(x)\otimes \tilde{X}_\ell(x)$ and $\tilde{Z}_{\ell'}(z) \otimes \tilde{Z}_{\ell'}(z)$ approximately stabilize $\ket{\psi}$ to conclude that the
  provers' shared state is close to the target state $\ket{\epr_q}^{\otimes n}$ (under the action of the appropriate isometry).
\end{enumerate} 

\subsection{Arbitrary strategies in the test $\qld$}
\label{sec:strat}

We start with the following preliminary claim, which establishes basic properties of successful strategies in the test $\qld(m,d,q)$.

\begin{claim}\label{claim:strategies}
Let $m,d,q=p^t$, $\eps$, $\ket{\psi}$ and $\{M_{W,s,s'}^r\}$ be as in
Lemma~\ref{lem:soundness}. There exists $\delta_M = \poly(p)\cdot\poly(\eps)$ such that
the following hold. For $W\in\{X,Z\}$ and $s\subseteq \Fq^m$ 
there exist projective measurements $\{{M}_{W,s}^{r}\}_{r\in\deg_d(s)}$ and $\{\comp{M}_{W,s}^{r}\}_{r\in\deg_d(s)}$ such that, on average
over $s\sim\pild$, 
\begin{equation}\label{eq:subspace-consistent}
 M_{W,s}^r \otimes  \Id \,\approx_{\delta_M} \, \Id \otimes \comp{M}_{W,s}^{r} \;,
\end{equation}
and moreover the $\{{M}_{W,s}^{r}\}_{r\in\deg_d(s)}$ and $\{\comp{M}_{W,s}^{r}\}_{r\in\deg_d(s)}$, together with the state $\ket{\psi}$, specify a strategy with success $1-\eps'$ in the test $\qld^{(1)}(m,d,q)$, for some $\delta = \poly(p)\cdot\poly(\eps)$. 

For $W\in\{X,Z\}$, $u\in\Fq$, and $w\in \Fq^m$ define\footnote{The map $w\mapsto w_{\bij}$ from $\Fq^m$ to $\Fq^n$ is defined in Section~\ref{sec:notation}.}
\begin{equation}\label{eq:def-xz-obs}
W_u(w_\bij) \,=\, \sum_a \omega^{\tr(a u)} M_{W,w}^a\;, \qquad \comp{W}_u(w_\bij) \,=\, \sum_a
\omega^{\tr(a u)} \comp{M}_{W,w}^a\;.
\end{equation}
Then for fixed $W$ and $w$, the $q$ observables
$\{W_u(w_\bij),\,u\in\Fq\}$ pairwise commute. For any
$a\in\Fq$, we can write the POVM elements $M_{W,w}^a$ and $\comp{M}_{W,w}^a$ in terms of the observables $W_u(w_\bij)$ and $\comp{W}_u(w_\bij)$ as
follows:\footnote{Note that here, and elsewhere, the superscript denotes the
  outcome of the measurement, not exponentiation.}
\begin{equation}\label{eq:m-from-w}
 M_{W,w}^a \,=\, \Es{u\in\Fq} \,\omega^{-\tr(au)}\,W_u(w_\bij)\;,\qquad  \comp{M}_{W,w}^a \,=\, \Es{u\in\Fq} \,\omega^{-\tr(au)}\,\comp{W}_u(w_\bij)\;.
\end{equation}
Moreover, on average over $w\in\Fq^m$ and for every $u\in\Fq$,
\begin{equation}\label{eq:xz-cons}
W_u(w_\bij) \otimes \Id \,\approx_{\delta_M} \,\Id \otimes \comp{W}_u(w_\bij)\;.
\end{equation}
Finally,
\begin{equation}\label{eq:xz-ac}
 X_u(x_\bij) Z_{u'}(z_\bij) \,\approx_{\delta_M}\, \omega^{\tr((u
   x_\bij) \cdot (u' z_\bij))} Z_{u'}(z_\bij) X_u(x_\bij)\;,
\end{equation}
on average over uniformly random $x,z\in\Fq^m$ and $u,u' \in \Fp^t$.
\end{claim}

\begin{proof}
For $s\subseteq \Fq^m$, $s'\subseteq\Fq^{m'}$ and $r\in\deg_d(s)$ let 
$$\comp{M}_{X,s,s'}^{r} \,=\, M_{X,s,s'}^{-r}\qquad\text{and}\qquad \comp{M}_{Z,s,s'}^{r} \,=\, M_{Z,s}^{r}\;.$$
With this definition,~\eqref{eq:subspace-consistent} follows from the assumption that the provers' strategy succeeds with probability $1-O(\eps)$ in part (a) of the test $\qld^{(2)}$. Moreover, for any fixed choice of $W$ and first part $s$ of the question $(s,s')$ in $\cld^{(2)}(m,d,q)$, the induced strategy is a successful strategy in $\cld^{(1)}(m',d',q)$, to which Theorem~\ref{thm:ml} can be applied. This defines the required POVM elements $\{{M}_{W,s}^{r}\}_{r\in\deg_d(s)}$ and $\{\comp{M}_{W,s}^{r}\}_{r\in\deg_d(s)}$.

 Commutation of the $\{W_u(w_\bij),\,u\in\Fq\}$ follows since $\{M_{W,w}^a\}$ are projective measurements. Eq.~\eqref{eq:m-from-w} follows by expanding $W_u(w_\bij)$ using the definition.

Using that the distribution $\pild$ from the low-degree test places constant probability on subspaces of dimension $0$, by Claim~\ref{claim:obs-meas-cons} the consistency conditions~\eqref{eq:subspace-consistent} imply~\eqref{eq:xz-cons}. 

Finally, using Theorems~\ref{thm:com_test} and~\ref{thm:ms-rigid} success with probability at least $1-\eps$ in part (b) of the test implies that the observables defined in~\eqref{eq:def-xz-obs} satisfy~\eqref{eq:xz-ac} for $\delta_M = O(p^3\sqrt{\eps})$.
\end{proof}

\subsection{Expanding the Hilbert space and defining commuting observables}
\label{sec:hatx}

From the initial strategy of the provers, satisfying the properties expressed in Claim~\ref{claim:strategies}, we define new observables on an extended Hilbert space that will be the main operators used in the proof. 

\begin{lemma}\label{lem:hats}
Let $m,d,q=p^t$, $\eps$, $\ket{\psi}$ and $\{M_{W,s}^r\}$ be as in Lemma~\ref{lem:soundness}, and $W_u(w_\bij)$ as in Claim~\ref{claim:strategies}. There exists a state 
$$\ket{\hat{\psi}}_{\reg{AA'A''}\reg{BB'B''}} \in \mH_\reg{A} \otimes (\C^q_{\reg{A'}}\otimes \C^q_{\reg{A''}})^{\otimes n} \otimes \mH_\reg{B} \otimes (\C^q_{\reg{B'}}\otimes \C^q_{\reg{B''}})^{\otimes n}\;,$$
and for $W\in\{X,Z\}$, $s\subseteq \Fq^m$, $u \in \Fq$, and $w\in\Fq^m$ there
are POVM $\{\hat{M}_{W,s}^r\}_{r \in\deg_d(s)}$, $\{\comp{\hat{M}}_{W,s}^r\}_{r\in\deg_d(s)}$ and observables 
$\hat{W}_u(w_\bij)$, $\comp{\hat{W}}_u(w_\bij)$ on $\mH_\reg{A}\otimes
(\C^q_{\reg{A'}})^{\otimes n} $ and $\mH_\reg{B}\otimes
(\C^q_{\reg{B'}})^{\otimes n} $ respectively, 
\begin{equation}\label{eq:def-hat-xz}
\hat{W}_u(w_\bij) \,=\, W_u(w_\bij) \otimes \comp{\qp}_W(u
w_\bij)\;,\qquad \comp{\hat{W}}_u(w_\bij) \,= \, 
\comp{W}_u(w_\bij) \ot \qp_W(u w_\bij)\;,
\end{equation}
such that the following hold for some $\delta_{\hat{M}}= \poly(\delta_M)$. On
average over $(s,w)\sim \pild$ and for any $u \in \Fq$ and $W\in\{X,Z\}$,
\begin{equation}\label{eq:m-cons}
\begin{aligned}
\big(\hat{M}_{W,s}^r \big)_{\reg{AA'}}\otimes \Id \,\approx_{\delta_{\hat{M}}}\, \Id\otimes \big(\comp{\hat{M}}_{W,s}^r\big)_{\reg{BA''}} \;,\\
 \hat{W}_u(w_\bij)_{\reg{AA'}} \otimes \Id \,\approx_{\delta_{\hat{M}}}\, \Id \otimes \comp{\hat{W}}_u(w_\bij)_{\reg{BA''}}\;,\\
\sum_{r\in\deg_d(s)} \big(\hat{M}_{W,s}^r\big)_{\reg{AA'}}
\otimes \big(\comp{\hat{W}}_u^{\tr(r(w) u)}(w_\bij)\big)_{\reg{BA''}} \,\approx_{\delta_{\hat{M}}}\, \Id\;,\\
\sum_{r\in\deg_d(s)} \big(\hat{M}_{W,s}^r\big)_{\reg{AA'}}
\otimes \big(\comp{\hat{M}}_{W,w}^{r(w)}\big)_{\reg{BA''}} \,\approx_{\delta_{\hat{M}}}\, \Id\;,
\end{aligned}
\end{equation}
where
\begin{equation}\label{eq:def-hatm-w}
 \hat{M}_{W,w}^a \,=\, \Es{u\in\Fq} \,\omega^{-\tr(au)} \,\hat{W}_u(w_\bij)\;,\qquad  \comp{\hat{M}}_{W,w}^a \,=\,  \Es{u\in\Fq} \,\omega^{-\tr(au)}\,\comp{\hat{W}}_u(w_\bij)\;.
\end{equation}
Finally, on average over uniformly random $x,z\in\Fq^m$ and
$u, u' \in \Fq$,
\begin{equation}\label{eq:xz-com}
\hat{X}_u(x_\bij)\hat{Z}_{u'}(z_\bij) \,\approx_{\delta_{\hat{M}}} \,\hat{Z}_{u'}(z_\bij)\hat{X}_{u}(x_\bij)\;.
\end{equation}
\end{lemma}

\begin{proof} 
We first define the state $\ket{\hat{\psi}}$. For this we enlarge the Hilbert space $\mH_{\reg{A}}\otimes \mH_{\reg{B}}$ in two ways. First we assume that each prover has access to a sufficiently large number $N$ of qubits initialized in the state $\ket{0}$. This allows us
to apply Naimark's dilation theorem to simulate a POVM measurement applied by the
provers by a projective measurement, whenever it is convenient (we will always specify when we do so). Second, for each prover $D\in\{A,B\}$ we adjoin two ancilla registers
$\reg{D'},\reg{D''}$ initialized in state 
$$\ket{\psi_\epr}_{\reg{D'D''}}\,=\,\ket{\epr_q}^{\otimes
  n}_{\reg{D'D''}}\;,$$
where $\ket{\epr_q}$ is defined
in~\eqref{eq:def-epr-p}. The state of the enlarged system is 
\begin{equation}\label{eq:ancilla-def}
 \ket{\hat{\psi}}_{\reg{AA'A''}\reg{BB'B''}} \,=\, \big(\ket{\psi} \ot
\ket{0}^{\ot 2N}\big)_{\reg{AB}} \ket{\psi_\epr}_{\reg{A'A''}}
\ket{\psi_{\epr}}_{\reg{B'B''}}\;.
\end{equation}
Next, for $W\in\{X,Z\}$, $s\subset \Fq^m$ and $r\in \deg_d(s')$   let 
\begin{equation}\label{eq:hatm-def}
 \hat{M}_{W,s}^r \,=\, \sum_{\substack{r',r''\in\deg_d(s):\\ r' + r''=r}}
 M_{W,s}^{r'} \otimes \comp{\qp}_{W,s}^{r''}\;,
\end{equation}
where $\{\qp_{W,s}^{r''}\}$ is the ``honest'' subspace measurement defined
in~\eqref{eq:honest-m-def}.  Define complementary measurements 
$$\comp{\hat{M}}_{W,s}^r =
\sum_{\substack{r', r'' \in \deg_d(s):\\ r'+r''=r}} \comp{M}_{W,s}^{r'} \otimes
\qp_{W,s}^{r''}\;.$$
The following claim establishes~\eqref{eq:m-cons}.

\begin{claim}\label{claim:hat-m-cons}
For $W\in\{X,Z\}$ the subspace measurements are self-consistent, and consistent
with the point measurements: on average over $(s,w)\sim\pild$ and for any $u \in\Fq$,
\begin{align*}
\sum_{r\in\deg_d(s)} \big(\hat{M}_{W,s}^r \big)_{\reg{AA'}}\otimes \big(\comp{\hat{M}}_{W,s}^r\big)_{\reg{BA''}} &\approx_{\poly(\delta_M)} \Id\;,\\
\sum_{r\in\deg_d(s)} \big(\hat{M}_{W,s}^r\big)_{\reg{AA'}} \otimes
  \big(\comp{\hat{W}}_u^{\tr(r(w) u)}(w_\bij)\big)_{\reg{BA''}} &\approx_{\poly(\delta_M)} \Id\;,\\
	\sum_{r\in\deg_d(s)} \big(\hat{M}_{W,s}^r\big)_{\reg{AA'}}
\otimes \big(\comp{\hat{M}}_{W,w}^{r(w)}\big)_{\reg{BA''}} \,\approx_{\delta_{\hat{M}}}\, \Id\;,
\end{align*}
and
$$ \hat{W}_u(w_\bij)_{\reg{AA'}} \otimes \Id \,\approx_{\poly(\delta_M)}\, \Id \otimes \comp{\hat{W}}_u(w_\bij)_{\reg{BA''}}\;.$$
\end{claim}

\begin{proof}
Note that $\{\hat{M}_{W,w}^a\}$ as defined in~\eqref{eq:def-hatm-w} is a well-defined projective measurement, given the definition of the $\{\hat{W}_u(w_\pi)\}$ in~\eqref{eq:def-hat-xz}.

For the first identity, decompose $\hat{M}_{W,s}^r$ using the definition~\eqref{eq:hatm-def} and use self-consistency of $M_{W,s}^{r'}$, which is expressed in~\eqref{eq:subspace-consistent}, and of $\qp_{W,s}^{r''}$, which follows since the ancilla introduced in~\eqref{eq:ancilla-def} is maximally entangled on $\reg{A'A''}$. 

For the second identity, decompose $\hat{M}_{W,s}^r$ and
$\hat{W}_u^{\tr(u r(w))}(w_\bij)$ using the definition to get
\begin{align*}
\sum_{r\in\deg_d(s)}& \hat{M}_{W,s}^r \otimes \comp{\hat{W}}_u^{\tr(u
  r(w))}(w_\bij) \\
  &=\sum_{\substack{r',r'',a',a'':\\\tr( (r'  + r'')(w) \cdot u) =(a' +
  a'')}}\big(M_{W,s}^{r'}\big)_{\reg{A}} 
  \otimes \big(\comp{\qp}_{W,s}^{r''}\big)_{\reg{A'}} 
  \otimes \big(\comp{W}_u^{a'}(w_\bij)\big)_{\reg{B}} 
  \otimes \big(\qp_W^{a''}(u w_\bij)\big)_{\reg{A''}} \\
  &=\sum_{\substack{r',a':\\\tr(r'(w)\cdot u)=a'}}\big(M_{W,s}^{r'}\big)_{\reg{A}}  \otimes
  \big(\comp{W}_u^{a'}(w_\bij)\big)_{\reg{B}} \\ 
  &\approx_{\poly(\eps)} \Id\;,
\end{align*}
where the second line uses consistency of $\comp{\qp}_{W,s}^{r''}$ with
$\qp_W ^{\tr(u r''(w))} (u w_\bij)$ (which follows from the analysis of the honest
strategy given in the proof of Lemma~\ref{lem:completeness}), and
the third follows from success of the provers' strategy in the low-degree test and the definition of $\comp{W}_u(w_\bij)$ in~\eqref{eq:def-xz-obs}.    

The third identity follows from the second and the definition~\eqref{eq:def-hatm-w}.

Finally, the last relation follows from  the first, specialized to $s=w$. Alternatively, combine 
consistency between $W_\reg{A}$ and $\comp{W}_\reg{B}$ (shown in~\eqref{eq:xz-cons}) and of $(\qp_W)_\reg{A'}$ and $(\comp{\qp}_W)_\reg{A''}$, which follows since the ancilla state is $\ket{\psi_\epr}_{\reg{A'A''}}$. 
\end{proof}

The next claim establishes~\eqref{eq:xz-com}. 

\begin{claim}\label{claim:hat-xz-commute}
On average over uniformly random $x, z\in\Fq^m$ and $u, u' \in \Fq$,
$$\hat{X}_u(x_\bij)\hat{Z}_{u'}(z_\bij)\,\approx_{\poly(\delta_M)}\, \hat{Z}_{u'}(z_\bij)\hat{X}_{u}(x_\bij)\;.$$
\end{claim}

\begin{proof}
Write
\begin{align*}
\hat{X}_u(x_\bij)\hat{Z}_{u'} (z_\bij) &=  X_{u} (x_\bij)Z_{u'} (z_\bij) \otimes
                                 \comp{\qp}_X(u x_\bij)
                                               \comp{\qp}_Z(u' z_\bij)
  \\
  &= X_u(x_\bij) Z_{u'}(z_\bij) \otimes \qp_X(-u x_\bij)
    \qp_Z(u' z_\bij) \\
&\approx_{\delta_M} \big( \omega^{-\tr((u x_\bij)\cdot (u' z_\bij)}
  Z_{u'}(z_\bij)X_{u}(x_\bij) \big)\otimes \big( 
  \omega^{\tr((u x_\bij)\cdot (u' z_\bij))}\qp_Z(u' z_\bij)
  \qp_X(-u x_\bij)\big)\\
  &= Z_{u'} (z_\bij) X_{u} (x_\bij) \otimes \comp{\qp}_Z(u' z_\bij)
    \comp{\qp}_X({u} x_\bij) \\
&= \hat{Z}_{u'}(z_\bij)\hat{X}_{u}(x_\bij)\;,
\end{align*}
where the approximation follows from~\eqref{eq:xz-ac} in Claim~\ref{claim:strategies}. 
\end{proof}

\end{proof}

\subsection{Combining $X$ and $Z$ measurements}
\label{sec:hatq}

In this section we combine the approximately commuting observables $\hat{X}_u(x_\bij)$ and $\hat{Z}_u(z_\bij)$ constructed in the proof of Lemma~\ref{lem:hats} into a single POVM. We then show that the POVM leads to a strategy for the classical low-degree test. Applying Theorem~\ref{thm:ml}, we obtain a single POVM $\{\hat{S}^{g_1,g_2}\}_{g_1,g_2\in\deg_{d}(s)}$ which is simultaneously consistent with both families of observables, as shown in the following lemma. 

\begin{lemma}\label{lem:hat-s}
Let $m,d,q=p^t$ be as in Lemma~\ref{lem:soundness}, and $\ket{\hat{\psi}}$, $\hat{W}_u(w_\bij)$, $\comp{\hat{W}}_u(w_\bij)$ and $\delta_{\hat{M}}$ as in Lemma~\ref{lem:hats}. There exists  projective measurements $\{\hat{S}_{\reg{AA'}}^{g_1,g_2}\}$ and $\{\comp{\hat{S}}_{\reg{BA''}}^{g_1,g_2}\}$ with outcomes in the set of pairs $(g_1,g_2)$ of polynomials on $\Fq^m$ of total degree at most $d$ each such that, on average over uniformly random $x,z\in\Fq^m$, and for all $u \in \Fq$, 
\begin{equation}\label{eq:hat-s-xz-cons}
\sum_{g_1, g_2} \hat{S}^{g_1,g_2} \otimes  \comp{\hat{X}}_{u}^{\tr(g_1(x) u)}(x_\bij)\,\approx_{\delta_S}\, \Id\;,\qquad   \sum_{g_1, g_2} \hat{S}^{g_1,g_2} \otimes  \comp{\hat{Z}}_{u}^{\tr(g_2(z) u)}(z_\bij)\, \approx_{\delta_S}\, \Id\;,
\end{equation}
for some $\delta_S = \poly(p)\cdot\poly(\delta_{\hat{M}})+\poly(d/q)$. Similar relations hold with $\comp{\hat{S}}$ and $\hat{X}$, $\hat{Z}$ instead of $S$ and $\comp{\hat{X}}$, $\comp{\hat{Z}}$ respectively.
\end{lemma}

The proof of Lemma~\ref{lem:hat-s} proceeds in two steps. In the first step, for each pair $(u,v)\in \Fq^2$ we combine the approximately  
commuting observables $\hat{X}_u(x_\bij)$ and $\hat{Z}_{v}(z_\bij)$ into a single POVM $\{\hat{Q}_{xu,zv}^{a,b}\}$  that essentially measures in their joint eigenbasis. The following lemma shows how this can be done in general. (See~\cite{glebsky2010almost} for a similar claim that applies to arbitrary unitaries but is restricted to the Frobenius norm.)

\begin{lemma}\label{lem:approx-commute-u}
Let $\eta>0$, $\ket{\psi} \in \mH\otimes \mH'$, for finite-dimensional Hilbert spaces $\mH,\mH'$, and $W_j\in\Unitary(\mH)$, $\comp{W}_j
\in\Unitary(\mH')$, for $j=1,\ldots,k$, be such that for all
$j,\ell\in\{1,\ldots,k\}$, the powers $(W_j)^p = (\comp{W}_j)^p = \Id$, $\|(W_j\otimes \Id - \Id \otimes \comp{W}_j)\ket{\psi}\|^2 \leq \eta$, and $\|(W_jW_\ell-W_\ell W_j)\otimes \Id \ket{\psi}\|^2 \leq \eta$. Then there exists an $\eta' = \poly(p,k)\cdot\poly(\eta)$ and a POVM $\{Q^{a}\}_{a\in \Fp^k}$ such that 
$$Q^{a} \approx_{\eta'} \prod_j W_j^{a_j}\qquad\text{and}\qquad
\forall j,\quad W_j \,\approx_{\eta'}\, \sum_{a}\, \omega^{a_j}
Q^{a}\;,$$
 where $W_j^{a_j}$ is the projector onto the
eigenspace of $W_j$ associated with the eigenvalue $\omega^{a_j}$.
Moreover, there exists a projective measurement $\{\comp{Q}^{a}\}_{a \in
  \Fp^k}$ satisfying analogous relations with respect to $\comp{W}_j$ and which is consistent with $\{Q^{a}\}$: 
$$ Q^{a} \otimes \Id \,\approx_{\eta'}\, \Id \otimes \comp{Q}^{a}\;.$$
\end{lemma}

\begin{proof} 
Write the eigendecompositions 
$$W_j = \sum_{a_j\in\Fp} \omega^{a_j} W^{a_j}\;,\qquad \comp{W}_j = \sum_{a_j\in\Fp} \omega^{a_j} \comp{W}^{a_j}\;, $$
 where $\omega = e^{2i\pi/p}$. Using
Claim~\ref{claim:obs-meas-cons} the consistency assumption on $W_j$ implies
that $\{W_j^{a_j}\}$ and $\{\comp{W}_j^{a_j}\}$ are also consistent projective measurements, with error
$O(p^2\eta)$.  We use the following general claim. 

\begin{claim}\label{claim:com-obs}
Let $\{P^a\}$ and $\{Q^b\}$, and $\{\comp{P}^a\}$ and $\{\comp{Q}^b\}$, be projective measurements with outcomes $a,b\in \Fp$ for a prime $p$ and such that $P^a \otimes \Id \approx_\eta \Id \otimes \comp{P}^a$ and $Q^b \otimes \Id \approx_\eta \Id \otimes \comp{Q}^b$ for some $\eta \geq 0$. Let $P = \sum_a \omega^a P^a$ and $Q = \sum_b \omega^b Q^b$. Then 
$$\sum_{a,b} \big\|(P^aQ^b-Q^bP^a) \otimes \Id \ket{\psi}\big\|^2 \,\approx_{\poly(\eta)} O\Big(p^2\, \|(PQ-QP)\otimes \Id \ket{\psi}\|^2\Big)\;.$$
\end{claim}

\begin{proof}
Expand
\begin{align} 
\bra{\psi} (PQ-QP)(PQ-QP)^\dagger \otimes \Id \ket{\psi} &= 2 - \sum_{a,b,c,d} \omega^{a+b-(c+d)}\bra{\psi}\big(  P^a Q^b P^c Q^d + Q^a P^b Q^c P^d \big)\otimes \Id \ket{\psi} \notag\\
&\approx_{\poly(\eta)} 2-  \bra{\psi}\Big(\sum_{a}\omega^a P^a \Big)\otimes \Big(\sum_{b,c,d}\omega^{b-(c+d)}\comp{Q}^d \comp{P}^c \comp{Q}^b\Big)\ket{\psi}\notag\\
&\hskip2cm - \bra{\psi}\Big(\sum_{a}\omega^{-a} P^a \Big)\otimes \Big(\sum_{b,c,d}\omega^{-b+(c+d)}\comp{Q}^b \comp{P}^c \comp{Q}^d\Big)\ket{\psi} \;.\label{eq:com-obs-1}
\end{align}
Let $\delta=\frac{1}{2} \|(PQ-QP)\otimes \Id \ket{\psi}\|^2$. From~\eqref{eq:com-obs-1} we get the
consistency relation 
\begin{align*}
\frac{1}{2}\sum_{a,b,c,d}\, \bra{\psi}\big( \omega^{(a-c)+(b-d)} Q^d P^c Q^b + \omega^{(c-a)+(d-b)}Q^b P^c Q^d\big) \otimes \comp{P}^{a} \ket{\psi}\, \approx_{\poly(\eta)} \,1-O(\delta)\;.
\end{align*}
We now repeat this argument, but replacing the observable $Q$ with its
power $(Q)^{\alpha}$ for $\alpha \in \Fp$ (not to be confused with the POVM
element $Q^a$). Starting from $\|(P(Q)^\alpha-(Q)^\alpha P)\otimes \Id \ket{\psi}\|^2 =  O(p\delta)$ for any $\alpha\in\Fp$ gives 
\begin{align*}
\frac{1}{2}\sum_{a,b,c,d}\, \bra{\psi}\big( \omega^{(a-c)+\alpha (b-d)} Q^d P^c Q^b + \omega^{(c-a)+\alpha (d-b)}Q^b P^c Q^d\big) \otimes \comp{P}^{a} \ket{\psi}\, \approx_{\poly(\eta)} \,1-O(p\delta)\;.
\end{align*}
 Averaging the relations over all $\alpha\in\Fp$ and using that $(b-d) \Fp = \Fp$ if
 $b-d\neq 0$ and $\{0\}$ otherwise yields
\begin{align*}
\frac{1}{2}\sum_{a,b,c}\, \bra{\psi} \big(\omega^{a-c} + \omega^{c-a}\big) Q^b
  P^a Q^b \otimes \comp{P}^{a} \ket{\psi}  \approx_{\poly(\eta)} 1-O(p^2\delta)\;.
\end{align*}
Using $|\omega^{a-c}+\omega^{c-a}|\leq 2-\Omega(1/p)$ for $a\neq c$ proves the claim. 
\end{proof}

We define the POVM elements $\{Q^a\}$ as 
$$ Q^a \,=\, W_k^{a_k} \cdots W_1^{a_1} \cdots W_k^{a_k}\;,$$
and define $\{\comp{Q}^a\}$ analogously. The two relations in the lemma then follow from the definition, the fact that $\{W_j^{a_j}\}$ are projections, and Claim~\ref{claim:com-obs}; similarly, consistency of $\{Q^a\}$ follows. 
\end{proof}

Lemma~\ref{lem:approx-commute-u} allows us to combine any pair of approximately commuting observables $\hat{X}_u(x_\bij)$ and $\hat{Z}_v(z_\bij)$ into a single POVM $\{\hat{Q}_{xu,zv}^{a,b}\}$, with outcomes $(a,b)\in \Fp^2$, that simultaneously refines both observables. In the following lemma we show that all $\{\hat{Q}_{xu,zv}^{a,b}\}$, for $u,v\in \Fq$, can be pieced together into a single POVM $\{\hat{Q}_{xz}^{a,b}\}$ with outcomes $(a, b) \in \Fq^2$  which simultaneously refines all $2 q$ observables $\hat{X}_u(x_\bij)$ and $\hat{Z}_v(z_\bij)$. Note that this is not trivial because we wish to avoid any dependence on $q$ when combining the $2q$ approximately commuting observables. To achieve this, we rely on the linearity test from~\cite{NV17}. 

\begin{lemma}\label{lem:hatq-meas-two-outcome}
For every $x,z\in\Fq^m$ there are projective
measurements $\{\hat{Q}_{xz}^{a,b}\}_{a,b\in \Fq^2}$ and
$\{\comp{\hat{Q}}_{xz}^{a,b}\}_{a,b\in \Fq^2}$  defined on $\reg{AA'}$ and
 $\reg{BA''}$ respectively such that the following hold, for some $\delta_Q = \poly(p)\cdot\poly(\delta_{\hat{M}})$:
\begin{enumerate}
\item The $\hat{Q}_{xz}$ are consistent with the $\comp{\hat{Q}}_{xz}$: 
$$ \sum_{a,b} \big(\hat{Q}_{xz}^{a,b}\big)_{\reg{AA'}} \otimes  \big(\comp{\hat{Q}}_{xz}^{a,b}\big)_{\reg{BA''}}\, \approx_{\delta_Q} \Id\;,$$
on average over uniformly random $x,z\in\Fq^m$.
\item The $\hat{Q}_{xz\alpha\beta}$ are consistent with
  $\comp{\hat{M}}_{X, x}$ and
  $\comp{\hat{M}}_{Z, z}$: 
 \[ \sum_{a,b} \big(\hat{Q}_{xz}^{a,b}\big)_{\reg{AA'}} \otimes  \Big(
 \comp{\hat{M}}_{X, x}^a
 \comp{\hat{M}}_{Z, z}^{b}\Big)_{\reg{BA''}} \approx_{\delta_Q} \Id\;, \]
on average over uniformly random $x,z\in\Fq^m$. 
\end{enumerate}
\end{lemma}

\begin{proof}
For every $x,z\in\Fq^m$ and $u,v\in \Fq$ we let $\{ \hat{Q}_{xu,zv}^{a,b} \}_{a,b\in\Fp}$ and
$\{\comp{\hat{Q}}_{xu,zv}^{a,b}\}_{a,b\in \Fp}$ be the projective
measurements guaranteed by Lemma~\ref{lem:approx-commute-u} when the observables $P$ and $Q$ are chosen as $\hat{X}_u(x_\bij)$ and $\hat{Z}_v(z_\bij)$ respectively; the assumptions of the lemma are satisfied by~\eqref{eq:m-cons} and~\eqref{eq:xz-com}. We also define, for $c\in \Fp$, 
\begin{equation}\label{eq:def-hatq-c}
\hat{Q}_{xu,zv}^{c} =  \sum_{a,b:\, a +  b = c} \hat{Q}_{xu,zv}^{a,b}\;,
\end{equation}
and similarly define associated measurements $\{\comp{\hat{Q}}_{xu,zv}^{c}\}_{c\in\Fp}$. In the following claim we identify $\Fq$ with the vector space $\Fp^t$ to interpret this family of measurements as a strategy in the linearity test over $\Fp^{2t}$, with the queries specified by $(u,v)\in\Fp^{2t}$ and the answers $c\in\Fp$.

\begin{claim}\label{claim:q-lin}
On average over $x,z\in\Fq^m$ the family of projective measurements $\{ \hat{Q}_{xu,zv}^{c} \}_{c\in\Fp}$, indexed by $(u,v)\in \Fp^{2t}$, together with $\{\comp{\hat{Q}}_{xu,zv}^{c} \}_{c\in\Fp}$, induces a strategy with success probability at least $1-\poly(\delta_{\hat{M}})$ in the two-prover linearity test from~\cite[Section 3]{NV17}. 
\end{claim}

\begin{proof}
Consistency between  $\hat{Q}_{xu,zv}^{c}$ and $ \comp{\hat{Q}}_{xu,zv}^{c}$ is clear by the definition and Lemma~\ref{lem:approx-commute-u}. We need to verify (approximate) linearity. First note that from the definition of the observables $\hat{W}_u$  in~\eqref{eq:def-hat-xz}, using that the underlying measurements are projective it follows that for $W\in\{X,Z\}$ the family of observables $\{\hat{W}_u(w_\bij)\}_{u\in\Fp^t}$ is linear, in the sense that for any $u,u'\in\Fp^t$, $\hat{W}_u(w_\bij)\hat{W}_{u'}(w_\bij)=\hat{W}_{u+u'}(w_\bij)$, where $u+u'$ is addition as elements of the vector space $\Fp^t$. Using Lemma~\ref{lem:approx-commute-u} it follows that $\{ \hat{Q}_{xu,zv}^{c} \}$ is approximately linear in both $u$ and $v$, i.e. 
$$ \Es{u,u',v,v'} \sum_{c,c'} \, \hat{Q}_{xu,zv}^{c} \hat{Q}_{xu',zv'}^{c'} \otimes \comp{\hat{Q}}_{x(u+u'),z(v+v')}^{c+c'} \,\approx_{\poly(p,\delta_{\hat{M}})}\, \Id\;.$$
Linearity for $\{ \hat{Q}_{xu,zv}^{c} \}$ then follows directly from the definition~\eqref{eq:def-hatq-c}. 
\end{proof}

Applying~\cite[Theorem 10]{NV17} for every $x,z$ we find a single POVM $\{ \hat{Q}_{x,z}^{a,b} \}_{a,b\in\Fq^{2}}$ such that, on expectation over $u,v\in \Fp^t$, $\hat{Q}_{xu,zv}^{e} \approx \sum_{a,b:\, \tr(a \cdot u +b \cdot v) = e} \hat{Q}_{x,z}^{a,b}$.\footnote{The results in~\cite[Theorem 10]{NV17} apply to the analysis of the linearity test over $\F_2^{2t}$, i.e. the case $p=2$ here. The same proof extends with, minor modifications, to the case of arbitrary prime $p$. We omit the details.}
Similarly, there exists a $\{ \comp{\hat{Q}}_{x,z}^{a,b} \}_{a,b\in\Fq^{2}}$ that is consistent with $\{ \hat{Q}_{x,z}^{a,b} \}_{a,b\in\Fq^{2}}$. The first item in the lemma now follows immediately from the consistency guarantees of the linearity test.

For the second item, note that from Lemma~\ref{lem:approx-commute-u} and the guarantees of the linearity test, it follows that $\hat{Q}_{x,z}^{(a,b)}$ is approximately consistent with a randomly chosen product of $\comp{\hat{X}_u}$ and $\comp{\hat{Z}_v}$ POVM elements.
\begin{align*} \Id &\approx_{\delta_{\hat{M}}'}  \Es{u,v} \sum_{(a,b)} (\hat{Q}_{x,z}^{a,b})_{\reg{AA'}} \ot \Big( \sum_{c,d : \tr(au + bv) = c +d} \comp{\hat{X}}_u^{c} (x_\pi) \comp{\hat{Z}}_v^{d} (z_\pi) \Big)_{\reg{BA''}} \\
  &= \Es{u,v} \sum_{(a,b)} (\hat{Q}_{x,z}^{a,b})_{\reg{AA'}} \ot \Big( \sum_{a', b': \tr(au + bv) = \tr(a'u + b'v)} \comp{\hat{M}}^{a'}_{X,x} \comp{\hat{M}}^{b'}_{Z,z} \Big)_{\reg{BA''}} \\
  &= \sum_{(a,b)} (\hat{Q}_{x,z}^{a,b})_{\reg{AA'}} \ot \Big( \Es{u,v} \sum_{a', b'} 1_{\tr(au + bv) = \tr(a'u + b' v)} \comp{\hat{M}}^{a'}_{X,x} \comp{\hat{M}}^{b'}_{Z,z} \Big)_{\reg{BA''}} \\
  &= \sum_{(a,b)} (\hat{Q}_{x,z}^{a,b})_{\reg{AA'}} \ot \Big( \comp{\hat{M}}^{a}_{X,x} \comp{\hat{M}}^{b}_{Z,z} + \frac{1}{p} (\Id - \comp{\hat{M}}^{a}_{X,x} \comp{\hat{M}}^{b}_{Z,z}) \Big)_{\reg{BA''}} \\
  &= \frac{1}{p} \Id + \Big(1 - \frac{1}{p}\Big) \sum_{(a,b)} (\hat{Q}_{x,z}^{a,b})_{\reg{AA'}} \ot \Big(\comp{\hat{M}}^{a}_{X,x} \comp{\hat{M}}^{b}_{Z,z} \Big)_{\reg{BA''}}\;, 
\end{align*}
where in the second to last equation we used the fact that any two distinct linear functions from $\Fp^{2t}$ to $\Fp$ agree with probability $1/p$. Moving the term proportional to $\Id$ to the left hand side, we obtain
\[  \Id \approx_{O(\delta'_{\hat{M}})} \sum_{a,b} (\hat{Q}_{x,z}^{a,b})_{\reg{AA'}} \ot \Big(\comp{\hat{M}}^{a}_{X,x} \comp{\hat{M}}^{b}_{Z,z} \Big)_{\reg{BA''}}\;,\]
which is the desired consistency expression.
\end{proof}

For a subsequent application of the low-degree test it is more convenient to have a POVM which takes values $c\in\Fq$, rather than $(a,b)\in\Fq^2$. In order not to lose information, the following lemma re-arranges $\hat{Q}_{xz}^{a,b}$ into a family of POVMs $\{\hat{Q}_{xz\alpha\beta}^c = \sum_{a,b:\,\alpha a + \beta b =c} \hat{Q}_{xz}^{a,b}\}$, and shows that this family obeys similar consistency properties.

\begin{lemma}\label{lem:hatq-meas}
For every $\alpha,\beta\in\Fq$ and $x,z\in\Fq^m$ there are projective
measurements $\{\hat{Q}_{xz\alpha\beta}^{c}\}_{c\in \Fq}$ and
$\{\comp{\hat{Q}}_{xz\alpha\beta}^{c}\}_{c\in \Fq}$  defined on $\reg{AA'}$ and
 $\reg{BA''}$ respectively such that the following hold, for $\delta_Q$ as in Lemma~\ref{lem:hatq-meas-two-outcome}:
\begin{enumerate}
\item The $\hat{Q}_{xz\alpha\beta}$ are consistent with the $\comp{\hat{Q}}_{xz\alpha\beta}$: 
$$ \sum_c \big(\hat{Q}_{xz\alpha\beta}^c\big)_{\reg{AA'}} \otimes  \big(\comp{\hat{Q}}_{xz\alpha\beta}^c\big)_{\reg{BA''}}\, \approx_{\delta_Q} \Id\;,$$
on average over uniformly random  $\alpha,\beta\in\Fq$ and $x,z\in\Fq^m$.
\item The $\hat{Q}_{xz\alpha\beta}$ are consistent with
  $\comp{\hat{M}}_{X, x}$ and
  $\comp{\hat{M}}_{Z, z}$: 
 \[ \sum_c \big(\hat{Q}_{xz\alpha\beta}^c\big)_{\reg{AA'}} \otimes  \Big(
 \sum_{\substack{a, a':\\ \alpha a + \beta a'=c}} \comp{\hat{M}}_{X, x}^a
 \comp{\hat{M}}_{Z, z}^{a'}\Big)_{\reg{BA''}} \approx_{\delta_Q} \Id\;, \]
on average over uniformly random  $\alpha,\beta\in\Fq$ and $x,z\in\Fq^m$. 
\end{enumerate}
\end{lemma}
\begin{proof}
Let $\hat{Q}_{x,z}^{a,b}$  and $\comp{\hat{Q}}_{x,z}^{a,b}$ be the families of POVMs guaranteed by Lemma~\ref{lem:hatq-meas-two-outcome}. We ``collapse'' the $\hat{Q}$ (and analogously, the $\comp{\hat{Q}}$) into a family of measurements with outcomes in $\Fq$ by defining, for every $\alpha,\beta\in\Fq$, $x,z\in\Fq^m$ and $c\in\Fq$, 
\begin{equation}\label{eq:hat-q-def}
 \hat{Q}_{xz\alpha\beta}^{c} \,=\,\sum_{a,b:\,\alpha a + \beta b = c} \hat{Q}_{x,z}^{a,b}\;.
\end{equation}
With this definition, both items in the claim follows from the definition~\eqref{eq:hat-q-def} and corresponding items of Lemma~\ref{lem:hatq-meas-two-outcome}.


\end{proof}

The next step in the proof of Lemma~\ref{lem:hat-s} is to use the family of measurements $\{\hat{Q}_{xz\alpha\beta}^c\}$ to devise a strategy for the provers in the classical $(2m+2)$-variable degree-$(d+1)$ test. Towards this the following claim establishes the existence of appropriate subspace measurements. 

\begin{claim}\label{claim:lines}
For every dimension-$2$ subspace $s \subseteq \Fq^{2m+2}$ there exists a POVM
$\{\hat{Q}_{s}^r\}_r$, with outcomes $r\in \deg_{d+1}(s)$, such that, on average over a uniformly random $s$ and $(x,z,\alpha,\beta)\in s$,
$$ \sum_{r,c:\,r(x,z,\alpha,\beta)\neq c}\, \bra{\psi} \hat{Q}_s^r \otimes
\comp{\hat{Q}}_{xz\alpha\beta}^c \ket{\psi} \,=\, \poly(\delta_Q) +
O\big(d/q\big)\;.$$
Likewise, there exists $\{\comp{\hat{Q}}_{s}^r\}$ that are consistent with $\{\hat{Q}_{xz\alpha\beta}^c\}$, on average.
\end{claim}

\begin{proof}
The argument is entirely symmetric between $\hat{Q}$ and $\comp{\hat{Q}}$. We give the construction for $\hat{Q}$. 

Let $s\subset \Fq^{2m+2}$ be a two-dimensional linear subspace spanned by two vectors $y_1 = (x_1,z_1,\alpha_1,\beta_1),y_2=(x_2,z_2,\alpha_2,\beta_2)\in\Fq^{2m+2}$. Let $s',s'' \subset \Fq^m$ be the (at most) two-dimensional subspaces spanned by $\{x_1,x_2\}$ and $\{z_1,z_2\}$ respectively. 
Any point $(\lambda, \mu)$ in the subspace $s$ corresponds to a point in the full space
\[ \#(\lambda, \mu) = (\lambda x_1 + \mu x_2, \lambda z_1 + \mu z_2, \lambda \alpha_1 + \mu \alpha_2, \lambda \beta_1 + \mu \beta_2). \]
Let $g(x,y, \alpha, \beta) = \alpha p(x) + \beta q(y)$ be a $2m + 2$-variate polynomial. Its restriction $r$ to the subspace $s$ takes the values
\[ r(\lambda, \mu) = g(\#(\lambda, \mu)) = (\lambda \alpha_1 + \mu \alpha_2) p(\lambda x_1 + \mu x_2) + (\lambda \beta_1 + \mu \beta_2) q(\lambda z_1 + \mu z_2). \]
From this expression, we see that $r$ can be evaluated if we have
access to the restrictions of $p$ to $s'$ and $q$ to $s''$,
respectively. We will now construct a measurement that, given $s',
s''$ as input, jointly measures $p$ and $q$ in a manner that is
consistent with the joint points measurement $\{\hat{Q}_{xz}^{a,b}\}$ guaranteed by Lemma~\ref{lem:hatq-meas-two-outcome}. Define
\begin{equation} T_{s',s''}^{p,q} = \hat{M}_{X,s'}^{p} \hat{M}_{Z,s''}^{q}
  \hat{M}_{X,s'}^{p}. \label{eq:joint_subspace} \end{equation}
The collection $\{T_{s', s''}^{p,q}\}_{p,q}$ forms a valid POVM. Its
inconsistency with $\hat{Q}_{xz}^{a,b}$ is:
\begin{align}
  &\Es{s',s''} \Es{x \in s', z \in s''} \sum_{\substack{a,b,p,q: \\ (a,b) \neq (f(x), q(z))}}
  \bra{\psi} T_{s', s''}^{p,q} \ot \comp{\hat{Q}}_{xz}^{a,b} \ket{\psi}
  \nonumber \\
  &\qquad=\Es{s', s''} \Es{x \in s', z \in s''} \sum_{\substack{a,b,p,q: \\ (a,b) \neq
  (p(x), q(z))}}\bra{\psi} \hat{M}_{X,s'}^{p} \hat{M}_{Z,s''}^{q} \hat{M}_{X,s'}^{p} \ot
  \comp{\hat{Q}}_{xz}^{a,b} \ket{\psi} \nonumber \\
  &\qquad= 1 - \Es{s', s''} \sum_{p, q} \Es{x \in s', z \in s''} \bra{\psi} \hat{M}_{X,s'}^{p} \hat{M}_{Z,s''}^{q} \hat{M}_{X,s'}^{p} \ot
  \comp{\hat{Q}}_{xz}^{p(x),q(z)} \ket{\psi} \nonumber\\
  &\qquad \leq 1  - \Es{s', s''} \sum_{p, q, p'} \Es{x \in s', z \in s''} \bra{\psi} \hat{M}_{X,s'}^{p} \hat{M}_{Z,s''}^{q} \hat{M}_{X,s'}^{p'} \ot
  \comp{\hat{Q}}_{xz}^{p(x),q(z)} \ket{\psi}   + \poly(\delta_Q) + O(d/q)
    \nonumber \\
  &\qquad = 1  - \Es{s', s''} \sum_{p, q} \Es{x \in s', z \in s''} \bra{\psi} \hat{M}_{X,s'}^{p} \hat{M}_{Z,s''}^{q}  \ot
  \comp{\hat{Q}}_{xz}^{p(x),q(z)} \ket{\psi}  + \poly(\delta_Q) + O(d/q) \nonumber\\
  &\qquad = \Es{s', s''} \sum_{p, q} \sum_{(a,b) \neq (p(x), q(z))} \Es{x \in s', z \in s''} \bra{\psi} \hat{M}_{X,s'}^{p} \hat{M}_{Z,s''}^{q}  \ot
  \comp{\hat{Q}}_{xz}^{a,b} \ket{\psi}  + \poly(\delta_Q) + O(d/q) \nonumber \\
  &\qquad =  \poly(\delta_Q) + O(d/q)\;. \label{eq:joint_subspace_points}
\end{align}
Here, in the first equation we used~\eqref{eq:joint_subspace}. In the third
equation we used item 2. of Lemma~\ref{lem:hatq-meas} to bound the error terms
where $p'(x) \neq p(x)$ by $\poly(\delta_Q)$, and used Lemma~\ref{lem:sz} to bound the probability
that $p(x) = p'(x)$ on a randomly chosen $x$ for distinct polynomials $p, p'$ by
$O(d/q)$. Finally, in the sixth equation, we again used item 2. of Lemma~\ref{lem:hatq-meas}.

Finally, we will use this to construct a joint subspace
measurement $\{\hat{Q}_{s}^r\}_r$:
\begin{equation} \hat{Q}_{s}^r = \sum_{p, q: r = (\lambda \alpha_1 + \mu \alpha_2) p(\lambda
    x_1 + \mu x_2) + (\lambda \beta_1 + \mu \beta_2) q(\lambda z_1 + \mu z_2)}
  T_{s', s''}^{p,q}\;. \label{eq:hatq_subspace_def} 
	\end{equation}
The consistency of $\hat{Q}_{s}^r$ with the re-arranged points measurement $\hat{Q}_{xz\alpha \beta}^{c}$ now follows directly from~\eqref{eq:hatq_subspace_def} together with the
consistency of $T_{s', s''}^{p,q}$ with $\hat{Q}_{xz}^{a,b}$ established in~\eqref{eq:joint_subspace_points}.
\end{proof}

Claim~\ref{claim:lines} shows that the families of measurements
$\{\hat{Q}_{s}^r\}_r$ and $\{\hat{Q}_{xz\alpha\beta}^c\}$ induce a strategy with
success probability $1-\eps'$ in the classical low-degree test, for some $\eps'
= \poly(p) \cdot \poly(\delta_{\hat{M}})+O(d/q)$ (to obtain this error bound,
recall that $\delta_Q = \poly(p) \cdot \poly(\delta_{\hat{M}})$). This allows us to apply Theorem~\ref{thm:ml} and conclude the proof of Lemma~\ref{lem:hat-s}. 


\begin{proof}[Proof of Lemma~\ref{lem:hat-s}]
The proof is symmetric under exchanging $\hat{S}, \hat{X}, \hat{Z}$ with $\comp{\hat{S}}, \comp{\hat{X}}, \comp{\hat{Z}}$, so we only present one direction.
Claim~\ref{claim:lines} establishes the existence of a strategy
$\{\hat{Q}_{s}^r,\comp{\hat{Q}}_s^r\}$ which succeeds with probability
$1-\eps'$, for some $\eps' = \poly(t)\cdot\poly(\delta_{\hat{M}})+O(d/q)$, in
the classical low-degree test for $m'=2m+2$ and $d'=d+1$. Increasing the error
artificially by replacing $O(d/q)$ with $md/q^{1/c}$ if needed, the assumption
$q \geq (dm/\eps')^c $ from Theorem~\ref{thm:ml} is satisfied. The theorem yields a POVM
$\{\hat{S}^g\}$ with outcomes in the set of polynomials 
$g:\Fq^{2m+2}\to\Fq$ of degree at most $d+1$ that is self-consistent and
consistent with the $\{\hat{Q}_{s}^r\}$ and
the family of POVM $\{\hat{Q}_{xz\alpha\beta}^c\}$ defined in Lemma~\ref{lem:hatq-meas}. Moreover, applying Naimark's
dilation theorem (taking advantage of the assumption that the state $\ket{\hat{\psi}}$ defined in~\eqref{eq:ancilla-def}, by definition, contains sufficiently many ancilla $\ket{0}$ qubits), we can assume without loss of generality that $\{\hat{S}^g\}$ is a projective measurement.

In general there is no a priori guarantee that $g$ takes the form $g=\alpha g_1 + \beta g_2$ for $g_1$ (resp. $g_2$) a degree-$d$ polynomial in $x$ (resp. $z$) only. Let $\mathcal{G}$ denote the latter set of polynomials. We first show that the probability of obtaining an outcome $g$ that does not fall within the set $\mathcal{G}'$ of $(2m+2)$-variate polynomials that are linear in $\alpha$ and $\beta$ is small. 
  Using consistency of $\{\hat{S}^g\}$ with the $\{\hat{Q}_{xz\alpha\beta}^c\}$, which follows from item 1. in Theorem~\ref{thm:ml}, 
  \begin{align}
    \sum_{g \notin \mathcal{G}'} \hat{S}^g 
    &\approx_{\poly(\eps')} \sum_{g \notin \mathcal{G}'}  \Es{(x,z,\alpha,\beta)\in\Fq^{2m+2}} \hat{S}^g
      \ot \comp{\hat{Q}}^{g(x,z,\alpha,\beta)}_{xz\alpha\beta} \notag\\
    &\approx_{\poly(\delta_Q)} \sum_{g \notin \mathcal{G}'} \sum_{a,b} \Es{x,z} \Big( \Es{\al,\be} 1_{\alpha a + \beta b = g(x,z,\alpha,\beta)} \Big) \hat{S}^g
      \ot \comp{\hat{M}}_{X, x}^{a} \comp{\hat{M}}_{Z,z}^b\;,\label{eq:hats-1}
			\end{align}
			where the second approximation follows from item 2. in
      Lemma~\ref{lem:hatq-meas}. If $g$ contains a term of degree higher than
      $1$ in $\alpha$ or $\beta$ the expectation inside the brackets
      in~\eqref{eq:hats-1} is at most $O(d/q)$. This bounds the contribution of
      outcomes $g\notin \mathcal{G}'$. Next we show that 
  outcomes in $\mathcal{G}'\backslash \mathcal{G}$ are unlikely, i.e. that $g_1$ should
  depend solely on $x$ and $g_2$ solely on $z$. (Note that either polynomial has degree at most $d$, since $g$ itself has degree at most $d+1$.) Towards this, starting from~\eqref{eq:hats-1}, write  
	\begin{align*}
	  \sum_{g = \alpha g_1 + \beta g_2\in \mathcal{G}'} \hat{S}^g 
    &\approx_{\poly(\eps',\delta_Q)} \sum_{g = \alpha g_1 + \beta g_2}  \hat{S}^g
      \ot \Es{z}\,\sum_b \,\Big( \sum_a \Es{x}
      1_{a=g_1(x,z)} 1_{b = g_2(x,z)} \comp{\hat{M}}_{X, x}^{a}
      \Big)\comp{\hat{M}}_{Z, z}^b\\
    &\approx_{\poly(\delta_Q)} \sum_{g = \alpha g_1 + \beta g_2} \hat{S}^g \ot
      \Es{ z} \, \sum_b \, \sum_a \, \Es{x} 1_{a=g_1(x,z)}
      1_{b=g_2(x,z)} \comp{\hat{Q}}^{a,b}_{x,z}\\
      &\leq \sum_{g = \alpha g_1 + \beta g_2} \hat{S}^g \ot
      \Es{z} \, \sum_b \, \Es{x} 
      1_{b=g_2(x,z)} \Big( \sum_a \comp{\hat{Q}}^{a,b}_{x,z} \Big)\\
		&\approx_{\poly(\delta_Q)} \sum_{g = \alpha g_1 + \beta g_2}  \hat{S}^g
      \ot \Es{z}\,\sum_b \,\Big(  \Es{x} 1_{b = g_2(x,z)}
      \Big)\comp{\hat{M}}_{Z, z}^b\;,
	\end{align*}
	where the inequality removes the first indicator by using positivity. 
	If $g_2(x,z)$ depended on $x$, the indicator $1_{b=g_2(x,z)}$ appearing within the expression inside the brackets would have probability at most $O(d/q)$ to be satisfied, for a random choice of $b$ and $z$, on average over $x$. Thus outcomes $g=\alpha g_1 + \beta g_2$ for which $g_2$ depends on $x$ have small probability; similarly for $z$.
The relation~\eqref{eq:hat-s-xz-cons} then follows directly from the above, Claim~\ref{claim:hat-m-cons}, and the second item in Lemma~\ref{lem:hatq-meas}.
\end{proof}

\subsection{Generalized $X$ and $Z$ observables}
\label{sec:tildex}

In this section we complete the last main step of the proof. For $W\in\{X,Z\}$, $w\in\Fq^m$ and $u\in \Fq$ let ${W}_u(w_\bij)$ be the observable defined in Claim~\ref{claim:strategies}, and $\hat{W}_u(w_\bij)$ defined in Lemma~\ref{lem:hats}. For convenience we consider a ``basis'' for these sets of observables, as $u$ ranges over $\Fp^t$, by defining 
\begin{equation}\label{eq:def-wl}
 \forall \ell\in\{1,\ldots,t\}\;,\qquad {W}_\ell(w_\bij)\,=\, {W}_{b_\ell}(w_\bij)\;\qquad\text{and}\qquad  \hat{W}_\ell(w_\bij)\,=\, \hat{W}_{b_\ell}(w_\bij)\;,
\end{equation}
where $\{b_1,\ldots,b_t\}$ is a self-dual basis for $\Fq$ over
$\Fp$. The corresponding ``true'' Paulis $\sigma_{W, \ell}(w_\bij)$ were defined in~\eqref{eq:pauli-l}. Similarly define $\comp{W}_\ell(w_\bij)$ and $\comp{\hat{W}}_\ell(w_\bij)$.
 In this section we use these observables to construct a family of observables
$\tilde{W}_\ell(a)$, for $a\in\Fq^n$, which satisfy appropriate self-consistency and twisted commutation relations, as expressed in the following lemma.
\begin{lemma}\label{lem:xz-lowdeg}
Let $m,d,q=p^t$ be as in Lemma~\ref{lem:soundness}, $\ket{\hat{\psi}}$ as in Lemma~\ref{lem:hats}, $W_\ell(w_\bij)$ and
$\comp{W}_\ell(w_\bij)$, 
$\hat{W}_\ell(w_\bij)$ and $\comp{\hat{W}}_\ell(w_\bij)$ as in~\eqref{eq:def-wl}, and
$\{\hat{S}^{g_1,g_2}\}$, $\{\comp{\hat{S}}^{g_1,g_2}\}$, and $\delta_S$ as in
Lemma~\ref{lem:hat-s}.  

For every $a\in\Fq^n$ and $\ell \in \{1, \dots, t\}$ there are observables $\tilde{X}_\ell(a)$ acting
on $\reg{AA'A''}$ and $\comp{\tilde{X}}_\ell(a)$ acting on $\reg{BB'B''}$,
and for every $b\in \Fq^n$ and $\ell' \in \{1, \dots , t\}$ observables $\tilde{Z}_{\ell'}(b)$ acting on
$\reg{AA'A''}$ and $\comp{\tilde{Z}}_\ell(b)$ acting on $\reg{BB'B''}$, such that the following properties hold for some $\delta_W = \poly(\delta_S)+\poly(d/q)$: 
\begin{enumerate}
\item The families of observables $\{\tilde{X}_\ell(a),\tilde{Z}_{\ell'}(b),\,
  a,b\in\Fq^n,\,\ell, \ell' \in \{1, \dots, t\} \}$ and
  $\{\comp{\tilde{X}}_\ell(a), \comp{\tilde{Z}}_{\ell'} (b), \,
  a,b, \in \Fq^n,\, \ell, \ell' \in \{1, \dots, t\} \}$ exactly satisfy the same 
  algebraic relations as the Pauli observables $\sigma_{W, \ell}$ over $\Fq$ defined in~\eqref{eq:pauli-l};
\item On average over $x,z\in\Fq^m$,  and for all $\ell, \ell' \in \{1, \dots, t\}$,
$$\tilde{X}_{\ell}^e(x_\bij) \approx_{\delta_W} X_{\ell}^e(x_\bij)\qquad
\text{and}\qquad \tilde{Z}_{\ell'}^e(z_\bij) \approx_{\delta_W} Z_{\ell'}^e(z_\bij)\;,$$
and the analogous relations between $\comp{\tilde{X}}_\ell,
\comp{\tilde{Z}}_{\ell'}$ and $\comp{X}_\ell, \comp{Z}_{\ell'}$ hold;
\item The $\tilde{X}_{\ell}(a)$, $\tilde{Z}_{\ell'}(b)$ are approximately consistent
  with the $\comp{\tilde{X}}_{\ell}(a), \comp{\tilde{Z}}_{\ell'}(b)$: in expectation over $a,b\in\Fq^n$,
$$ \sum_{e \in \Fp} \,\tilde{X}_\ell^e(a) \otimes \comp{\tilde{X}}_\ell^e(a)
\approx_{\delta_W} \Id\;,\qquad  \sum_{e \in \Fp} \,\tilde{Z}_{\ell'}^e(b) \otimes \comp{\tilde{Z}}_{\ell'}^e(b) \approx_{\delta_W} \Id\;.$$
\end{enumerate}
\end{lemma}

\begin{proof}
For any $a,b\in\Fq^n$ and $\ell \in \{1, \dots, t\}$ define 
\begin{equation}\label{eq:def-tilde-xz}
\begin{aligned}
\tilde{X}_\ell(a)_\reg{AA'A''} &=  \Big(\sum_{g_1,g_2} \,\omega^{\tr( b_\ell(g_1\cdot a))}\,
\hat{S}^{g_1,g_2}_\reg{AA'} \Big) \otimes \sigma_{X, \ell}(a)_\reg{A''}^\dag\;,\quad
\tilde{Z}_\ell(b)_\reg{AA'A''} =  \Big(\sum_{g_1,g_2} \,\omega^{\tr(b_\ell(g_2\cdot
  b))}
\,\hat{S}^{g_1,g_2}_\reg{AA'}\Big) \otimes \sigma_{Z, \ell}(b)_\reg{A''}^\dag\;, \\
\comp{\tilde{X}}_{\ell}(a)_\reg{BB'B''} &= \Big(\sum_{g_1,g_2} \,\omega^{\tr(b_\ell(g_1\cdot a))}\,
\comp{\hat{S}}^{g_1,g_2}_{\reg{BB'}} \Big) \otimes \comp{\sigma}_{X,\ell}(a)_\reg{B''}^\dag\;,\quad
\comp{\tilde{Z}}_{\ell}(b)_{\reg{BB'B''}} =  \Big(\sum_{g_1,g_2} \,\omega^{\tr(b_\ell(g_2\cdot b))}
\,\comp{\hat{S}}^{g_1,g_2}_\reg{BB'} \Big) \otimes \comp{\sigma}_{Z, \ell}(b)_\reg{B''}^\dag\;,
\end{aligned}
\end{equation}
where $\comp{\hat{S}}^{g_1,g_2}_{\reg{BB'}}$ is defined as $\comp{\hat{S}}^{g_1,g_2}_{\reg{BA'}}$, with the role of the register $\reg{A'}$ replaced by $\reg{B'}$ (note the two are isomorphic). 
From Lemma~\ref{lem:hat-s} we know that $\{\hat{S}^{g_1,g_2}\}$ is a projective
measurement. In particular the first component of each of the observables
defined in~\eqref{eq:def-tilde-xz} commute; hence the first item in the lemma
follows from the fact that $\sigma_{X, \ell}(a)$ and $\sigma_{Z,\ell}(b)$ themselves satisfy
the required Pauli relations. 

Next we verify the second item. We give the proof for the $X$ observables; the
arguments for $Z, \comp{X}, \comp{Z}$ are similar. Using the definition~\eqref{eq:def-tilde-xz}, on average over $x\in\Fq^m$,
\begin{align*}
\tilde{X}_\ell^a(x_\bij) &=\sum_{b,c \in \Fp:\,b - c=a}\Big(\sum_{g_1,g_2:\,
                           \tr(b_\ell g_1(x))=b}
                      \hat{S}_\reg{AA'}^{g_1,g_2} \Big) \otimes
                      \sigma_{X, \ell}^c(x_\bij)_\reg{A''}\\ 
&\approx_{\poly(\delta_S)} \sum_{b,c \in \Fp:\,b-c=a}
  \hat{X}_\ell^b(x_\bij)_\reg{AA'} \otimes \sigma_{X, \ell}^c(x_\bij)_\reg{A''} \\
&= \sum_{b',b'',c \in \Fp :\,b'+b'' - c=a} X_\ell^{b'}(x_\bij)_\reg{A} \otimes
  \comp{\sigma}_{X, \ell}^{b''}(x_\bij)_\reg{A'}\otimes \sigma_{X, \ell}^c(x_\bij)_\reg{A''} \\
& = X_{\ell}^a(x_\bij),
\end{align*}
where the second line follows from~\eqref{eq:hat-s-xz-cons} in Lemma~\ref{lem:hat-s} and the definition~\eqref{eq:def-wl}, the third uses the
definition~\eqref{eq:def-hat-xz} of $\hat{X}$ from Claim~\ref{claim:strategies}, and the last uses the fact
that $\ket{\hat{\psi}}$ defined in Lemma~\ref{lem:hats} is an EPR pair on $\reg{A'A''}$ together
with~\eqref{eq:epr_stab} to set $b'' = c$ (so the equality holds in the state-dependent distance).

Finally we show the third item in the lemma. We show consistency for $\tilde{X}$; consistency for $\tilde{Z}$ is similar 
For ease of notation we write $g$ for $g_1$ and omit the outcome $g_2$, which in this argument is always
summed over. Using the definition~\eqref{eq:def-tilde-xz}, 
\begin{align}
\sum_e\, \tilde{X}_\ell^e(a)_\reg{AA'A''} \otimes \comp{\tilde{X}}_\ell^e(a)_\reg{BB'B''}
 &= \sum_{\substack{g,c,g',c':\\ \tr(b_\ell(g \cdot a))  - c=  \tr(b_\ell(g' \cdot a)) -c'}} \hS^g_\reg{AA'} \otimes \sigma_{X,\ell}^c(a)_\reg{A''} \otimes \comp{\hS}^{g'}_\reg{BB'} \otimes \comp{\sigma}_{X,\ell}^{c'}(a)^\dag_\reg{B''}\notag\\
&\approx_{\poly(\delta_S)} \Es{x\in\Fq^m}\, \sum_{\substack{g,c,g',c':\\\tr(b_\ell( (g-g')\cdot a)) = c-c'}}  \hS^g_\reg{AA'} \otimes \sigma_{X,\ell}^c(a)_\reg{A''} \otimes \comp{\hS}^{g'}_\reg{BB'} \otimes \comp{\sigma}_{X,\ell}^{c'}(a)_\reg{B''}\notag\\
&\hskip1cm		\cdot \Big(\sum_{\substack{ r + s=g(x)\\ r' + s'=g'(x)}}
	 ({M}_{X, x}^{r})_\reg{A} \otimes
  (\comp{\qp}_{X, x}^{s})_\reg{A'} \otimes
  (\comp{M}_{X, x}^{r'})_\reg{B} \otimes (\qp_{X, x}^{s'})_\reg{B'}\Big)\;,\label{eq:gen-xz-1}
	\end{align}
where the approximation uses~\eqref{eq:hat-s-xz-cons},~\eqref{eq:m-cons}, and~\eqref{eq:def-hat-xz} and~\eqref{eq:m-from-w}. 
 Using consistency of $M$ with $\comp{M}$, as
shown in Claim~\ref{claim:strategies}, 
additionally imposes the constraint that $r=r'$, i.e. $(g-g')(x)=s - s'$. Now, recall
that $\sigma^c_{X,\ell}(a)_\reg{A''}$ is implemented by measuring all $n$ qudits of register
$\reg{A''}$ in the $X$ basis, obtaining an outcome $h\in\Fq^n$, and reporting
$c=\tr(b_\ell(h\cdot a))$ as the outcome. Similarly, $(\qp^s_{X,x})_{\reg{A''}}$ is implemented by
measuring all $n$ qudits of register $\reg{A''}$ in the $X$ basis,
obtaining an outcome $h\in\Fq^n$, and reporting $s=h(x) = h \cdot x_\bij$ as the
outcome. Moreover, as the registers $\reg{A'A''}$ and $\reg{B'B''}$ are in
EPR states, we can apply the exact consistency relations~\eqref{eq:epr_stab}
between Pauli measurements. This
lets us rewrite~\eqref{eq:gen-xz-1} as 
	\begin{align}
	&\approx_{\poly(\delta_M)} \Es{x\in\Fq^m}\, \sum_{\substack{g,h,g',h':\\
    \tr(b_\ell (g-g')\cdot a) = \tr(b_\ell (h-h')\cdot a) \\ (g-g')(x)=(h-h')(x)}} \sum_{r = (g-h)(x)} \hS^g_\reg{AA'}  ({M}^{r}_{X,x})_\reg{A} \otimes \comp{\hS}^{g'}_\reg{BB'}  (\comp{M}^{r}_{X,x})_\reg{B}
	 \otimes   (\qp_X^{h})_\reg{A''} \otimes (\comp{\qp}_X^{h'})_\reg{B''}\notag\\
	&= \Es{x\in\Fq^m}\, \sum_{\substack{g,h,g',h':\\ \tr(b_\ell (g-g')\cdot a)
          = \tr(b_\ell (h-h')\cdot a)\\ (g-g')(x)=(h-h')(x)}} \sum_{\substack{ r =
          g(x)\\ r'=g'(x)}} \hS^g_\reg{AA'}  (\hat{M}^{r}_{X,x})_\reg{AA'} \otimes \comp{\hS}^{g'}_\reg{BB'}  (\comp{\hat{M}}^{r'}_{X,x})_\reg{BB'}
	 \otimes   (\qp_X^{h})_\reg{A''} \otimes (\comp{\qp}_X^{h'})_\reg{B''}\notag\\
	&\approx_{\poly(\delta_S)} \Es{x\in\Fq^m}\, \sum_{\substack{g,h,g',h':\\
    \tr(b_\ell (g-g')\cdot a) = \tr(b_\ell (h-h')\cdot a) \\ (g-g')(x)=(h-h')(x)}} \hS^g_\reg{AA'} \otimes \comp{\hS}^{g'}_\reg{BB'}  
	 \otimes   (\qp_X^{h})_\reg{A''} \otimes (\comp{\qp}_X^{h'})_\reg{B''}\;,\label{eq:gen-xz-2}
	\end{align}
	where the second line uses the definition of $\hat{M}$, and for the last
  approximation we removed the constraint on $r$
  using~\eqref{eq:hat-s-xz-cons}. If $g - g'\neq h - h'$ then by Lemma~\ref{lem:sz} the
  condition $(g-g')(x)=(h-h')(x)$ holds at a random point $x\in\Fq^m$ with
  probability at most $O(d/q)$. If $g - g'=h - h'$ the constraint $\tr(b_\ell (g-g')\cdot
  a) = \tr(b_\ell (h-h')\cdot a)$ is superfluous. Hence, under the expectation,
  we can replace the constraints in the summation in~\eqref{eq:gen-xz-2} by the
  constraint $(g - g')(x) = (h - h')(x)$, incurring an error of
  $O(d/q)$:
	\begin{align*}
	&\approx_{\poly(d/q)} \Es{x\in\Fq^m}\, \sum_{\substack{g,g',s, s':\\  (g-g')(x)=(s-s')}} \hS^g_\reg{AA'} \otimes \comp{\hS}^{g'}_\reg{BB'}  
	 \otimes   (\qp^{s}_{X,x})_\reg{A''} \otimes( \comp{\qp}^{s'}_{X,x})_\reg{B''}\\
  &\approx_{\delta_S} \Es{x\in\Fq^m}\, \sum_{\substack{e,e',s, s':\\
    e - e'=s-s'}} (\hat{M}^{e}_{X,x})_\reg{AA'} \otimes (\comp{\hat{M}}^{e'}_{X,x})_\reg{BB'}  
	 \otimes   (\qp^{s}_{X,x})_\reg{A''} \otimes(
    \comp{\qp}^{s'}_{X,x})_\reg{B''}\\
  &= \Es{x\in\Fq^m}\, \sum_{\substack{e,e',f, f', s, s':\\  e + f -(e' + f')=s-s'}}
    (M^{e}_{X,x})_\reg{A} \otimes (\comp{M}^{e'}_{X,x})_\reg{B}  
	 \otimes (\comp{\qp}^{f}_{X,x})_\reg{A'} \otimes
    (\qp^{s}_{X,x})_\reg{A''} \otimes (\comp{\qp}^{f'}_{X,x})_\reg{B'} \otimes(
    \comp{\qp}^{s'}_{X,x})_\reg{B''}\\
  &= \Es{x \in \Fq^m} \, \sum_{e} (M_{X,x}^e)_A \ot (M_{X,x}^{e})_B \\
  &\approx_{\delta_M} \Id\;,
\end{align*}
where the second line is by~\eqref{eq:hat-s-xz-cons}, the third by the
definition of $\hat{M}$ (Lemma~\ref{lem:hats}), the fourth by consistency of $\qp$ on the EPR state, and the last is
by self-consistency of $M$ (Claim~\ref{claim:strategies}). 
\end{proof}

\subsection{Proof of Lemma~\ref{lem:soundness}}
\label{sec:finish_soundness}
We conclude the proof of Lemma~\ref{lem:soundness}. Lemma~\ref{lem:xz-lowdeg} shows that whenever there exists a strategy for the provers that succeeds in the test $\qld(m,d,q)$ with probability at least $1-\eps$, for some $\eps$ small enough with respect to $d/q$ so that the quantity $\delta_W$ in the lemma is a small constant, there exist operators $\tilde{X}_\ell(a)$
  and $\tilde{Z}_\ell(b)$, for $a,b \in \Fq^n$, acting on the extended local spaces $\reg{AA'A''}$ and $\reg{BB'B''}$ respectively, that exactly satisfy the group relations of the generalized
  Pauli group (also known as the finite Heisenberg group). 
	
	Using these relations it is fairly straightforward to construct isometries $V_A$, $V_B$ acting on
  each prover's space such that $V_A$ maps
  $\tilde{W}_\ell(a)_\reg{AA'A''}$ to $\Id_\reg{AA'} \ot
  \sigma_{W, \ell}(a)_\reg{A''}$ for $W \in \{X, Z\}$, $\ell \in \{1, \dots, t\}$, and $a\in\Fq^n$, and likewise for $V_B$. 
	The definition of $V_A$ and $V_B$ is analogous, so we drop the subscript. 
	To explicitly define $V$, let $U$ be a unitary such that $U
  \sigma_{W,\ell}(a)^\dag U^\dag = \sigma_{W,\ell}(a)$ for all $a\in\Fq^n$, $\ell \in \{1,\dots,t\}$ and $W\in\{X,Z\}$; such a $U$ can be explicitly defined through its expansion in the Pauli basis. Define
  \[ V = \sum_{g_1 g_2} \hat{S}^{g_1, g_2}_\reg{AA'} \ot
    U \qp_X(g_2) \qp_Z(g_1)\;, \]
where in the notation $\qp_W(g)$ we interpret $g$ as the corresponding vector of
coefficients in $\Fq^n$.
  To see that this accomplishes the desired map, using that $\{\hat{S}^{g_1g_2}\}$ is projective, evaluate
  \begin{align*}
    V \tilde{X}_{\ell}(a)_\reg{AA'A''} V^\dag 
    &= \sum_{g_1, g_2} \omega^{\tr(b_\ell (g_1
    \cdot a))} \hat{S}_\reg{AA'}^{g_1, g_2} \ot U\qp_X(g_2)
    \qp_Z(g_1) \sigma_{X,\ell}(a)^\dag \qp_Z(g_1)^\dag
      \qp_X(g_2)^\dag U^\dag \\
    &= \sum_{g_1, g_2} \omega^{\tr(b_\ell (g_1
    \cdot a))} \hat{S}_\reg{AA'}^{g_1, g_2} \ot U\qp_X(g_2)
    \qp_Z(g_1) \qp_{X}(b_\ell a)^\dag \qp_Z(g_1)^\dag
      \qp_X(g_2)^\dag U^\dag \\
    &= \sum_{g_1, g_2} \hat{S}_\reg{AA'}^{g_1, g_2} \ot U
      \qp_X(a b_\ell)^\dag U^\dag
    \\
    &= \sum_{g_1, g_2} \hat{S}_\reg{AA'}^{g_1, g_2} \ot U
      \sigma_{X,\ell}(a)^\dag U^\dag \\
    &= \Id \ot \sigma_{X,\ell}(a)\;.
  \end{align*}
  A similar calculation can be done for $Z$.

The combination of all isometries considered in the analysis might have increased the local dimension by a large amount. However, using the outcome consistency between $\tilde{X}_\ell(a)$ and $X_\ell(a)$ (item 3. in Lemma~\ref{lem:xz-lowdeg}), we
  know that the state  $\ket{\hat{\psi}}$ satisfies 
  \begin{equation}
    \Es{\ell \in \{1, \dots, t\}} \Es{a \in \Fq^n } \,\sum_e \,\bra{\hat{\psi}} \tilde{X}_{\ell}^e(a) \ot
    \comp{\tilde{X}}_{\ell}^e(a)\ket{\hat{\psi}}\, \geq\, 1 - O(\delta_{W})\;,
    \label{eq:tilde-x-con}
  \end{equation}
and an analogous property also holds for the $\tilde{Z}_\ell(b)$. Let $H_W = \Es{\ell} \Es{a}
\sum_e \sigma_{W,\ell}(a)^e \ot \comp{\sigma}_{W,\ell}(a)^e$ for $W \in \{X,
Z\}$, and $\ket{\psi'} = V_A\otimes V_B \ket{\hat{\psi}}$. In this notation, we can rewrite~\eqref{eq:tilde-x-con} as 
  \begin{equation}\label{eq:hw-1}
    \bra{\psi'} H_W \ket{\psi'}  \,\geq\, 1 - O(\delta_W)\;, 
  \end{equation}
  for $W \in \{X, Z\}$. 
  By construction, the operators $H_W$ are Hermitian with $0
  \leq H_W \leq \Id$, and $H_X H_Z = H_ZH_X$. Hence, $H = H_X
  H_Z$ is Hermitian and $0 \leq H \leq \Id$. An application
  of the Cauchy-Schwarz inequality to~\eqref{eq:hw-1} yields
  \begin{equation}\label{eq:hw-2}
	\bra{\psi'} H \ket{\psi'} \geq 1 -
    O(\sqrt{\delta_W})\;. 
		\end{equation}
  Moreover, a direct calculation reveals that in fact $H = \ket{\psi_\epr}\bra{\psi_\epr}$. First, we evaluate $H_W$:
  \begin{align*}
    H_W &= \Es{a} \Es{\ell} \sum_e \sigma_{W, \ell}^e(a) \ot \comp{\sigma}_{W,\ell}^e(a)  \\
          &= \Es{a} \Es{\ell} \sum_{h, h': \tr(b_\ell h \cdot a) =\tr(b_\ell h' \cdot a) } \qp_W^h\ot \comp{\qp}_W^{h'} \\
          &= \Es{a}\Es{\ell} \sum_{h, h'} \omega^{\tr(b_\ell a \cdot (h - h'))} \qp_W^h \ot
            \comp{\qp}_W^{h'} \\
          &= \Es{a}\Es{\ell} \sigma_{W, \ell}(a) \ot \comp{\sigma}_{W,\ell}(a)^\dag \\
        &= \Es{\ell} \Big( \Es{a \in \Fq} (\sigma_{W,\ell}(a) \ot \comp{\sigma}_{W,\ell}^\dag(a))
            \Big)^{\ot n}\;, \\
        &= \Es{\ell} \Big( \Es{a \in \Fq} (\qp_{W}(a b_\ell) \ot \comp{\qp}_W^\dag(a b_\ell))
            \Big)^{\ot n}\;, \\
        &= \Big(\Es{a \in \Fq} (\qp_{W}(a) \ot \comp{\qp}_W^\dag(a))\Big)^{\ot n},
  \end{align*}
  where in going to the last line, we did a change of variables on the variable $a$, to absorb the factor of $b_\ell$. Now, we evaluate $H$:
  \begin{align*}
    H &= H_X H_Z \\
      &= \Big(\Es{a,b \in \Fq} \qp_X(a) \qp_Z(b ) \ot
            \comp{\qp}_X^\dag (a) \comp{\qp}_Z^\dag(b) \Big)^{\ot n} \\
      &= \Big(\Es{a,b \in \Fq} \qp_X(a) \qp_Z(b ) \ot
        \qp_X(a) \qp_Z(-b) \Big)^{\ot n} \\
      &= \Big(\Es{a,b \in \Fq} \sum_{i, j} \omega^{\tr(i b )}
            \ket{i + a}\bra{i} \ot \omega^{-\tr(j  b)} \ket{j+a}
            \bra{j} \Big)^{\ot n} \\
          &= \Big(\Es{a \in \Fq} \sum_{i} \ket{i+a}
            \bra{i} \ot \ket{i+a}\bra{i} \Big)^{\ot n} \\
          &= \ket{\psi_\epr}\bra{\psi_\epr}\;,
  \end{align*}
  where in going from the fourth to the fifth line, we have used the fact that
  the summation over $b$ vanishes unless $j = i$, and for the last we use that $\Es{a\in\Fq} \ket{i+a}\ket{i+a} =
  q^{-1/2}\ket{\psi_\epr}$ for any $i\in\Fq$. 
  Hence, from~\eqref{eq:hw-2} we conclude that
  \[ \big\| \braket{\psi_\epr}{\psi'} \big\|^2 \,\geq\, 1 - O(\sqrt{\delta_W})\;. \]
This completes the proof of Lemma~\ref{lem:soundness}.

\section{A test for codewords}
\label{sec:code}

In this section we show that the low-degree test $\qld$ can be combined with any weakly self-dual quantum CSS code $\mC$ defined over $\Fq$ to obtain a self-test for states in the $n$-fold tensor product of the codespace, as well as certain tensor products of generalized Pauli observables on the codespace (including all single-qudit Pauli observables). The test uses as many provers as the dimension of the code. 

\subsection{CSS codes}
\label{sec:codes}

We consider weakly self-dual  \emph{Calderbank-Shor-Steane (CSS)
  codes}~\cite{CalderbankShor96,Steane96}. Let $C$ be a classical $[k,k']$ linear error-correcting code over $\Fq$, for a prime power $q$: $C$ is specified by a generator matrix $G \in \Fq^{k\times k'}$ and a parity check matrix $K\in \Fq^{(k-k')\times k}$ such that $C = \ima(G)=\ker(K)$. We say that $C$ is weakly self-dual if the dual code $C^\perp$, with generator matrix $K^T$, is such that $C\subseteq C^\perp$; equivalently, $G^T G=0$. To any such  code $C$ we associate a subspace $\mathcal{C}$ of 
$(\C^q)^{\otimes k}$ that is the simultaneous $+1$ eigenspace of a set of stabilizers 
$\{S_{W,j}\}_{W\in\{X,Z\},j\in\{1,\ldots,k'\}}$ such that $S_{W,j}$ is
a tensor product of Pauli $W$ observables over $\Fq$ in the locations indicated by the
$j$-th column of the generator matrix $G$, i.e.
\[ S_{W,j} = \qp_W(G_{1j}) \ot \qp_W(G_{2j}) \ot \dots
  \ot  \qp_W(G_{kj}), \]
where $G_{ij}$ is the $(i,j)$-th entry of $G$.
The condition that $G^TG=0$ implies that all the $S_{W,j}$ commute, so that $\mathcal{C}$ is well-defined. We refer to~\cite{gottesman1999fault,ketkar2006nonbinary} for more background on the theory of stabilizer codes over qudits. 


\begin{example}[EPR code]
A simple example of a weakly self-dual $2$-qudit code with dimension
$1$ is the ``EPR code'' (our terminology) with generator matrix $G
= \begin{pmatrix} 1  \\ 1\end{pmatrix}$ in case $q=2$, and $G
= \begin{pmatrix} 1  \\  \sqrt{-1} \bmod q \end{pmatrix}$ for $q\equiv
1 \bmod 4$. The associated code subspace $\mC$ has dimension $1$, and it is spanned by a maximally entangled state of two qudits.
\end{example}

\begin{example}[Quadratic residue code]
\label{ex:quad_res_code}
The $7$-qudit code is a weakly self-dual CSS code that has $k=7$, $k'=3$, and encodes one qudit over $\Fq$ for any prime power $q=p^e$ such that $p$ is a quadratic residue modulo $7$. For example $p=2$ is a quadratic residue modulo $7$. See Theorem~40 in~\cite{ketkar2006nonbinary} for a more general construction.  
\end{example}

\subsection{The $\code$ test}
\label{sec:code-protocol}

Let $n$ be an integer and $C$ a weakly self-dual $[k,k']$ linear code. Let $\mC$ be the associated CSS code, as described in Section~\ref{sec:codes}.\footnote{All results in this and the next section can be obtained by restricting attention to the $7$-qubit code described in Example~\ref{ex:quad_res_code}.} The test $\code(C,n)$ is summarized in Figure~\ref{fig:code-test}.
The test builds on the (composed) low-degree test described in Figure~\ref{fig:protocol}. Recall the following properties of the honest strategy for the provers in the test (see Section~\ref{sec:lowdeg-completeness}):
\begin{itemize}[nolistsep]
\item In the first part of the test, each prover is sent a query of the form $(W,s,s')$, where $W\in\{X,Z\}$ designates a choice of basis and $s$, $s'$ are the specification of a pair of subspaces. The honest prover measures each of his $n$ qudits in the basis $W$, obtaining a string $a\in\Fq^n$. From $a$, the prover computes the degree-$d$ polynomial $g_a$ specified in~\eqref{eq:def-ga}, and returns the restriction of the (suitably encoded) bivariate polynomial $(g_a)_{|s}$ to the subspace $s'$.
\item In the second part of the test, the prover is sent a query of the form $(W_1,W_2)$, where $W_1,W_2\in\{X,Z,Y\}^n$ are commuting $n$-qudit observables. The honest prover jointly measures $W_1$ and $W_2$ and returns the pair of outcomes obtained. 
\end{itemize}
We now describe the test $\code(C,n)$. In the test, the verifier splits the $k$
provers into two groups. One prover, indexed by $t\in\{1,\ldots, k\}$, is chosen
at random and called the ``special prover''. The remaining $(k-1)$ provers are
jointly called ``composite prover''. In general a
prover is not told whether it is the special prover, or a composite prover. In the test the
verifier simulates queries from the two-prover low-degree test for the special
and composite provers using the following  scheme. 

\begin{definition}[Composite queries and answers]\label{def:queries}
Let $G \in \Fq^{k\times k'}$ be the generator matrix for a $[k,k']$ weakly self-dual code $C$. 
Let $Q$ be a query in the test $\qld$.
\begin{enumerate}
\item The composite query associated with $Q$, denoted $\comp{Q}$, is
  obtained by sending each prover forming the composite prover the query $Q$.
\item 
Given answers $(A_j)_{j \neq t}$ from the $(k-1)$ provers
  forming the composite prover, the composite answer $\comp{A}$ is
  obtained by selecting a uniformly random vector $v$ in the column
  span of $G$ such that $v_t=1$, and computing the sum $\comp{A} =
  - \sum_{j \neq t} v_j A_j$. \footnote{The specific way in which this summation is performed depends on the form of
the query $Q$. In general each $A_i$ is expected to be
either a low-degree polynomial, or of a pair of values in $\Fq$. In
both cases, there is a natural way to add up the answers in order to
obtain an answer $\comp{A}$ that is formatted as the prover's
answer in the low-degree test. }
\end{enumerate}
\end{definition}

This definition is consistent with the notation
$\comp{M}$ used in $\qld$; in both cases, the answers
obtained from the composite prover (in the case of the two-player
test, the second prover) are multiplied by the appropriate entry of the
generator matrix of a code. The test $\qld$ differs only insofar as the
EPR state is not a CSS code state, so the $X$ and $Z$ stabilizers are
not identical. Moreover, in both cases, for honest strategies, the special and
composite prover obtain the same outcome when given the same query. This fact is formalized
in the following lemma.

\begin{lemma}\label{lem:fq-linear}
Let $\ell \geq 1$ be an integer and $f: \Fq^n \to \Fq^\ell$ a linear function. Suppose that $k$ provers
share a $(nk)$-qudit state $\ket{\Psi}$ that is a valid qudit-by-qudit encoding of an $n$-qudit state $\ket{\psi}$ according to a $k$-qudit self-dual CSS code $\mathcal{C}$. Let $W\in\{X,Z\}^n$ and suppose that for each $j\in\{1,\ldots,k\}$, the $j$-th prover measures the $i$-qudit of its share of the state in the basis $W_i$, for each $i\in\{1,\ldots,n\}$, to obtain an assignment $a_j \in \Fq^n$, and
returns the value $y_j = f(a_j) \in \Fq^\ell$. Then for any index $t\in\{1,\ldots,k\}$ for the special
prover, and vector $v\in\Fq^k$ chosen as in item 2. in Definition~\ref{def:queries}, the special prover's answer $y_t$ and the composite prover's answer
$\comp{y} = -\sum_{j \neq t} v_j y_j$ are equal with certainty.
\label{lem:composite-linear}
\end{lemma}

Before giving the proof, we note that the functions computed in the low-degree
test, i.e. polynomials $g_a$ as in~\eqref{eq:def-ga}, evaluated on points or restricted to  subspaces, are linear functions of $a$ of the form considered in the lemma.

\begin{proof}
 It follows from the definition of $v$ and the stabilizer property of the code that
  \[ \sum_{j} v_j a_j  = 0\;. \]
  Write the linear function $f$ as $f(a) = K a $, for $K\in\Fq^{\ell\times n}$. Then, a simple
  calculation shows that
\[
    \comp{y} = -\sum_{j \neq t} v_j \,f(a_j)  = - \sum_{j \neq t} v_j \,( K a_j)  = K\Big(- \sum_{j \neq t} v_j \,a_j\Big)  = K a = y\;.
\]
  \end{proof}

\begin{figure}[htbp]
\rule[1ex]{16.5cm}{0.5pt}\\
Test~$\code(C,n)$:
Given is the generator matrix $G$ for a $[k,k']$ weakly self-dual
linear code $C$ over $\Fq$, as described in Section~\ref{sec:codes},
and $n$ an integer. Let $d = \lceil \log n \rceil \cdot \lceil
\frac{\log n}{\log \log n} \rceil$ and $m =\lceil \frac{\log(n)}{\log\log(n)} \rceil$. 
\begin{itemize}
\item[(a)] The verifier selects one of the $k$ provers at random and assigns it the label of ``special prover''. All remaining provers are given the label of ``composite prover''. (The provers are not told how they are labeled.) Let $t\in\{1,\ldots,k\}$ be the index of the special prover. 
\item[(b)] The verifier executes the verifier for the test
  $\qld^{(2)}(m,d,q)$ to generate a pair of queries $(Q,Q')$ for the two
  provers in that test. The verifier sends the query $Q$ to the
  special prover, and distributes the query $\comp{Q'}$ to the composite
  prover. He receives answers $A$ and $\comp{A'}$ respectively.  
\item[(c)] The verifier accepts if and only if $(A,\comp{A'})$ is a
  pair of valid answers to queries $(Q,Q')$ in the low-degree
  test. 
\end{itemize}
\rule[1ex]{16.5cm}{0.5pt}
\caption{The procedure $\code(C,n)$ verifies that $k$ provers share an entangled state which lies in the $n$-fold tensor product of the code $\mathcal{C}$, defined over $k$ qudits each of dimension $q$.}
\label{fig:code-test}
\end{figure}

\subsection{Analysis of the $\code$ self-test}

We first show completeness of the test $\code$. 

\begin{lemma}\label{lem:code-completeness}
Let $C$ be a $[k,k']$ weakly self-dual linear code over $\Fq$, and $n$ an integer. Let $\mC$ be the associated CSS code, as described in Section~\ref{sec:codes}.  Then for any $(nk)$-qubit state $\ket{\Psi}\in \mC^{\otimes n}$ there exists a strategy for the $k$ provers based on sharing $\ket{\Psi}$ and measuring tensor products of Pauli observables, such that the strategy is accepted with probability $1$ in the test $\code(C,n)$. 
\end{lemma}

\begin{proof}
Fix $\ket{\Psi}\in \mC^{\otimes n}$. 
The strategy for the provers is simple: each prover
directly applies the honest strategy in the test $\qld^{(2)}$, as
described in Section~\ref{sec:lowdeg-completeness}. 

It remains to verify that the answers $(A,\comp{A'})$ computed by the
verifier in step (c) of the test $\code$ are valid answers to $(Q,Q')$
in $\qld^{(2)}$. Fix a choice of codeword $v$ as made by the verifier in the
computation of the composite query $\comp{Q'}$ at step (b) of
$\code$ (see Definition~\ref{def:queries}). 
We make the following observations on the joint measurement performed
by the provers that constitute the composite prover. Consider first
queries of the form $Q'=(W,s,s')$. Upon receipt of such a query,
the $i$-th prover that constitutes the composite prover measures each
of its $n$ qudits using the projective measurement $\qp_{W}$ to obtain outcomes
$a'_i=(a'_{i,1},\ldots,a'_{i,n})$, from which it computes a low-degree
polynomial $g_{a'_i}$ as in~\eqref{eq:def-ga}. Since $a'\mapsto
g_{a'}$ is a linear function, the answer $\comp{A'}$ computed
by the verifier is the restriction to $s'$ of the (suitably encoded) bivariate polynomial $(g_{a'})_{|s}$, where $a' = \sum_{i \neq t} v_i a'_i$ is the outcome of an
imaginary joint measurement performed by the composite prover using
the measurement $\comp{\qp}_{W}^{a'} = \sum_{a'_i : \sum_{i \neq t} v_i
  a'_i = a'} \otimes_{i\neq t} \qp_{W}^{a'_i}$. The
situation in case the query $Q'$ is taken from the second part of the
low-degree test is similar.


	To conclude we show that for any choice of codeword $v$ made by the verifier, the provers' strategy, conditioned on $v$, is isometric to the honest strategy for the low-degree test, as defined in the proof of Lemma~\ref{lem:completeness}. 
	
	To see this, observe that by definition, for a fixed $v$, the operators $X \otimes \comp{X}$ and $Z \otimes \comp{Z}$ stabilize each group of $k$ qudits of $\ket{\Psi}$, where $X=\qp_X(1)$, $\comp{X} = \otimes_{i\neq t} \qp_X(v_i)$, and $Z=\qp_Z(2)$, $\comp{Z} = \otimes_{i\neq t} {\qp_Z}(v_i)$; indeed this is because $v$ defines both an $X$ and a $Z$ stabilizer for $\mC$. Moreover, $\comp{X}$ and $\comp{Z}$ satisfy the same twisted commutation relation as $\comp{\qp}_X$ and $\comp{\qp}_Z$; this is because $v_t=1$ and $v\cdot v=0$ by weak self-duality. Thus there exists a local isometry acting jointly on all provers forming the composite prover, which maps $\comp{X}\mapsto \comp{\qp}_X$ and $\comp{Z}\mapsto \comp{\qp}_Z$. The image of $\ket{\Psi}$ under this isometry is stabilized by $\qp_X\otimes \comp{\qp}_X$ and $\qp_Z \otimes \comp{\qp}_Z$, hence must be the state $\ket{\epr_q}$. Lemma~\ref{lem:completeness} then lets us conclude that the above-defined strategy succeeds with probability $1$ in the test.

\end{proof}

The next theorem shows soundness of the test $\code$. 

\begin{theorem}\label{thm:codeword_test}
Let $n,k,k'$ be integer. Let $q=p^t$ be a prime power such that $\Fq$ admits a self-dual basis over $\Fp$. Let $C$ be a $[k,k']$ weakly self-dual linear code over $\Fq$, and $\mC$ the associated CSS code. Let $m,d$ be as in Figure~\ref{fig:code-test}. Suppose a strategy using state
$\ket{\Psi} \in \otimes_{i=1}^k \mH_{i}$ and projective
measurements $\{M_{W,w}^a\}$ for the special prover succeeds in test $\code(C,n)$ with probability at
least $1-\eps$, for some $\eps \geq 0$. Then there is a $\delta_C = \max(\poly(p)\cdot\poly(\eps),\poly(q^{-1}))$ and isometries $V_i: 
\mH_{i} \to (\C^q)^{\otimes n}\otimes \mH'_{i}$ for
$i\in \{1,\ldots,t\}$, and states $\ket{\psi}\in\mC$ and $\ket{\aux}\in \otimes_i \mH'_{i} $  such that  
$$ \big\| (\otimes_i V_i)\ket{\Psi} -  \ket{\psi} \ket{\aux} \big\|^2 \,\leq\, \delta_C\;,$$
and for all $W\in \{X,Z\}$, 
$$
\Es{w\in\Fq^m} \,\sum_{a\in \Fq}\, \big\|(\otimes_i V_i )(M_{W,w}^a \ot \Id) \ket{\Psi} -
  (\qp_{W,w}^a \ot \Id)\ket{\psi} \ket{\aux} \big\|^2 \, \leq \,
  \delta_C\;.$$
\end{theorem}

\begin{proof}
Fix a strategy for the $k$ provers in the test that is accepted with probability
at least $1-\eps$. Fix any $t\in\{1,\ldots,k\}$. By combining the $(k-1)$
strategies employed by provers $\{1,\ldots,k\}\backslash\{t\}$, when they play
the role of the composite prover, into a single strategy, we obtain a two-prover
strategy for the test $\qld^{(2)}(m,d,q)$ that has success probability at least
$1-\eps$. Applying Theorem~\ref{thm:qld} shows the self-testing claim for the
observables applied by prover $t$, when designated as the special prover. The
same applies for all $t\in\{1,\ldots,k\}$, proving the theorem.  
\end{proof}

\begin{remark}\label{rm:code_check_bits}
 We record here the bit complexity of the
protocol $\code$. The test invokes the composed quantum low-degree test $\qld^{(2)}(m,d,q)$ with $m = \lceil \frac{\log
  n}{\log \log n} \rceil$ and $d = \Theta(
\frac{\log^2 n}{ \log \log n} )$. Hence, the number of bits in the verifier's questions
scales as $O(\frac{\log n}{\log \log n} \log q)$, and
the number of bits in the provers' responses scales as $O((\log \log
n)^2 \log q)$, so the overall bit complexity is $O(\frac{\log n}{\log
  \log n} \log q)$. 
\end{remark}

\section{Energy tests}
\label{sec:testing}

In the previous section we gave a test that enforces that the provers' shared state is
close to a valid code state of an error-correcting code. In this section, we
 show how to further test any property of the encoded state that can be
expressed in terms of a local Hamiltonian of the appropriate form. We achieve 
this by using interactive protocols to ``command'' the provers to measure a subset of the
terms of a given Hamiltonian, perform a computation on the measurement outcomes,
and return the result. We introduce the required tools from the classical PCP literature
in Section~\ref{sec:delcomp}, and adapt them to our setting in Section~\ref{sec:sum-test} and  Section~\ref{sec:multi-basis}. In
Section~\ref{sec:constant_gap} we give a protocol to estimate the ground energy of a
(not necessarily local) Hamiltonian up to constant precision, provided that the
terms of the Hamiltonian have a certain form. As a consequence, we 
show that is $\QMA$-hard to approximate the maximum success
probability of a nonlocal game (i.e. one-round $\MIP^*$ protocol) with logarithmic communication, either conditionally on the
constraint satisfaction quantum PCP conjecture
(Corollary~\ref{cor:qma-generalizedXZ}), or unconditionally, but under
randomized reductions (Corollary~\ref{cor:randomized}). In
Section~\ref{sec:ff} we use similar tools to give a protocol to estimate the ground energy of
a  class of (not necessarily local) frustration-free Hamiltonians up
to inverse polynomial accuracy.

As a note on terminology, in previous sections, we introduced ``tests'' for states with certain properties, such as of being a valid codestate. In this section we provide tests for states that encode the answer to a computational problem (in particular, variants
of the local Hamiltonian problem). To formulate these tests we use the language of interactive proofs, and often refer to a test for a property as an $\MIP^*$ protocol for the
corresponding computational problem. The notions of a test, a nonlocal
game, and a one-round $\MIP^*$ protocol are roughly synonymous, and their meaning should always be clear from context.

\subsection{Classical PCPs for linear functions}
\label{sec:delcomp}

The codeword test introduced in Section~\ref{sec:code} gives the verifier the ability to
query a prover for a location in the low-degree encoding of the string of
outcomes obtained by the prover when measuring all its qudits in the Pauli $X$ or $Z$ basis. For our applications, we would like to have the ability to 
command the provers to compute more general functions of their measurement
outcomes. For example, upon obtaining an outcome $a \in \Fq^n$, we may want to ask the prover
to compute the inner product $a \cdot b$ with a
given string $b \in \Fq^n$ (that may not necessarily correspond to an entry in the low-degree encoding of $a$), agreed on in
advance between the prover and the verifier. One approach to doing this is to
use the sum-check protocol of~\cite{lund1992algebraic}, but this requires a logarithmic number of rounds of interaction, resulting in a polylogarithmic number of bits of
communication. To achieve a protocol with logarithmic communication we rely on the following classical PCP construction, that can be extracted from~\cite{BGHSV05}.

\begin{theorem}\label{thm:composedpcp}
Let $p = 2$, $n,t\geq 1$ integer such that $t=\Theta(\log\log n)$, and $q = 2^t$. For any $a, b \in \Fq^n$, there exists a proof
$\Pi_{a,b} \in \F_q^{n'}$, with $n' = O(\poly(n))$, each of whose bits is an $\F_q$-linear function of $a$, such that the following holds. There exists an efficient test $\lin(b)$ that uses
$O(\log n)$ random bits and reads a total of $O(1)$ bits from
$\Pi_{a,b}$ and from the evaluation table of the low-degree extension $g_a$ of $a$, as
well as a value $c\in\Fq$, with the following properties:

\begin{enumerate}
\item Completeness: If $b \cdot a = c$ the test accepts with certainty.
\item Soundness: If $b \cdot a \neq c$, for any claimed proof $\Pi$, the test rejects with constant probability.
\end{enumerate}
\end{theorem}

\begin{proof}
  We use the language of ``PCPs of proximity'' from~\cite{BGHSV05}. A PCP of
  proximity (PCPP) consists of an algorithm $V$ to verify that a given input $a \in
  \{0,1\}^n$ (called the \emph{assignment}) satisfies a \emph{property}
  specified by a poly-sized Boolean circuit. The verifier $V$ is given access to $a$ and to an
  auxiliary proof string $\Phi$ of polynomial length, but is only allowed to
  query a small (e.g. constant) number of bits of $a$ and $\Phi$. The
  completeness and soundness requirements on the verifier are that whenever $a$
  satisfies the property (the YES case), there exists a proof $\Phi_a$ that convinces the
  verifier $V$ to accept with certainty, and when $a$ is $(\delta n)$ far in Hamming distance from
  any string $a'$ satisfying the property (the NO case), then for all proofs $\Phi$, the
  verifier $V$ rejects with at least constant probability. The parameter $\delta$ is
  called the \emph{proximity parameter} of the PCPP. 

  From~\cite[Theorem 3.3]{BGHSV05}, with the parameter $t$ appearing in
  that theorem\footnote{Not
    to be confused with our parameter $t$, which controls the field
    size!} chosen to be a constant greater than $3$, and for a
  proximity parameter $\delta = \Theta(1/t)$, there exists a PCPP for properties encoded
  by circuits of size $O(n)$, consisting of a proof
   $\Pi$ and verifier $V$ that uses $O(\log^{2/t} n)$
  random bits, reads $O(t) = O(1)$ bits of the proof and
  assignment, and in the NO case rejects with probability
  $\Omega(1/t)$.  Moreover, from the discussion in Section~8.4 of~\cite{BGHSV05} it follows that if the property
  can be expressed as an AND of linear constraints (i.e. of the form $b \cdot a
  = c$ over $\F_2$), then the bits of the proof string $\Phi$ are $\F_2$-linear functions of the
  assignment, and the checks performed by the verifier $V$ are $\F_2$-linear
  constraints on $\Phi$. 
  
Ideally, we would like to directly apply this PCPP to check that the string
  $a \in \Fq^{n}$, interpreted as a bit string in $\{0,1\}^{nt}$,
  satisfies the condition $ b \cdot a = c$, which is a linear
  condition over $\Fq$. To do this, we need to address two
  issues. First, the result of~\cite{BGHSV05} applies to $\F_2$-linear
  conditions, and the proof string $\tilde{\Pi}$ is an $\F_2$-linear function
  of the input, whereas the present theorem requires linearity
  over $\Fq$. We resolve this by noting that the $\Fq$-linear
  condition $b \cdot a = c$ can be expressed as a conjunction of $\F_2$-linear conditions $\tr[(b \cdot a) \chi_i] = \tr[c \chi_i]$ for
  every element $\chi_i$ of a self-dual basis for $\F_q$ over
  $\F_2$. Moreover, any $\F_2$-linear function $f: \F_2^{tn} \to \F_2$
  can be expressed a function $f: \F_q^{n} \to \F_2$ of the form
  $\tr[u \cdot a]$ for some $u \in \F_q^{n}$, and hence can be
  extended to an $\F_q$-linear function $\ul{f} : \F_q^n \to \F_q$
  given by $\ul{f} = u \cdot a$. Applying this extension to each bit
  of $\Phi$, we can construct a proof $\Pi_{a,b}  \in \Fq^{\poly(n)}$,
  such that each entry of $\Pi_{a,b}$ is an $\Fq$-linear function of $a$,
  and from which the PCPP verifier can recover the original proof
  $\Phi$ by taking the trace of each entry. Since $\tilde{\Pi}$ could be verified
  by querying a constant number of bits, $\Pi$ can be verified by
  querying a constant number of $\Fq$-valued entries.

  The second issue concerns the soundness of the proof system. The statement of the present theorem requires soundness to
  hold against all non-satisfying $a$, not just those that satisfy the promise
  of the PCPP. Thus, instead of applying the PCPP directly to $a$, we apply it to the evaluation table of the
  low-degree encoding $g_a$ of the assignment, which has length
  $O(n \log n)$. The condition $b \cdot
  a = c$ can be expressed as a linear condition on the evaluation
  table of $g_a$. Moreover, if $b \cdot a \neq c$, then the encoding $g_a$ 
  differs from the encoding $g_{a'}$ of any $a'$ such that $b \cdot a' = c$
  on at least a constant fraction of positions. (This follows from the
  Schwartz-Zippel lemma: if $a \neq a'$, then $g_a - g_{a'}$ is a nonzero
  polynomial and hence by Lemma~\ref{lem:sz}, it cannot be $0$ on
  more than $d/|\Fq|$ fraction of the points.) Hence, the soundness promise
  holds on the encoded input $g_a$. Finally, to check that the part of the proof that corresponds to $g_a$ is a valid low-degree encoding, the verifier executes a standard low-degree test (such as the PCP version of the test $\cld$); this can be done using a constant number of additional queries to the evaluation table of $g_a$ (provided restrictions to lines and planes are included). 
\end{proof}

\subsection{A test for non-local Pauli observables}
\label{sec:sum-test}

Theorem~\ref{thm:composedpcp} specifies a test certifying that $a \cdot b =
c$, given $O(1)$ queries to a proof $\Pi_{a,b}$ whose bits are $\F_q$-linear
functions of $a$. Based on this test, in Figure~\ref{fig:sumgame} we give a
multiprover protocol $\sumgame(C,W,b)$ in which the verifier commands one of $k$ provers (supposedly) sharing an encoding of a state $\ket{\psi}$ according to $\mathcal{C}$ to
measure their share of the state in a specified basis ($X$ or $Z$) to obtain an
assignment $a$, and report the value $a \cdot b$. The verifier checks that this
value was computed correctly by using the test $\lin(b)$ from Theorem~\ref{thm:composedpcp}, together with the
guarantees of the low-degree test. The test $\lin(b)$ requires a constant number of queries to both $g_a$ and $\Pi_{a,b}$. In order to aggregate these queries, the verifier first asks the provers to encode $\Pi_{a,b}$ as a low-degree
polynomial $h$; to query a constant number of entries of $\Pi_{a,b}$ the verifier then asks for the restriction
$h|_s$ of $h$ to a curve $s$ that goes through all points to be queried. However, the number of bits required
to specify the restriction $h|_s$ is, for our choice of parameters,
polylogarithmic in $n$. To get around this we apply composition, in a similar way to the composed low degree test (Theorem~\ref{thm:2ml}). For concreteness,
we explicitly state how to do this, following the variable substitution technique
in~\cite[Section 3.1.2]{Vidick13xor}, which appeared earlier in~\cite[Section 4.4]{DFKRS11}.

\begin{definition}\label{def:variable-substitution}
  Define the substitution map 
	\begin{equation}\label{eq:substitution}
	\#_d : \Fq \to \Fq^{\mu(d)}\;, \quad \#(x) = (x, x^2, x^4, \dots, x^d)\;,
	\end{equation}
  where $\mu(d) = 2 \lceil \log (d+1) \rceil$. 
  \label{def:variable-substition}
\end{definition}

Under this map, any univariate degree-$d$ polynomial $f(x)$ can be viewed as a
degree $\delta(d) = O(\log d)$, $\mu(d)$-variate polynomial $f(\#x)$, by
formally identifying powers of $x$ with products of the substituted
variables. Thus, instead of querying for the restriction $h|_s$, we view
this restriction itself as a multivariate polynomial over $\Fq^{\mu(d)}$, and
query the restriction of that polynomial to a curve over $\Fq^{\mu(d)}$. This
can be described in logarithmically fewer bits. The precise form of the queries
we make to the prover is specified in Figure~\ref{fig:sumgame}.

\begin{figure}[htbp]
\rule[1ex]{16.5cm}{0.5pt}\\
Test $\sumgame(C,W, b = \{b_1,\ldots,b_k\})$:  \\
The verifier sends the basis label $W$ to all provers. In the test, the verifier sends pairs of questions, generally formatted as in the test $\qld^{(2)}(m,d,q)$ and $\cld^{(2)}(m',d,q)$. We write the first question as $s$, and the second as $s'$. Note that, as in $\qld(m,d,q)$, $s$ (resp. $s'$) itself can consist of a point in $\Fq^m$ (resp. $\Fq^{m'}$), or a pair $(s_1,s_2)$ (resp. $s'_1,s'_2$) of subspaces.  
\begin{enumerate}
\item[(a)] Send the special prover a question $s=(s_1,s_2)$ distributed as in
  $\qld^{(2)}(m,d,q)$ (conditioned on the basis choice $W$ having been made). Receive
  a polynomial $r\in\deg_d(s_2)$. Send the composite prover the composite
  query $\comp{(b_{t'}, (s, s'))}$ consisting of the vector $b_{t'}$ as well as a
  question $(s,s')$, where $s'$ is distributed as in $\cld^{(2)}(m',d,q)$. Receive the
  composite answer $\comp{(r',r'')}$. Reject if $r \neq \comp{r'}$. 
\item[(b)] Send $b_{t}$, where $t$ is the index of the special prover, to both provers. Execute the tests $\qld^{(2)}(m,d,q)$ and $\cld^{(2)}(m',d,q)$ in parallel. Accept if and only if the provers' answers pass both tests. 
\item[(c)] Send $b_{t'}$ to both provers. Simulate the test $\lin(b_{t'})$ from
  Theorem~\ref{thm:composedpcp} to obtain a tuple
  $(x_1,\ldots,x_\ell,i_1,\ldots,i_{\ell'})$ of queries, where $x_i \in \Fq^m$
  are queries to $g_a$ and $i_j$ are indices of bits to be queried in the PCPP proof
  $\Pi_{a,b}$. Let $s_1$ be a constant-degree curve in $\Fq^m$ that goes
  through all the $x_i$, i.e. a constant-degree polynomial $s_1: \Fq \to \Fq^m$
  whose image contains each point $x_i$, and likewise $s'_1$ a constant-degree curve in $\Fq^{m'}$ that
  goes through all $\bij(i_{j})$. Moreover, let $s_2$ be a constant-degree curve
  in $\Fq^{\mu(d)}$ that goes through all of the points
  $\#(s_1^{-1}(x_i))$, and let $s'_2$ be a constant-degree curve in
  $\Fq^{\mu(d')}$ that goes through all of the points
  $\#((s'_1)^{-1}(\bij(i_{ij})))$, where $\#$ and $\mu(\cdot)$ are as in Definition~\ref{def:variable-substitution}. Perform one of the following tests with 
  probability $\frac{1}{2}$ each.
  \begin{enumerate}
  \item[(i)] Choose uniformly random points $y$ on $s_1$ and $y'$ on $s'_1$. Send $(y, y')$ to
    the special prover and $\comp{((s, \#y), (s', \#y'))}$ to the composite prover. Receive  answers $(\alpha, \beta) \in \Fq^2$ and $\comp{(\gamma, \delta)}
    \in \Fq^2$, respectively. If $\alpha = \comp{\gamma}$ and $\beta =
    \comp{\delta}$, then accept, else reject.
  \item[(ii)]
    Send $((s_1,s_2), (s'_1, s'_2))$ to the special prover, and two points
    $\comp{(z,z')}$ to the composite prover, where $z$ is uniformly random in $s_2$ and $z'$ is
    uniformly random in $s'_2$. Receive from the special prover a pair of
    polynomials $(r, r')$, where $r$ is $\mu(d)$-variate and $r'$ is
    $\mu(d')$-variate, and from the composite prover a pair of values
    $\comp{(\alpha, \beta)} \in \Fq^2$. If the answers are consistent on
    points $z, z'$ (i.e. $r(z) = \comp{\alpha}$ and $r'(z') = \comp{\beta}$) and if the entries of $\Pi$ and $g$ decoded from the answers $r$ and $r'$
     on $s$ and $s'$ would be accepted in the test $\lin(b)$, then
     accept, else reject. \\
  \end{enumerate}
\item[(d)] For $j\in\{1,\ldots,k\}$, send the $j$-th prover the vector $b_j$ and a query $((s_1,
  s_2), (s'_1, s'_2))$ chosen as in part (c). Receive from each
  prover a value $c_j\in\Fq$, as well as a pair of polynomials $(r_j, r_j')$. If for each prover $j$,
  the entries of $\Pi$ and $g$ decoded from $r_j$ and $r_j'$ would be
  accepted in the test $\lin(b_j)$, then return $\omega^{\tr[\sum_{j=1}^{k} c_j]}$.
\end{enumerate}
\rule[1ex]{16.5cm}{0.5pt}
\caption{Procedure $\sumgame(C,W,b)$, where $C$ is a self-dual CSS code, $W\in\{X,Z\}$ a basis label, and $b=\{b_1,\ldots,b_k\}$, where each $b_j\in\Fq^n$.}
\label{fig:sumgame}
\end{figure}

Before stating the completeness and soundness properties of the test $\sumgame(C,W,b)$, we state the strategy followed by the honest prover.

\begin{definition}\label{def:sumgame-honest}
In the game $\sumgame(C,W,  b=\{b_1,\ldots,b_k\})$, the honest strategy is defined as follows:
\begin{itemize}
  \item State: the provers share a state $\ket{\Psi}$ which is a
    qubit-by-qubit encoding of a state $\ket{\psi} \in (\C^q)^{\otimes
      n}$ using the CSS code $\mathcal{C}$ derived from the self-dual classical code $C$.
  \item Measurements: each prover performs a measurement of each of its qudits in the basis $W$, to obtain as
outcome a string $a\in\Fq^n$. Using $a$, the prover determines a polynomial $g_a : \Fq^m \to \Fq$. In addition, the prover
may be sent a vector $b_j\in\Fq^n$. If this is the case, the prover replies with
$c=a\cdot b_j\in\Fq$. In addition, it computes a polynomial $h_a : \Fq^{m'}\to \Fq$ that
is a low-degree extension $h_a = g_{\Pi_{a,b_j}}$ for the PCPP proof
$\Pi_{a,b_j}$ that verifies $c=a\cdot b_j$, as described in
Theorem~\ref{thm:composedpcp}. Finally, the prover is sent a pair 
$(s,s')$. Here $s$ (resp. $s'$) may be: a question from the tests $\qld^{(2)}$ (resp.
$\cld^{(2)}$), or, in part (c) of the test, a
single point in $\Fq^{m}$ (resp. $\Fq^{m'})$, or a specification of a curve of
constant degree in $\Fq^m$ (resp. $\Fq^{m'}$) together with a point or a curve
in the space $\Fq^{\mu(d)}$ (resp. $\Fq^{\mu(d')}$). The honest prover responds with the
appropriate restriction of $g_a$ (resp. $h_a$), composed with
the variable substitution map whenever appropriate. 
\end{itemize}
\end{definition}

\begin{theorem}\label{thm:sum-game}
 Let $C$ be a
 $[k,k']$ weakly self-dual linear code. Let $n$ be an integer, and $d,m,q$ integer such that $q = 2^t$ and $d,m,\log q$
 are polynomially bounded in $n$. Let $\eps \geq d/q$. 
Let $b_1, \dots, b_k \in\Fq^n$. The procedure $\sumgame(C,W, \{b_1,
\dots, b_k\})$
 is a 1-round, $k$-prover interactive protocol with the following completeness
 and soundness properties.
\begin{enumerate}
  \item \emph{Completeness:} If the provers follow the strategy introduced in
    Definition~\ref{def:sumgame-honest}, they pass the test with certainty.
  \item \emph{Soundness:}
    Suppose that a strategy for the provers is accepted with
    probability at least $1-\eps$ in $\sumgame(C,W, \{b_1, \dots, b_k\})$. Suppose further that the restriction of the strategy to questions formatted as in $\qld^{(2)}(m,d,q)$ succeeds in that test with probability at least $1-\eps$.
    Then there is a $\delta = \poly(\eps, \delta_{C})$, where $\delta_{C}$ is as
    specified in Theorem~\ref{thm:codeword_test}, such that the following holds. There exists an encoded state $\ket{\Psi} \in
    (\C^q)^{\otimes nk}$ such that the value returned by the verifier in step
    (d) of the protocol has expectation that is within $\delta$ of 
    $\bra{\Psi}\otimes_{i=1}^{k} \qp_{W}( b_i)\ket{\Psi}$. 
\end{enumerate}
\end{theorem}

\begin{remark}\label{rk:sum-complexity}
  The number of bits communicated to a prover in $\sumgame$ is at most the
  number of bits needed in the (composed) low-degree test, plus the number of bits needed
  to specify a constant-degree curve over $\Fq^n$. By
  Theorem~\ref{thm:qld}, both are at most $O(m \log q) = O(\frac{\log n}{\log \log n} \log q)$. The number of bits in
  the provers' answers is at most the maximum of the number of bits needed in
  the (composed) low-degree test, and the the number of bits needed to specify a
  degree-$\delta(d) = O(\log d)$
  polynomial restricted to a constant-degree curve. By Theorem~\ref{thm:qld}, the
  former is at most $O(\log^2(d) \log(q))$, while the latter is
  at most $O(\log(d)\log(q))$.
\end{remark}

\begin{proof}
Completeness follows from the definition of the honest strategy, Lemma~\ref{lem:composite-linear}
and the completeness of $\lin(b)$ as described in Theorem~\ref{thm:composedpcp}.

We show soundness. A strategy for a prover in the test $\sumgame(C,W,b)$ consists of a family of
measurements $\{M_{b,s,s'}^{r,r'}\}$ used in response to questions of the form
$(s,s')$. As the subscripts indicate, these measurements depend on the
vector $b$ received by the prover. In addition, part (a) of the test
involves cross-checking these measurements with a strategy for the quantum low-degree
test $\qld^{(2)}(m,d,q)$, in which the players are \emph{not} told $b$. The strategy
used for this test is described by measurements $\{N_{s}^{r}\}$, which do not
depend on $b$.

We show the soundness of the test in two steps. First, we note that
success in parts (a) and (b) of $\sumgame(C,W, \{b_1, \dots, b_k\})$ implies that the measurements used
by the players are close to product form:
\[ M_{b,s,s'}^{rr'} \approx_{\poly(\eps)} A_{b,s}^r B_{b,s'}^{r'}, \]
where $A_{b,s}^r = \Es{s'}\sum_{r'} M_{b,ss'}^{rr'}$ and $B_{b,s'}^{r'} = \Es{s}
\sum_{r} M_{b,ss'}^{rr'}$ are the measurements obtained by marginalizing the joint measurements
$M_{b,ss'}^{r''}$. This follows from a standard ``oracularization'' analysis,
similar to the one in the proof of Theorem~\ref{thm:com_test}.

From success in part (a) of the test it follows that the measurements
$A_{b,s}^r$ must be $\poly(\eps)$-close to the measurements $N_{s}^r$ used in the
test $\qld^{(2)}(m,d,q)$. By applying Theorem~\ref{thm:codeword_test} to the
strategy $N_{s}^r$, which is independent of $b$, this implies that, after
applying a suitable isometry, 
\[ A_{b,s}^r \approx_{\poly(\eps, \delta_C)} \qp_{W, s}^{r} \,=\,\sum_{a:\, (g_a)_{|s}=r}  \qp_W^a \;, \]
where $\qp_{W,s}^{r}$ is the measurement used in the honest strategy for
$\qld^{(2)}(m,d,q)$ as defined in
Definition~\ref{def:pauli-strategy}. Moreover, this implies that the
shared state is $\poly(\eps, \delta_C)$-close, under the isometry, to some encoded
state $\ket{\Psi}$.

Moreover, 
success in part (b) of the test implies that the measurements $\{B_{b,s'}^{r'}\}$ must
constitute a good strategy for the classical low-degree test $\cld^{(2)}(m',
d,q)$ (in which the players \emph{are} informed of $b$ and $j$). This implies that there exists a measurement $\{B_{b}^{h}\}$ with outcomes
$h$ that are $m'$-variate degree-$d$ polynomials over $\Fq$, such that
\[ B_{b, s'}^{r'} \approx \sum_{h : h|_{s'} = r'} B_{b}^{h}\;. \]

Finally, from part (c), we conclude that the set of outcomes $(g,h)$ that are such that the string $\Pi \in \Fq^{n'}$ for which $h$ is the low-degree extension, together
with $g$, are not accepted by the PCPP verifier, must have low probability of
being obtained when performing the corresponding measurement. Hence, by  the soundness of the PCPP from
Theorem~\ref{thm:composedpcp}, it follows that, for each prover $j$,
the polynomial $g$ obtained by this prover encodes an
assignment $a_j$ which satisfies $b_j \cdot a_j = c_j$ with high
probability. This implies that the expectation value of the output
$\omega^{\tr[\sum_j c_j]}$ computed by the verifier in part (d) is close to the expectation value of the Pauli observable
$ \otimes_{i=1}^{k} \qp_{W}(b_i)$,
as desired. 
\end{proof}

\subsection{Evaluating multiple-basis operators}
\label{sec:multi-basis}

Building on the test $\sumgame$, we design a test $\eval$ that can measure
operators which are tensor products of both $X$- and $Z$-basis
Paulis. In anticipation of our application to testing Hamiltonians, we
describe $\eval$ as taking as input a distribution over logical operators to
be measured. The process of translating these logical operators into
physical operators to be measured by each prover is bundled into the
test. 

\begin{figure}[htbp]
\rule[1ex]{16.5cm}{0.5pt}\\
Test $\eval_\xi(C,\pi,\ol{x},\ol{z})$: Given is a distribution $\pi$ over $\{S\subseteq\{1,\ldots,n\}\}\times \Fq^n\times\{\pm 1\}$, a $[k,k']$ weakly self-dual linear code $C$, $\ol{x},\ol{z}\in\Fq^k$ such that $\ol{X}=\qp_X(\ol{x})$ and $\ol{Z}=\qp_Z(\ol{z})$ are $\qp_X$ and $\qp_Z$ logical operators for $\mC$ respectively, and a parameter $0\leq \xi \leq 1$. \\
The verifier samples $(S,u,\epsilon)\sim \pi$.  
	The verifier performs one of the following two tests, the first with probability $(1-\xi)$, and the second with probability $\xi$: 
\begin{enumerate}
\item[(a)] (Test) Let $u_S$ and $u_{\ol{S}}$ be the substrings of $u$ indexed by $S$ and $\ol{S}$ respectively. 
\begin{enumerate}
\item[(i)] Do either of the following, with probability $1/2$ each:  
\begin{enumerate}
\item[1.] Send either $S$, $\ol{S}$, $\emptyset$ or $\{1,\ldots,n\}$ to all provers, with probability $1/4$ each. Execute the test $\code(C,n)$.
\item[2.] Do as in 1., except the sets sent to the special and composite provers are complemented ($S,\ol{S}$ or $\emptyset,\ol{\emptyset}$), and so is the choice of basis $W$ in $\code(C,n)$. 
\end{enumerate} 
\item[(ii)] Send $S$ to the special prover, and either $\emptyset$ or
  $\ol{\emptyset}$ to the composite prover. Choose a random
  vector $v \in \Fq^k$ in the column span of $G$, as in Definition~\ref{def:queries}. Execute the test
  $\sumgame(C,X,\{v_1 s, v_2s, \dots, v_k s\})$ on query string $s=u_{\ol{S}}$ (case $\emptyset$) or
  $s=u_{{S}}$ (case $\ol{\emptyset}$). Reject if the test $\sumgame$
  rejects, or if the returned value is not $1$. 
\end{enumerate}
\item[(b)] (Eval) Let $u_S$ and $u_{\ol{S}}$ be the substrings of $u$
  indexed by $S$ and $\ol{S}$ respectively. For $i\in \{1,\ldots,k\}$
  let $u_i$ be the string with substrings $u_{S,i}= \ol{x}_i u_S$
  indexed on $S$ and $u_{\ol{S},i} = \ol{z}_i u_{\ol{S}}$ indexed on $\ol{S}$. 
Send all provers the set $S$. Execute part (d) of $\sumgame(C,X,\{u_{1},
\dots u_{k}\})$ with all the provers. Let $E$ be the returned
value. Return $\epsilon \cdot E$.
\end{enumerate}
\rule[1ex]{16.5cm}{0.5pt}
\caption{Procedure $\eval(C,\pi,\ol{x},\ol{z})$ to evaluate the expectation of a set of Pauli operators, chosen according to distribution $\pi$, on an encoded state.}
\label{fig:eval}
\end{figure}

The procedure $\eval_\xi(C,\pi,\ol{x},\ol{z})$ is described in Figure~\ref{fig:eval}. It takes as input a  $[k,k']$ weakly
self-dual linear code $C$, a distribution\footnote{There should be no confusion between $\pi$ and the coordinate expansion map $\bij$ used in previous sections.} $\pi$ over $\{S \subseteq\{1,\ldots,n\}\}\times \Fq^n\times\{\pm 1\}$, and strings $ \ol{x},\ol{z}\in\Fq^k$ such that
the operators $\ol{X}=\qp_X(\ol{x})$ and $\ol{Z}=\qp_Z(\ol{z})$ are
respectively logical $\qp_X$ and $\qp_Z$ operators for the CSS code $\mC$
associated with $C$. To any triple $(S,u,\eps)$ in the support of $\pi$ we
associate a qudit Pauli operator
\begin{equation}\label{eq:def-hs}
h_{S}(u) \,=\, \otimes_{i\in S} \qp_X(u_i) \otimes_{i\in \ol{S}} \qp_Z(u_i)\;.
\end{equation}

The procedure is divided into a ``test'' and an ``eval''
part. The relative weight given to each part is governed by the parameter $0\leq
\xi \leq 1$. The goal of the testing part is to ensure that the provers' answers
in the evaluation part are distributed according to a distribution that can be
obtained by performing Pauli measurements on the
encoding of a fixed $n$-qudit state $\ket{\psi}$ using the CSS code $\mC$
associated with $C$, where the Pauli $\qp_X$ and $\qp_Z$ are encoded using
$\ol{X}$ and $\ol{Z}$ respectively.  The test is formulated using the notion of
``special'' and ``composite'' prover introduced in
Section~\ref{sec:code-protocol}; recall the scheme for distributing queries to
the composite prover specified in Definition~\ref{def:queries}.

\begin{lemma}\label{lem:eval}
Let $C$ be a $[k,k']$ weakly self-dual linear code and $\ol{X}=\qp_X(\ol{x})$,
$\ol{Z}=\qp_Z(\ol{z})$ logical $\qp_X$ and $\qp_Z$ operators for the
associated CSS code $\mC$ respectively. Let $n$ be an integer and $\pi$ a
distribution 
over $\{S\subseteq\{1,\ldots,n\}\} \times \Fq^n \times\{\pm 1\}$, and let $\xi$ be a real number between $0$ and $1$. Then the procedure $\eval_\xi(C, \pi , \ol{x},
\ol{z})$ has the following properties:
\begin{itemize}
\item \emph{(Completeness)} For any state $\ket{\psi} \in (\C^q)^{\otimes n}$
  there is a strategy for the provers that is accepted with probability $1$ in
  part (a) of the test, and such that the value returned by the
  verifier in part (b), conditioned on the choice of $(S,u,\epsilon)$, has expectation
	$ \epsilon\cdot \Re(\bra{\psi}h_S(u)\ket{\psi})$, where $h_S(u)$ is defined in~\eqref{eq:def-hs}.
\item \emph{(Soundness)} Suppose a strategy for the provers is accepted with probability at least $1-\eps$ in each of the
  ``test'' rounds of the procedure $\eval_\xi(C,\pi,\ol{x},\ol{z})$. Then there exists a state $\ket{\psi} \in (\C^q)^{\otimes n}$ such that, on expectation over $(S,u,\epsilon)\sim \pi$, the value returned by the verifier in step (b) of the protocol, conditioned on the choice of $(S,u,\epsilon)$, has expectation that is within $\poly(\eps,\delta_{C})$ of 
	$ \epsilon\cdot \Re(\bra{\psi}h_S(u)\ket{\psi})$, where $\delta_{C}$ is as
  specified in Theorem~\ref{thm:codeword_test}.
	\end{itemize}
\end{lemma}

\begin{remark}\label{rm:eval_bits}
The amount of bits communicated to any prover in the procedure $\eval$ is at most the number of bits necessary to communicate an element sampled from $\pi$
(which scales as the logarithm of the support size of $\pi$) plus the maximum of
the number of bits communicated in either the $\sumgame$ or $\code$
tests. It follows from Remark~\ref{rk:sum-complexity} that the former
requires $O(m \log q) = O(\frac{\log n}{\log \log n} \log q )$ bits, and from
Remark~\ref{rm:code_check_bits} that the latter similarly requires $O(\frac{\log n}{\log \log n} \log q)$
bits. 
\end{remark}

\begin{proof}  
We first show the completeness property. Let $\ket{\psi} \in (\C^q)^{\otimes n}$ and
$\ket{\Psi} \in (\C^q)^{\otimes nk}$ a qudit-by-qudit encoding of $\ket{\psi}$
according to $\mC$. The strategy for the provers uses $\ket{\Psi}$ as a shared
state. When a prover is sent a set $S$, it immediately applies an $F$ gate~\eqref{eq:fourier-f} to all qudits in $S$. If sent a query from the test $\code$, it applies the honest strategy for the test, as described in
Lemma~\ref{lem:code-completeness}. If asked to execute the protocol
$\sumgame(C,X,\{s_1, \dots, s_k\})$, on a query string $s_j \in\Fq^n$, it measures all its qubits in the $X$ basis to obtain a string $a\in \Fq^n$ and then executes the protocol honestly, following the strategy specified in Definition~\ref{def:sumgame-honest}. 

We verify that this strategy succeeds in each of the sub-tests of part (a) with probability $1$. For (i) this is a direct consequence of success in $\code$ and the fact that the code is self-dual; application of $F$ merely exchanges the role of the $X$ and $Z$ bases for all provers. For (ii) this follows from the completeness property in Theorem~\ref{thm:sum-game}. 

Regarding soundness, assume that a strategy for the provers succeeds with probability at least $1-\eps$ in each of the sub-tests executed in part (a). Using (i)1., by~\thmref{codeword_test}, for each $S$ the associated 
  strategy is isometric to a $\delta_S$-extension of a Pauli $\mC$-codeword strategy,
  where $\Es{S\sim\pi}\delta_S=\delta_{C}(\eps,q)$, as stated in the theorem.
	Note that in general the implied isometry depends on the choice of $S$. For the remainder of the proof, assume that a prover applies observable $\qp_{S,W}(w)$ when sent a query of the form $(W,w)$, after having been told the set $S$. Using the symmetry in the tests we may also assume that $\qp_{S,X}(w)=\qp_{\ol{S},Z}(w)$ for all $w$ and $S$. 
	
Next consider part (ii). Since the composite prover is not sent the
set $S$, the value $\comp{c}$ it claims also does not depend on
$S$. Since the only way for the test $\sumgame$ to return $1$ in part
(d) of the test is for the values $\comp{c}$ and $c$ to have identical
trace, from the previous analysis it follows that we may assume that
the value $c=a\cdot s$ returned by the special prover is obtained from
the outcome  obtained by a measurement of the composite prover in the
eigenbasis of $\qp_{\emptyset,X}$ (case $s=u_S$) or
$\qp_{\emptyset,Z}$ (case $s=u_{\ol{S}}$).

As a result, the distribution of claimed values obtained in part (b) of the test is close to what would be obtained if all provers were to perform a measurement in the eigenbasis of $\qp_{\emptyset,X}$ for the qudits in $S$, and $\qp_{\emptyset,Z}$ for the qudits in $\ol{S}$. By definition of the strings $u_{S,i}$ and $u_{\ol{S},i}$ that are actually sent to prover $i$, the resulting physical observable implements the logical
    $n$-qudit observable $h_S(u)$, as desired.   
\end{proof}

\subsection{Efficient energy test for local Hamiltonians}
\label{sec:constant_gap}

We show how to use the $\eval$ test to estimate the energy of a
Hamiltonian up to constant accuracy, provided that the terms of the Hamiltonian are
(not necessarily local) Pauli operators of a particular form, which we call
$Y$-free. From this, we deduce two results in the direction of the quantum
games PCP conjecture: we show
$\QMA$-hardness of approximating the maximum success probability of a
nonlocal game with logarithmic communication, either conditionally, assuming the Local Hamiltonian problem is QMA-complete for
a constant-error approximation (Corollary~\ref{cor:qma-generalizedXZ}), or
unconditionally, under randomized reductions (Corollary~\ref{cor:randomized}).

\begin{definition}\label{def:gen-h}
Let $n$ be an integer and $q$ a prime power. 
We say that a Hamiltonian $H$ on $(\C^q)^{\otimes n}$ is a Hamiltonian in $Y$-free form if $H$ can be expressed as
\begin{equation}\label{eq:gen-h}
 H = \Es{S \subseteq \{1, \dots, n\}, u \in \Fq^n} \,\frac{ \alpha_{S,u}
 }{2}\,\big(h_S(u) + h_S(u)^\dagger\big)\;,
\end{equation}
where each term $h_S(u)$ is a Pauli operator of the form described
in~\eqref{eq:def-hs}, the weights
$\alpha_{S,u} \in \R$ are such that $|\alpha_{S,u}|\leq 1$ for all
$(S,u)$, and the expectation is taken according to a distribution $\pi$ with polynomial-size support. 
\end{definition}

The term $Y$-free refers to the fact that there are no Pauli $Y$ operators
(i.e. products of $X$ and $Z$ acting on the same qudit) in any of the terms.
As motivation for considering this class of Hamiltonians, we remark that in the
case of $q = 2$, i.e. for qubits, our definition of $Y$-free Hamiltonians
includes the generalized XZ model of~\cite{CM13}.\footnote{Called the $XY$ model
  in their convention; to convert to ours it suffices to relabel the Pauli $Y$
  and $Z$ operators.} In
that reference, it was shown that the local Hamiltonian problem for the XZ model
is $\QMA$-complete, for an inverse-polynomial promise
gap. The class of $Y$-free Hamiltonians is considerably more general
as it imposes no limits on the locality of the terms in the Hamiltonian, and
accommodates qudits of dimension up to $\poly(\log n)$. 

The following lemma shows that it is possible to embed qubit
Hamiltonians of the XZ model into qudit
$Y$-free Hamiltonians with local dimension $q = 2^t$  for any $t$. This will be useful since the low-degree test requires fields of large enough size. 

\begin{lemma}
  Given any Hamiltonian $H$ in the XZ model over qubits, and $q = 2^t$, there exists
  a Hamiltonian $H'$ in $Y$-free form over qudits of dimension $q$ with the same
  spectrum (up to multiplicity) as $H$.
  \label{lem:embed_qubit_hamiltonian}
\end{lemma}
\begin{proof}
  Recall from Section~\ref{sec:qauli} that when $q = 2^t$ for any $t$, the field
  $\Fq$ admits a self-dual basis over $\F_2$, and a qudit of dimension $q$
  decomposes as a tensor product of $t$ qubits. Moreover, qubit Pauli operators
  $\{\sigma_{W, \ell}\}_{\ell \in \{1, \dots, t\}}$  acting on a single
  ``sub-qubit'' can be recovered from the qudit Paulis by the formula 
  \begin{equation}
    \sigma_{W, \ell} =  \qp_W(b_\ell),
    \label{eq:embed_qubit_pauli}
  \end{equation}
  where $\{b_1, \dots, b_t\}$ is a
  self-dual basis for $\Fq$ over $\F_2$. Extending this to multiple qudits, we can
  view a system of $n$ qudits of dimension $q=2^t$ each as a collection of $tn$
  qubits of dimension $2$ each. Let us index these qubits by pairs $(i, \ell)$, where $i
  \in \{1, \dots, n\}$ labels a qudit, and $\ell \in \{1, \dots, t\}$ labels
  a sub-qubit of the $i$th qudit. Then, given a qubit Hamiltonian $H$ over
  $n$ qubits, we construct the desired $H'$ by, for each qubit $X$ or $Z$ Pauli term in $H$ acting on qubits $i, j$, including the corresponding Pauli term acting
  on qubits $(i, 1)$ and $(j, 1)$. By~\eqref{eq:embed_qubit_pauli} this can be
  implemented by a generalized Pauli $\qp_X$ or $\qp_Z$ acting on qudits $i$ and $j$, and hence $H'$
  is in $Y$-free form. Moreover, $H'$ decomposes as a tensor product $H
  \ot \Id$ of $H$ acting on qubits $(1,1), (2,1), \dots, (n,1)$, and $\Id$
  acting on the remaining qubits. Hence $H'$ has the same spectrum (up to multiplicity) as $H$.
\end{proof}

Given a Hamiltonian $H$ in $Y$-free form provided as input, we describe a test whose maximum success probability is linearly related to the minimum energy of the Hamiltonian. The test requires the honest provers to share an encoding of a minimum-energy eigenstate of $H$ according to a quantum code $\mC$, and relies on the procedure $\eval$ described in Figure~\ref{fig:eval} to estimate the energy of the encoded state under $H$.  The energy test is described in Figure~\ref{fig:energy}, and its guarantees are stated in Theorem~\ref{thm:energy} below. 

\begin{figure}[H]
\rule[1ex]{16.5cm}{0.5pt}\\
Test~$\energy_\xi(C,H)$: Given as input is a Hamiltonian in $Y$-free form, specified by real coefficients $\{\alpha_{S,u}\}$ as in~\eqref{eq:gen-h}, a $[k,k']$ weakly self-dual linear code $C$ such that the associated CSS code $\mC$ encodes at least one logical qubit, and a parameter $0\leq \xi \leq 1$. 
\begin{enumerate}
\item Let $\ol{x},\ol{z}\in\Fq^n$ be such that $\qp_X(a\ol{x})$ and $\qp_Z(b\ol{z})$ are logical $\qp_X(a)$ and $\qp_Z(b)$ operators for the code $\mC$, respectively. 
\item Let $\pi$ be the distribution over $\{S \subseteq \{1, \dots, n\}\} \times
  \Fq^n \times \{\pm 1\}$ that is obtained by sampling $(S, u)$ uniformly, and
  returning $(S,u, \text{sign}(\alpha_{S,u}))$ with probability $|\alpha_{S,u}|$, and a symbol ``$\perp$'' with probability $1-|\alpha_{S,u}|$.
\item Execute $\eval_\xi(C, \pi, \ol{x},\ol{z})$. If the element sampled from $\pi$ is $\perp$, then automatically accept in case part (a) of the test is executed, and reject in case part (b) is executed. Otherwise, if the test returns ACCEPT or REJECT, then
  accept or reject accordingly. Finally, if the test returns a value $e$ such that $\Re(e)\in [-1,1]$, then accept with probability $\frac{1}{2}(1-\Re(e))$. 
\end{enumerate}
\rule[1ex]{16.5cm}{0.5pt}
\caption{Test $\energy_\xi(C,H)$ for the ground state of a Hamiltonian $H$ in $Y$-freeform.}
\label{fig:energy}
\end{figure}

\begin{theorem}\label{thm:energy}
  Let $H$ be a Hamiltonian in $Y$-free form, $C$ a weakly self-dual linear code, and $0\leq \eta\leq 2$. 
	\begin{itemize}
	\item \emph{(Completeness)}
If $\lambda_{\min}(H) \leq -1+\eta$, there is a strategy for the provers such that for any $0\leq \xi\leq 1$ the test $\energy_\xi(C,H)$ accepts with
  probability at least $1-\frac{1}{2}\eta\xi$. 
	\item \emph{(Soundness)}
	If there exists a strategy with probability of success in the test
  $\energy_\xi(C,H)$ at least $ 1- \eps$, then $\lambda_{\min}(H) \leq -1+\eta'$
  for $\eta' = \frac{2\eps}{\xi} + \poly(\frac{\eps}{1 - \xi},
  \delta_C(\frac{\eps}{1 - \xi}))$, where $\delta_C$ is as specified in Theorem~\ref{thm:codeword_test}.
	\end{itemize}
  \label{thm:energy_test}
\end{theorem}

We show the completeness and soundness properties claimed in Theorem~\ref{thm:energy} in two separate lemma. 

\begin{lemma}[Completeness]\label{lem:energy-completeness}
Let $H$ be a Hamiltonian in $Y$-free form such that $H$ has an eigenstate $\ket{\psi}$ with associated eigenvalue $\lambda \in [-1,1]$, and $0\leq\xi\leq 1$. Then for any weakly self-dual linear code $C$ there is a strategy for the provers, based on sharing an encoding of the state $\ket{\psi}$ according to $\mC$, whose success probability 
  in the test $\energy_\xi(C,H)$ is $ (1-\xi) + \frac{\xi}{2}(1-\lambda)$.
\end{lemma}

\begin{proof}
Let $\ket{\psi} \in (\C^q)^{\otimes n}$ be as in the lemma, and $\ket{\Psi}\in (\C^q)^{\otimes kn}$ a qudit-by-qudit encoding of $\ket{\psi}$ under the code $C$, where each individual qudit is encoded according to the logical operators $\ol{x}$, $\ol{z}$ used by the verifier in the test $\energy$. When sent a query by the verifier, each prover applies the honest strategy in the test $\code$, as specified in Lemma~\ref{lem:code-completeness}. By definition this strategy succeeds with probability $1$  in part (a) of $\eval$. 

Regarding part (b), it follows from the definition of the distribution $\pi$ and the completeness property of the procedure $\eval$ stated in Lemma~\ref{lem:eval} that the probability of accepting in the third step of the procedure $\energy$, conditioned on part (b) of $\eval$ being executed by the verifier, is precisely $\frac{1}{2}(1-\bra{\psi} H \ket{\psi})$. 
\end{proof}

\begin{lemma}[Soundness]\label{lem:energy-soundness}
Let $H$ be a Hamiltonian in $Y$-free form, $0\leq\xi\leq 1$ and $C$ a weakly self-dual linear code. Suppose there exists a strategy for the provers whose success probability 
  in the test $\energy_\xi(C,H)$ is $1-\eps$, for some $\eps\leq\xi$. Then $H$
  has an eigenvector with energy at most $-1 + \frac{2\eps}{\xi} +
  \poly(\frac{\eps}{1 - \xi},\delta_{C})$. 
\end{lemma}

\begin{proof}
By definition of the test $\energy_\xi$, the provers' strategy must succeed with
probability at least $1-\frac{\eps}{1-\xi} = 1-\eps'$ in part (a) of
$\eval_\xi$. Applying Lemma~\ref{lem:eval}, it follows that the value returned
by the verifier in part (b) of the test 
is a random variable whose expectation is within $\poly(\eps',\delta_{C})$ of 
$\frac{1}{2}(1-\bra{\psi}H\ket{\psi})$, for some state $\ket{\psi}$. Thus
letting $\lambda =  \bra{\psi}H\ket{\psi}$ and $p_a$, $p_b$ the provers' success probability in parts (a) and (b) of the test respectively we have  
\begin{align*} 
  p_{\text{success}} &= \frac{(1 - \xi)}{2} p_a  + \xi p_b \;,
	\end{align*}
	thus $ p_{\text{success}} \geq 1 - \eps $ implies that 
  $\xi p_b \geq \xi - \eps$. Using that $p_b \leq \frac{1}{2}(1-\lambda + \poly(\eps',\delta_C))$ yields 
$$  \lambda  \,\leq\, -1 + \frac{2\eps}{\xi} + \poly(\eps', \delta_{C})\;,$$
giving the conclusion of the lemma. 
\end{proof}

\begin{corollary}\label{cor:qma-generalizedXZ}
  Assume the Local Hamiltonian problem for qubit Hamiltonians in the $XZ$ model
  with
  promise gap $b - a = \Omega(1)$ is $\QMA$-complete. Then, there is a
  one-round, $7$-prover $\MIP^*$ protocol for the class $\QMA$ with 
  $O(\log(n))$-bits of communication and constant completeness-soundness gap.
\end{corollary}

\begin{proof}
  First, note that by the hardness assumption made in the corollary and Lemma~\ref{lem:embed_qubit_hamiltonian}, it
  follows that estimating the ground energy of qudit Hamiltonians with local dimension $2^t$ in
  $Y$-free form with $\Omega(1)$ promise gap is $\QMA$-hard for any choice of
  $t$. Thus, to establish the conclusion, it suffices to show that there
  exists an $\MIP^*$ protocol with the desired parameters for this variant of
  the Local Hamiltonian Problem. 

  Let a Hamiltonian in $Y$-free form be given, scaled such that  the energy threshold in the YES case is $-1 + \eta$, and in the NO case is
  $ -1 + \eta'$ with $\eta' - \eta = \Omega(1)$. Furthermore, let $q = 2^t$
  where $t = \Theta(\log\log(n))$, and let $C$ be the quadratic residue code
  from Example~\ref{ex:quad_res_code}, which has $k=7$. Using
  Theorem~\ref{thm:energy}, the test $\energy_\xi(C,H)$ succeeds with
  probability $p_{\text{YES}} \geq 1 - \frac{\eta \xi}{2}$ in the YES case, and
  in the NO case with probability $p_{\text{NO}} \leq 1- \eps$ for $\eps$
  such that 
  \begin{equation}
    \frac{2\eps}{\xi}  + \poly\Big(\frac{\eps}{1 - \xi}, \delta_C\Big) = \eta'\;, \label{eq:energy_soundness}
  \end{equation}
  where $\delta_C = \max(\poly(\frac{\eps}{1 - \xi}), \poly(q^{-1}))$.
  Denote the difference between these two probabilities by $\Delta$; it is given by
  \begin{align}
    \Delta&=  p_{\text{YES}} - p_{\text{NO}} \notag\\
		&\geq \eps - \frac{\eta \xi}{2} \notag
    \\
    &= \frac{\xi}{2}(\eta' - \eta) - \frac{\xi}{2} \poly(\frac{\eps}{1 - \xi},
      \delta_C)\;. \label{eq:energy_gap}
  \end{align}
  For $\Delta$ to be positive it suffices to ensure that the quantity
  $\poly(\frac{\eps}{1 - \xi}, \delta_C)$ is less than, say,
  $\frac{1}{2}(\eta' - \eta)$, which is a constant. Given the
  definition of $\delta_C$ and our choice of $q \sim n^{\log n}$, this can be ensured
  for some constant $\eps, \xi$ depending only on $\frac{1}{2}(\eta' -
  \eta)$. This results in a constant $\Delta$. Therefore, the energy
  test $\energy_\xi(C,H)$ constitutes a $\MIP^*$
  protocol for $\QMA$ with a constant completeness-soundness
  gap. Regarding the communication cost, by plugging $q = O(\log \log
  n)$ into the bounds for the test $\eval$ given in
  Remark~\ref{rm:eval_bits}, we find that the test $\energy_\xi(C,H)$
  requires $O(\log n)$ bits communication.
\end{proof}

Without making any assumptions, we can show a $\QMA$-hardness result under randomized reductions.  

\begin{corollary}
\label{cor:randomized}
It is $\QMA$-hard under poly-time randomized Karp reductions to
determine whether
the maximum acceptance probability of a one-round $\MIP^*$ protocol with
logarithmic communication is at least $1$ or at most $ \frac{1}{2}$.
\end{corollary}

The idea for the proof of Corollary~\ref{cor:randomized} is to start with a $\QMA$-hard instance of
the Local Hamiltonian problem, with inverse-polynomial promise gap, and amplify
this gap by taking a tensor power of the Hamiltonian. Expanding the tensor powers results in a Hamiltonian that is an average of exponentially many terms. We then apply the
Ahlswede-Winter matrix Chernoff bound to randomly sub-sample a set of terms from
the amplified Hamiltonian, yielding a Hamiltonian with polynomially many terms whose ground state energy can be
tested using the $\energy$ test.

\begin{lemma}[Gap amplification, Lemma 26 of~\cite{NV17}]
Let $H$ be an $n$-qudit Hamiltonian with
minimum energy $\lambda_{\min}(H)\geq 0$ and such that $\|H\|\leq 1$. 
  Let $p(n), q(n)$ be polynomials such that $p(n) >
  q(n)$ for all $n$. Let
	$$H' =  \Id^{\otimes a} - (\Id - ( H - a^{-1} \Id ))^{\otimes a},\qquad \text{where}\qquad a = \Big(\frac{1}{q}-\frac{1}{p}\Big)^{-1}.$$
	Then $H'$ is a (non-local) Hamiltonian over $an = O(np(n))$
        qudits with norm $\|H'\|=O(1)$ and with each term having norm $O(1)$, such that if $\lambda_{\min}(H) \leq 1/p$, then $\lambda_{min}(H') \leq 1/2$, whereas if $\lambda_{\min}(H) \geq 1/q$, then $\lambda_{min}(H') \geq 1$.
  \label{lem:amplify}
\end{lemma}

\begin{proof}[Proof of Corollary~\ref{cor:randomized}]
  We start by recalling that the Local Hamiltonian problem is $\QMA$-complete
  for qubit Hamiltonians in the $XZ$ model, up to inverse-polynomial
  promise gap~\cite{CM13}. Let 
  \[ H = \Es{j \in \{1, \dots, \ell\}} H_j \]
  be a given Hamiltonian on $n$ qubits from the $XZ$ model (also allowing terms
  that are multiples of the identity), with $\ell =
  \poly(n)$ local terms $H_j$, normalized
  such that $0 \leq H \leq \Id$. As can be seen from Definition~\ref{def:gen-h}, this Hamiltonian can be equivalently
  viewed as a Hamiltonian in $Y$-free form acting on qubits. We aim to give a protocol that distinguishes
  between the cases $\lambda_{\min}(H) \leq 1/p$ (YES) or $\lambda_{\min}(H) \geq
  1/q$ (NO),
  where $0 \leq 1/p < 1/q \leq 1$ and $p$ and $q$ are polynomial functions of
  $n$. By applying Lemma~\ref{lem:amplify} to $H$ and scaling down the resulting
  Hamiltonian, we obtain a new Hamiltonian
  \[H'  = c\left(\Id^{\otimes a} - (\Id - ( H - a^{-1}\Id ))^{\otimes
        a}\right)  \]
  acting on $a n = \poly(n)$ qudits with norm $\|H'\| = 1$ and all of whose terms have norms bounded by $1$, such that $\lambda_{\min}(H') \leq c/2$ in the YES case and
  $\lambda_{\min}(H') \geq c$ in the NO case, for some constant $0 < c < 1$. For
  our purposes, it will be useful to express $H'$ as an average
  \[ H' = \Es{J \in \{1, \dots \ell'\}} H'_{J}\;, \]
where each term $H'_J$ is of the form $c \Id^{\ot a} - \alpha_J h_{S_J}(u_J)$ for
some Pauli operator $h_{S_J}(u_J)$ and weight $\alpha_J \in [-1, 1]$. The number
of terms in this decomposition is $\ell' = (\ell+1)^{a}$, which is exponential in $n$. This means that executing the
  test $\energy(C,H')$ would require $\poly(n)$ bits of communication with the
  verifier, just to specify a single  term in the Hamiltonian. To avoid this problem, we
   use randomness to sample a subset of the terms. First,
  rescale $H'$ so that all of the terms are positive and have norm at most $1$:
  \[ H'' = \Es{J \in \{1, \dots, \ell'\}} H''_{J}\;, \qquad H''_J = \frac{1}{2}
    \big(H'_J + \Id - c\big)\;. \]
  This rescaled Hamiltonian satisfies $\lambda_{\min}(H'') = \frac{1}{2}(1
  -c + \lambda_{\min}(H')) \geq \frac{1}{2}(1 - c) $.
  Now, let $H'''$ be a Hamiltonian obtained by
  uniformly sampling $m$ terms at random from $H''$, where $m$ is a parameter
  to be chosen. By the matrix Chernoff Bound~\cite [Theorem 19]{AW02}, for any
  $\eps \in [0, 1/2]$,   
  \[ \Pr[\lambda_{\min}(H''') \notin [(1-\eps) \lambda_{\min}(H''),
  (1+\eps) \lambda_{\min}(H'')]] \leq 2 \cdot \exp\left(an \ln2 -m \frac{\eps^2
      \lambda_{\min}(H'')}{2\ln 2}\right). \]
In particular, taking $\eps \leq c/(4 - 2c)$ and $m = \poly(n)$, we obtain that, with probability exponentially
close to $1$, in the YES case $\lambda_{\min}(H''') \leq c_1$ and in the NO case
$\lambda_{\min}(H''') \geq c_2$ where $c_2 - c_1 \geq c/8$. 
Moreover, $H'''$ is a $Y$-free Hamiltonian with polynomially many terms. Hence,
by the same arguments as in the proof of Corollary~\ref{cor:qma-generalizedXZ}, there exists
a $7$-prover $\MIP^*$ protocol with $O(\log(n))$-bit messages and constant
completeness-soundness gap to estimate the ground energy of $H'''$ up to
precision $c/16$, and hence to solve the Local Hamiltonian problem for $H$.
\end{proof}

\subsection{Energy test for frustration-free Hamiltonians with small gap}
\label{sec:ff}

In this section we show how the procedure $\sumgame$ can be used in a different scenario than the one considered in the previous section: the case of an $n$-qubit Hamiltonian that is either frustration-free, or has ground state energy that is at most an inverse polynomial in $n$. The tests described in this section are more restrictive than those considered in the previous section, but they have the advantage of not relying on a randomized reduction. They apply to the following form of ``linear XZ Hamiltonian''.

\begin{definition}\label{def:linear-xz}
A $n$-qudit Hamiltonian $H$, where each qudit has dimension a prime power $q$, is in \emph{linear XZ form} if it can
be written as
\[ H = \Es{W\in\{X,Z\}, j \in \{1, \dots, \ell\}} \Pi_{W,j}\;, \]
where the expectation is taken under the uniform distribution, and for each $W \in \{X, Z\}$ and $j\in\{1,\ldots,\ell\}$ the term $\Pi_{W,j}$
is a projector that is diagonal in the basis $W$ (for each qubit), and such that
the nullspace of $\Pi_{W,j}$ can be described by a collection of $t_{W,j}$
linear equations $\{s_{W,j,i} \cdot a = b_{W,j,i} ,\, i\in \{1,\ldots,t_{W,j}\}\}$
over $\Fq$, where here $a=(a_1,\ldots,a_n)\in\Fq^n$ specifies a basis state $\ket{a}_W$ in basis $W$ for the $n$ qudits, and $s_{W,j,1},\ldots,s_{W,j,t_j} \in \Fq^n$ and $b_{W,j,i}\in\Fq$ are coefficients of linear equations. 
\end{definition}

Note that a special case of the definition is one in which some of the
$\Pi_{W,j}$ have rank $1$, since any fixed element $a\in\Fq^n$ can be uniquely
specified by a system of $n$ linear equations.  

A Hamiltonian $H$ in linear XZ form is specified by the collection of equations $\{(s_{W,j,i},b_{W,j,i}),\, W\in\{X,Z\}, j\in\{1,\ldots,\ell\},i\in\{1,\ldots,t_j\}$. We will be interested in the problem of deciding whether $H$ has ground state energy $0$, or at least some inverse polynomial in $n$, $\gamma(n)$. Replacing each $\Pi_{W,j}$ in $H$ by an average of $t_{W,j}$ terms, each associated with a single equation $(s_{W,j,i},b_{W,j,i})$, preserves the distinction between these two cases, up to a polynomial multiplicative scaling in $\gamma$. Therefore, for the remainder of this section we assume that $t_{W,j}=1$ for all $W,j$, and write $(s_{W,j},b_{W,j}) $ for $(s_{W,j,1},b_{W,j,1})$. 

The main result of this section is an interactive protocol for deciding between
the cases where a Hamiltonian $H$ in linear XZ form is frustration free, or has
energy at least some inverse polynomial in $n$. The main ingredients for the
protocol are the low-degree test from Theorem~\ref{thm:qld} and the test $\sumgame$. As for the case of the $Y$-free Hamiltonians considered in
Section~\ref{sec:constant_gap}, it would be straightforward to extend the
results of this section to Hamiltonians as in Definition~\ref{def:linear-xz}, but
allowing a polynomial number of possible basis choices for the $n$ qudits, chosen among $\{X,Z\}^n$, instead of only $X^n$ and $Z^n$. For simplicity, we focus on the case of two bases only.


\begin{figure}[H]
\rule[1ex]{16.5cm}{0.5pt}\\
Test~$\XZ_N(H)$: Given as input is a $n$-qudit Hamiltonian in linear $XZ$ form, where each qudit is of dimension a prime power $q$, and an integer $N$. Let $C$
be a $[k,k']$ weakly self-dual linear code over $\Fq$, known to all parties, such that $C$ encodes at least one qudit.  
The verifier performs one of the following, with probability $1/2$ each: 
\begin{enumerate}
\item Select a basis $W\in\{X,Z\}$ uniformly at random. Select an
  equation $(s,b)$ as described in the proof of
  Theorem~\ref{thm:xz-form} (this depends on the parameter $N$).
 Choose $\ol{w} \in \Fq^k$ to be a random vector
  such that $\qp_W(\ol{w})$ is a logical operator for $\mC$, and execute the test
  $\sumgame(C, W, \{ \ol{w}_1 s,  \ol{w}_2 s, \dots, \ol{w}_k s\})$ with the provers.
  Reject if the protocol rejects, or if the linear combination $\sum_{i=1}^{k}
  c_k$ of the claimed values is not equal to $E$. Else,
  accept. 
\item Execute test~$\code(C,n')$ with the provers, where $n'= Nn$ is as in the proof of Theorem~\ref{thm:xz-form}.  
\end{enumerate}
\rule[1ex]{16.5cm}{0.5pt}
\caption{Test $\XZ_N(H)$ for the ground state energy of a Hamiltonian $H$ in linear XZ form.}
\label{fig:xz-test}
\end{figure}

 We state the main result of this section. 

\begin{theorem}\label{thm:xz-form}
Let $n$ be an integer, $q = p^t$ a prime power such that $q=\Theta(\poly\log n)$, and $\gamma(n) = \Omega(\poly^{-1}(n))$. There exists a universal constant $\eps_0 >0$ and $N = O(\poly(n))$ such that the following holds. For any $n$-qudit Hamiltonian $H$ in linear XZ form, 
\begin{itemize}
\item If $H$ has ground state energy $0$, then there is a strategy for the provers that is accepted in the test $\XZ_N(H)$ with probability $1$.
\item If $H$ has ground state energy at least $\gamma(n)$, then no strategy for the provers is accepted with probability more than $1-\eps_0$ in the test $\XZ_N(H)$. 
\end{itemize}
By basing the test on the code from Example~\ref{ex:quad_res_code}, the test can be executed with $7$ provers and a total amount of communication between the verifier and the provers that is $O(\log n)$.
\end{theorem}

\begin{proof}
We start by amplifying the promise gap by taking a tensor product of $N$ copies of $H$, for $N = \lceil\delta\gamma^{-1}(n)\rceil$, for some $0<\delta\leq 1$ to be determined. For any $W\in\{X,Z\}$ and $j\in\{1,\ldots,\ell\}$ let $\tilde{\Pi}_{W,j} = \Id - \Pi_{W,j}$. Define 
\begin{align}
H' &= \Id - (\Id - H)^{\ot N} \notag\\
&= \Id - (\Es{(W,j)}(\Id - \Pi_{W,j}))^{\ot N} \notag\\
 &= \Id - \Es{W_1, j_1, \dots, W_N, j_N} \tilde{\Pi}_{W_1, j_1} \ot \dots \ot
  \tilde{\Pi}_{W_N, j_N}\;,\label{eq:hprime}
	\end{align}
	where all expectation are uniform over the appropriate sets. 
Then $H'\geq 0$. If $H$ is frustration-free then $H'$ is frustration-free as well. If $H$ has ground energy at least $\gamma$, then $H'$ has ground energy at least $1-e^{-\delta} \geq \delta/2$. Note that $H'$ is again a Hamiltonian in linear XZ form such that $H'$ acts on $n' = N n$ qudits.  

For any $\vec{W}=(W_1,\ldots,W_N)$ and $\vec{j}=(j_1,\ldots,j_N)$ let $\tilde{\Pi}^X_{\vec{W},\vec{j}} = \otimes_{i=1}^N \tilde{\Pi}^X_{W_i,j_i}$, where $\tilde{\Pi}^X_{W,j} = \tilde{\Pi}_{W,j}$ if $W=X$ and $\tilde{\Pi}^X_{W,j} = \Id$ otherwise. Define $\tilde{\Pi}^Z_{\vec{W},\vec{j}}$ similarly. From~\eqref{eq:hprime}, we get 
\begin{equation}\label{eq:hprime-2}
H'\,=\, \Id - \Es{\vec{W},\vec{j}} \tilde{\Pi}^X_{\vec{W},\vec{j}}\tilde{\Pi}^Z_{\vec{W},\vec{j}}\;.
\end{equation}
We use the following claim:

\begin{claim}\label{claim:commute-bound}
  Let $ A, B$ be two positive semidefinite operators such that $A, B \leq \Id$ and $AB =
  BA$. Then 
	$$ \Id-AB \,\leq\, \big(\Id-A\big) + \big(\Id-B\big)\;.$$
  \label{claim:commuting-pigeonhole}
\end{claim}
\begin{proof}
  Note that since $B$ commutes with $A$, it must also commute with the positive 
  square root of $A$. Hence $AB = A^{1/2} B A^{1/2} \geq 0$. Likewise, $(\Id
  - A)$ commutes with $(\Id - B)$, so $(\Id - A)(\Id - B) \geq 0$. 
\end{proof}

Starting from~\eqref{eq:hprime-2} and applying Claim~\ref{claim:commute-bound},
\begin{align}
H' &\leq 2\Es{W\in\{X,Z\}}\Big( \Id - \Es{\vec{W},\vec{j}}   \tilde{\Pi}^W_{\vec{W},\vec{j}}\Big)\;.\label{eq:hprime-3}
\end{align}
For $W\in \{X,Z\}$ let $H'_W = \Id - \Es{\vec{W},\vec{j}}   \tilde{\Pi}^W_{\vec{W},\vec{j}}$. If $H'$ has ground energy zero, then both $H'_Z$ and $H'_X$ have ground energy zero as well. If $H'$ has ground energy at least $\delta/2$, then for any vector $\ket{\psi}$, either $\bra{\psi}H'_X\ket{\psi} \geq \delta/4$ or $\bra{\psi}H'_Z\ket{\psi} \geq \delta/4$. 

The goal of the test  $\XZ_N(H)$ described in Figure~\ref{fig:xz-test}
is to distinguish between these two cases. To complete the description
of the test we specify how the linear equation $(s,b)$ considered in
item 1. of the test is obtained. First form the set
$S=\{(s_\ell,b_\ell),\,1\leq\ell\leq t\}$ that is the union of all
equations which specify the $+1$ eigenspace of each individual
$\tilde{\Pi}_{W_i,j_i}$ such that $W_i=W$. Using the notation from
Definition~\ref{def:linear-xz} we have $|S|=t = \sum_{i=1}^N
1_{W_i=W}t_{W,j_i}$, which is polynomial in $n$. Then $(s,b)$ is
obtained by sampling $(\delta/8)$-biased random variables
$(y_1,\ldots,y_t)\in\Fq^t$ and setting $s = \sum_\ell y_\ell s_\ell$
and $b=\sum_\ell y_\ell b_\ell$. We refer to
e.g.~\cite{azar1998approximating} for a construction of such random
variables using $\poly\log( t, q,\delta^{-1})$ random bits; note that
this use of $\delta$-biased random variables is analogous to their use
in the classical exponential PCP for QUADEQ. Thus the amount of
communication required to specify $s$ to a prover is
$\poly\log(t,q,\delta^{-1})=\poly\log \log(n)$.

We show that the test $\XZ_N(H)$ satisfies the requirements of the theorem. We first argue completeness, and then soundness.

\begin{claim}[Completeness]\label{claim:xz-completeness}
Suppose that $H$ has ground energy $0$. Then there is a strategy for the provers that is accepted with probability $1$ in the test $\XZ_N(H)$. 
\end{claim}

\begin{proof}
Let $\ket{\psi}$ be a ground state of $H$. Let $C$ be the linear code used by the verifier. Consider the strategy for test $\code(C,n')$ described in Lemma~\ref{lem:code-completeness}, where the encoded state is the $n'$-qudit state $\ket{\psi}^{\otimes N}$. From the lemma it follows that the strategy succeeds with probability $1$ in item 2. of the test $\XZ_N(H)$.

It remains to describe the behavior of the  provers when elected to perform the
test $\sumgame$ in item 1. of $\XZ_N(H)$. The prover $j$ first measures its share of
each qudit in the basis $W$, obtaining outcomes $a_j\in\Fq^{n'}$. The prover then executes the honest behavior in
test $\sumgame(C, W, \{\ol{w}_1 s, \dots, \ol{w}_k s\} )$, with the
claimed value being $c_j = \ol{w}_j s\cdot
a_j$. Moreover, since $\ket{\psi}$ is a ground state of $H$, the
condition $\sum_{j} c_j = s \cdot \sum_{j} \ol{w}_j a_j = b$ is satisfied with
certainty. Hence, the honest strategy is accepted with probability $1$.
\end{proof}

The next claim shows soundness of the protocol. 

\begin{claim}[Soundness]\label{claim:xz-soundness}
Suppose that $H$ has ground energy at least $\gamma$. Then any strategy for the provers in the test $\XZ_N(H)$ is accepted with probability at most $1-\delta_s$, for some $\delta_s = \max(\poly(\delta),\poly(q^{-1}))$. 
\end{claim}

\begin{proof}
Fix a strategy for the provers that has success probability at least $1-\eps$ in
the test, for some $\eps>0$. This consists of a state $\ket{\Psi}$ and
measurement operators $\{M_s^q\}$ for the special prover. The strategy must
succeed with probability at least $1-2\eps$ in item 2. of test
$\XZ_N(H)$. Applying Theorem~\ref{thm:codeword_test}, it follows that there
exists a state $\ket{\psi}\in(\C^p)^{\otimes n'}$ such that, up to local
isometries, $\ket{\Psi}$ is within distance $\delta_{C}$ of a valid $n'k$-qudit
encoding of some state $\ket{\psi}$ according to the code $C$; under the same
isometry, each prover's measurement upon query $(W,w)$ is $\delta_{C}$-close to
an application of the observable $\qp_W(w_\bij)$ on the prover's share of the encoding. Up to an increase of $\delta_{C}$ in the error we  assume for the remainder of the proof that all provers apply the honest strategy when given queries distributed as in the test $\code$. 

We now analyze item 1. of test $\XZ_N(H)$. By assumption, the honest strategy, based on state $\ket{\psi}$, succeeds with probability at least $1-O(\delta_{C})$ in this part of the test (we can assume $\delta_{C} \geq \eps$ without loss of generality). 
Since $H$ has ground energy at least $\gamma$,
by~\eqref{eq:hprime-3} either $\bra{\psi} H'_X\ket{\psi} \geq \delta/4$, or
$\bra{\psi}H'_Z\ket{\psi}\geq \delta/4$. Assume the former. This means that
whenever each of the $n'$ qudits of $\ket{\psi}$ is measured in the $X$ basis,
the probability that the outcome string $a$ is such that $a$ satisfies all linear equations
$s_\ell\cdot a = b_\ell$, $\ell\in\{1,\ldots,t\}$, considered by the verifier,
when the basis is chosen to be $W=X$, is at most $1-\delta/4$. By definition of
the equation $(s,b)$, the probability (which now includes the verifier's coin
tosses in selecting $(s,b)$) that $s\cdot a =b$ is at most $1-\delta/8$. 
It follows from the soundness part of Theorem~\ref{thm:sum-game} that the provers must be rejected with probability  $\poly(\delta) - O(\delta_{C})$. This gives a contradiction for any $\eps$ small enough such that $O(\delta_{C}) \ll \poly(\delta)$. 
\end{proof}

To conclude the proof of the theorem, we choose the constant $\delta$ to be sufficiently small so that $\delta_s$ in Claim~\ref{claim:xz-soundness} is a positive constant. 
\end{proof}

\bibliography{quantum_pcp}

\newcommand{\etalchar}[1]{$^{#1}$}
\begin{thebibliography}{BSGH{\etalchar{+}}05}

\bibitem[AALV09]{AharonovILV09detectability}
Dorit Aharonov, Itai Arad, Zeph Landau, and Umesh Vazirani.
\newblock The detectability lemma and quantum gap amplification.
\newblock In {\em Proc. 41st STOC}, pages 417--426, New York, NY, USA, 2009.
  ACM.

\bibitem[AAV13]{AharonovAV13qpcp}
Dorit Aharonov, Itai Arad, and Thomas Vidick.
\newblock The quantum {PCP} conjecture.
\newblock Technical report, 2013,
  \href{http://arxiv.org/abs/1309.7495}{{\ttfamily arXiv:1309.7495}}.
\newblock Appeared as guest column in ACM SIGACT News archive Volume 44 Issue
  2, June 2013, Pages 47--79.

\bibitem[AE15]{AharonovE15commuting}
Dorit Aharonov and Lior Eldar.
\newblock The commuting local {H}amiltonian problem on locally expanding graphs
  is approximable in {NP}.
\newblock {\em Quantum Information Processing}, 14(1):83--101, January 2015.

\bibitem[AFB17]{arnon2017device}
Rotem Arnon-Friedman and Jean-Daniel Bancal.
\newblock Device-independent certification of one-shot distillable
  entanglement.
\newblock 2017,  \href{http://arxiv.org/abs/1712.09369}{{\ttfamily
  arXiv:1712.09369}}.

\bibitem[AFY17]{arnon2017noise}
Rotem Arnon-Friedman and Henry Yuen.
\newblock Noise-tolerant testing of high entanglement of formation.
\newblock 2017,  \href{http://arxiv.org/abs/1712.09368}{{\ttfamily
  arXiv:1712.09368}}.

\bibitem[ALM{\etalchar{+}}98]{AroLunMotSudSze98JACM}
Sanjeev Arora, Carsten Lund, Rajeev Motwani, Madhu Sudan, and Mario Szegedy.
\newblock Proof verification and the hardness of approximation problems.
\newblock {\em J. ACM}, 45(3):501--555, 1998.

\bibitem[AMN98]{azar1998approximating}
Yossi Azar, Rajeev Motwani, and Joseph~Seffi Naor.
\newblock Approximating probability distributions using small sample spaces.
\newblock {\em Combinatorica}, 18(2):151--171, 1998.

\bibitem[Ara02]{Arvind:02}
P.~K. Aravind.
\newblock The magic squares and {B}ell's theorem.
\newblock Technical report, 2002,
  \href{http://arxiv.org/abs/quant-ph/0206070}{{\ttfamily
  arXiv:quant-ph/0206070}}.

\bibitem[AS98]{AroSaf98JACM}
Sanjeev Arora and Shmuel Safra.
\newblock Probabilistic checking of proofs: A new characterization of {NP}.
\newblock {\em J. ACM}, 45(1):70--122, 1998.

\bibitem[AW02]{AW02}
Rudolf Ahlswede and Andreas Winter.
\newblock Strong converse for identification via quantum channels.
\newblock {\em IEEE Transactions on Information Theory}, 48(3):569--579, 2002,
  \href{http://arxiv.org/abs/quant-ph/0012127}{{\ttfamily
  arXiv:quant-ph/0012127}}.

\bibitem[BFL91]{BabForLun91CC}
L{\'{a}}szl{\'{o}} Babai, Lance Fortnow, and Carsten Lund.
\newblock Non-deterministic exponential time has two-prover interactive
  protocols.
\newblock {\em Computational Complexity}, 1:3--40, 1991.

\bibitem[BH13]{BrandaoH13product}
Fernando~G.S.L. Brandao and Aram~W. Harrow.
\newblock Product-state approximations to quantum ground states.
\newblock In {\em Proc. 45th STOC}, 2013.

\bibitem[BLR93]{BLR93}
Manuel Blum, Michael Luby, and Ronitt Rubinfeld.
\newblock Self-testing/correcting with applications to numerical problems.
\newblock {\em Journal of Computer and System Sciences}, 47:549--595, 1993.

\bibitem[BSGH{\etalchar{+}}05]{BGHSV05}
Eli Ben-Sasson, Oded Goldreich, Prahladh Harsha, Madhu Sudan, and Salil Vadhan.
\newblock Short pcps verifiable in polylogarithmic time.
\newblock In {\em Computational Complexity, 2005. Proceedings. Twentieth Annual
  IEEE Conference on}, pages 120--134. IEEE, 2005.

\bibitem[BVY17]{bavarian2017hardness}
Mohammad Bavarian, Thomas Vidick, and Henry Yuen.
\newblock Hardness amplification for entangled games via anchoring.
\newblock In {\em Proceedings of the 49th Annual ACM SIGACT Symposium on Theory
  of Computing}, pages 303--316. ACM, 2017.

\bibitem[CGJV17]{coladangelo2017verifier}
Andrea Coladangelo, Alex Grilo, Stacey Jeffery, and Thomas Vidick.
\newblock Verifier-on-a-leash: new schemes for verifiable delegated quantum
  computation, with quasilinear resources.
\newblock 2017,  \href{http://arxiv.org/abs/1708.07359}{{\ttfamily
  arXiv:1708.07359}}.

\bibitem[CM14]{CM13}
Toby Cubitt and Ashley Montanaro.
\newblock Complexity classification of local {H}amiltonian problems.
\newblock In {\em Foundations of Computer Science (FOCS), 2014 IEEE 55th Annual
  Symposium on}, pages 120--129. IEEE, 2014.

\bibitem[CN16]{CN16}
Matthew Coudron and Anand Natarajan.
\newblock The parallel-repeated {Magic} {Square} game is rigid.
\newblock Technical report, 2016,
  \href{http://arxiv.org/abs/1609.06306}{{\ttfamily arXiv:1609.06306}}.

\bibitem[Col16]{Coladangelo16}
Andrea~W. Coladangelo.
\newblock Parallel self-testing of (tilted) {EPR} pairs via copies of (tilted)
  {CHSH}.
\newblock Technical report, 2016,
  \href{http://arxiv.org/abs/1609.03687}{{\ttfamily arXiv:1609.03687}}.

\bibitem[CRSV17]{chao2017test}
Rui Chao, Ben~W. Reichardt, Chris Sutherland, and Thomas Vidick.
\newblock Test for a large amount of entanglement, using few measurements.
\newblock In {\em Proceedings of the 2017 Conference on Innovations in
  Theoretical Computer Science (ITCS)}, 2017.

\bibitem[CS96]{CalderbankShor96}
A.~R. Calderbank and Peter~W. Shor.
\newblock Good quantum error-correcting codes exist.
\newblock {\em Phys. Rev. A}, 54:1098--1105, 1996,
  \href{http://arxiv.org/abs/quant-ph/9512032}{{\ttfamily
  arXiv:quant-ph/9512032}}.

\bibitem[CS17]{ColadangeloS17MS}
Andrea Coladangelo and Jalex Stark.
\newblock Robust self-testing for linear constraint system games.
\newblock 2017,  \href{http://arxiv.org/abs/1709.09267}{{\ttfamily
  arXiv:1709.09267}}.

\bibitem[DFK{\etalchar{+}}11]{DFKRS11}
Irit Dinur, Eldar Fischer, Guy Kindler, Ran Raz, and Shmuel Safra.
\newblock {PCP} characterizations of {NP}: Toward a polynomially-small
  error-probability.
\newblock {\em {C}omputational {C}omplexity}, 20(3):413, 2011.

\bibitem[EH15]{eldar2015local}
Lior Eldar and Aram~W Harrow.
\newblock Local {H}amiltonians whose ground states are hard to approximate.
\newblock 2015,  \href{http://arxiv.org/abs/1510.02082}{{\ttfamily
  arXiv:1510.02082}}.

\bibitem[FH15]{FitzsimonsH15}
Joseph Fitzsimons and Michal Hajdu{\v{s}}ek.
\newblock Post hoc verification of quantum computing.
\newblock Technical report, 2015,
  \href{http://arxiv.org/abs/1512.04375}{{\ttfamily arXiv:1512.04375}}.

\bibitem[FV15]{FV14}
Joseph Fitzsimons and Thomas Vidick.
\newblock A multiprover interactive proof system for the local {H}amiltonian
  problem.
\newblock In {\em Proceedings of the 2015 Conference on Innovations in
  Theoretical Computer Science}, pages 103--112. ACM, 2015.

\bibitem[GH15]{gowers2015inverse}
W.~T. Gowers and O.~Hatami.
\newblock Inverse and stability theorems for approximate representations of
  finite groups.
\newblock 2015,  \href{http://arxiv.org/abs/1510.04085}{{\ttfamily
  arXiv:1510.04085}}.

\bibitem[GKP16]{grilo2016pointer}
Alex~B Grilo, Iordanis Kerenidis, and Attila Pereszl{\'e}nyi.
\newblock Pointer quantum {PCP}s and multi-prover games.
\newblock 2016,  \href{http://arxiv.org/abs/1603.00903}{{\ttfamily
  arXiv:1603.00903}}.

\bibitem[Gle10]{glebsky2010almost}
Lev Glebsky.
\newblock Almost commuting matrices with respect to normalized
  {Hilbert-Schmidt} norm.
\newblock 2010,  \href{http://arxiv.org/abs/1002.3082}{{\ttfamily
  arXiv:1002.3082}}.

\bibitem[Got99]{gottesman1999fault}
Daniel Gottesman.
\newblock Fault-tolerant quantum computation with higher-dimensional systems.
\newblock {\em Chaos Solitons and Fractals}, 10(10):1749, 1999.

\bibitem[IV12]{IV12}
Tsuyoshi Ito and Thomas Vidick.
\newblock A multi-prover interactive proof for {NEXP} sound against entangled
  provers.
\newblock {\em Proc. 53rd FOCS}, pages 243--252, 2012,
  \href{http://arxiv.org/abs/arXiv:1207.0550}{{\ttfamily
  arXiv:arXiv:1207.0550}}.

\bibitem[Ji16a]{ji2015classical}
Zhengfeng Ji.
\newblock Classical verification of quantum proofs.
\newblock In {\em Proceedings of the 48th Annual ACM SIGACT Symposium on Theory
  of Computing}, pages 885--898. ACM, 2016.

\bibitem[Ji16b]{ji16nexp}
Zhengfeng Ji.
\newblock Compression of quantum multi-prover interactive proofs.
\newblock 2016,  \href{http://arxiv.org/abs/1610.03133}{{\ttfamily
  arXiv:1610.03133}}.

\bibitem[KKKS06]{ketkar2006nonbinary}
Avanti Ketkar, Andreas Klappenecker, Santosh Kumar, and Pradeep~Kiran
  Sarvepalli.
\newblock Nonbinary stabilizer codes over finite fields.
\newblock {\em IEEE Transactions on Information Theory}, 52(11):4892--4914,
  2006.

\bibitem[KRR14]{KRR14}
Yael~Tauman Kalai, Ran Raz, and Ron~D. Rothblum.
\newblock How to delegate computations: The power of no-signaling proofs.
\newblock In {\em Proceedings of the Forty-sixth Annual ACM Symposium on Theory
  of Computing}, STOC '14, pages 485--494, New York, NY, USA, 2014. ACM.

\bibitem[LFKN92]{lund1992algebraic}
Carsten Lund, Lance Fortnow, Howard Karloff, and Noam Nisan.
\newblock Algebraic methods for interactive proof systems.
\newblock {\em Journal of the ACM (JACM)}, 39(4):859--868, 1992.

\bibitem[MBG{\etalchar{+}}13]{menezes2013applications}
Alfred~J Menezes, Ian~F Blake, XuHong Gao, Ronald~C Mullin, Scott~A Vanstone,
  and Tomik Yaghoobian.
\newblock {\em Applications of finite fields}, volume 199.
\newblock Springer Science \& Business Media, 2013.

\bibitem[NV17a]{NV17}
Anand Natarajan and Thomas Vidick.
\newblock A quantum linearity test for robustly verifying entanglement.
\newblock In {\em Proceedings of the 49th Annual ACM SIGACT Symposium on Theory
  of Computing}, STOC 2017, pages 1003--1015, New York, NY, USA, 2017. ACM,
  \href{http://arxiv.org/abs/1610.03574}{{\ttfamily arXiv:1610.03574}}.

\bibitem[NV17b]{NatarajanV17twoprover}
Anand Natarajan and Thomas Vidick.
\newblock Two-player entangled games are {NP}-hard.
\newblock 2017,  \href{http://arxiv.org/abs/1710.03062}{{\ttfamily
  arXiv:1710.03062}}.
\newblock To appear in the proceedings of CCC'18.

\bibitem[OV16]{OV16}
Dimiter Ostrev and Thomas Vidick.
\newblock Entanglement of approximate quantum strategies in {XOR} games.
\newblock Technical report, 2016,
  \href{http://arxiv.org/abs/1609.01652}{{\ttfamily arXiv:1609.01652}}.

\bibitem[RRR16]{RRR16}
Omer Reingold, Guy~N. Rothblum, and Ron~D. Rothblum.
\newblock Constant-round interactive proofs for delegating computation.
\newblock In {\em Proceedings of the Forty-eighth Annual ACM Symposium on
  Theory of Computing}, STOC '16, pages 49--62, New York, NY, USA, 2016. ACM.

\bibitem[RS97]{raz1997sub}
Ran Raz and Shmuel Safra.
\newblock A sub-constant error-probability low-degree test, and a sub-constant
  error-probability {PCP} characterization of {NP}.
\newblock In {\em Proceedings of the Twenty-ninth Annual ACM Symposium on
  Theory of Computing}, STOC '97, pages 475--484, New York, NY, USA, 1997. ACM.

\bibitem[Sch80]{schwartz1980fast}
J.~T. Schwartz.
\newblock Fast probabilistic algorithms for verification of polynomial
  identities.
\newblock {\em Journal of the ACM}, 27(4):701--717, 1980.

\bibitem[Ste96]{Steane96}
Andrew Steane.
\newblock Multiple-particle interference and quantum error correction.
\newblock In {\em Proceedings of the Royal Society of London A: Mathematical,
  Physical and Engineering Sciences}, volume 452, pages 2551--2577. The Royal
  Society, 1996,  \href{http://arxiv.org/abs/quant-ph/9601029}{{\ttfamily
  arXiv:quant-ph/9601029}}.

\bibitem[Vid13]{Vidick13xor}
Thomas Vidick.
\newblock Three-player entangled {XOR} games are {NP}-hard to approximate.
\newblock In {\em Proc. 54th FOCS}, 2013,
  \href{http://arxiv.org/abs/1302.1242}{{\ttfamily arXiv:1302.1242}}.

\bibitem[Zip79]{Zippel79}
Richard Zippel.
\newblock Probabilistic algorithms for sparse polynomials.
\newblock In {\em Proceedings of the International Symposiumon on Symbolic and
  Algebraic Computation}, EUROSAM '79, pages 216--226, London, UK, UK, 1979.
  Springer-Verlag.

\end{thebibliography}

\notesendofpaper

\end{document}